\documentclass[aps,prx,superscriptaddress,longbibliography,twocolumn,notitlepage,nofootinbib,10pt]{revtex4-2}
\usepackage{xcolor}
\definecolor{color1}{rgb}{0,0,0.7}
\definecolor{color2}{rgb}{0.85,0,0}

\definecolor{darkblue}{RGB}{0,0,149}

\usepackage[colorlinks,citecolor=darkblue,linkcolor=darkblue,urlcolor=darkblue, breaklinks=true]{hyperref}
\usepackage{graphicx}
\usepackage{amsmath}
\usepackage{amssymb}
\usepackage{amsfonts}
\usepackage{amsthm}
\usepackage{bbm}
\usepackage{hyperref}
\usepackage{soul}
\usepackage{textcomp}
\usepackage{bm}

\usepackage{dsfont}
\usepackage{relsize}
\usepackage{braket}
\usepackage{enumitem}
\usepackage{cleveref}
\usepackage{csquotes}
\usepackage{bigints}
\usepackage{float}
\usepackage{framed}
\usepackage[makeroom]{cancel}
\usepackage{comment}
\usepackage{mathrsfs}

\newtheorem{theorem}{Theorem}[]
\newtheorem{corollary}{Corollary}[]

\def\Tr{{\rm Tr}}
\def\rl{\!\braket{\hat r_E\hat\ell_E}_{\hat\sigma_E}\!}
\newcommand{\rel}[2]{\!\braket{\hat r_E\tilde E_{#1}(#2)\hat\ell_E}_{\hat\sigma_E}\!}

\def\Tr{{\rm Tr}}

\newcommand{\nocontentsline}[3]{}

\newcommand{\fref}[1]{\textcolor{blue}{\hyperref[#1]{Fig.~\bfseries{\ref{#1}}}}}
\newcommand{\sref}[1]{\textcolor{blue}{\hyperref[#1]{Section~\ref{#1}}}}
\newcommand{\tref}[1]{\textcolor{blue}{\hyperref[#1]{Table~\bfseries{\ref{#1}}}}}

\newcommand{\aref}[1]{\textcolor{blue}{\hyperref[#1]{App.$\,$\ref{#1}}}}

\newcounter{example}

\begin{document}

\title{Autonomization of Quantum Systems and the Emergence of the Work Operator}

\author{Carlo Cepollaro}
\email{carlo.cepollaro@univie.ac.at}
\affiliation{Vienna Center for Quantum Science and Technology (VCQ), Faculty of Physics, University of Vienna, Boltzmanngasse 5, 1090 Vienna, Austria}
\affiliation{Institute for Quantum Optics and Quantum Information - IQOQI Vienna, Austrian Academy of Sciences, Boltzmanngasse 3, 1090 Vienna, Austria}
\author{Danial Chughtai}
\affiliation{Vienna Center for Quantum Science and Technology (VCQ), Atominstitut, TU Wien, Stadionallee 2, 1020 Vienna, Austria}
\affiliation{Institute for Quantum Optics and Quantum Information - IQOQI Vienna, Austrian Academy of Sciences, Boltzmanngasse 3, 1090 Vienna, Austria}
\author{Alberto Spalvieri}
\affiliation{Vienna Center for Quantum Science and Technology (VCQ), Faculty of Physics, University of Vienna, Boltzmanngasse 5, 1090 Vienna, Austria}
\affiliation{Institute for Quantum Optics and Quantum Information - IQOQI Vienna, Austrian Academy of Sciences, Boltzmanngasse 3, 1090 Vienna, Austria}
\author{Marcus Huber}
\affiliation{Vienna Center for Quantum Science and Technology (VCQ), Atominstitut, TU Wien, Stadionallee 2, 1020 Vienna, Austria}
\affiliation{Institute for Quantum Optics and Quantum Information - IQOQI Vienna, Austrian Academy of Sciences, Boltzmanngasse 3, 1090 Vienna, Austria}
\author{Alberto Rolandi}
\email{alberto.rolandi@tuwien.ac.at}
\affiliation{Vienna Center for Quantum Science and Technology (VCQ), Atominstitut, TU Wien, Stadionallee 2, 1020 Vienna, Austria}
\affiliation{Institute for Quantum Optics and Quantum Information - IQOQI Vienna, Austrian Academy of Sciences, Boltzmanngasse 3, 1090 Vienna, Austria}

\begin{abstract}
    Estimation of work at the quantum scale remains a central challenge in quantum thermodynamics. Canonical approaches have a fundamental drawback: they assume classical control of the quantum system, as implicit in a time-dependent Hamiltonian. Yet the energetic cost of implementing this control is omitted and can exceed the system’s energy scale by orders of magnitude, calling into question the operational significance of work values. We address this by autonomizing the controlled energy transfer, embedding the driven dynamics into an energy-conserving evolution on a larger quantum system. Work is then unambiguously identified with the energy transferred between the two systems, singling out a unique observable on the driven system: the work operator. We show that the quantum work operator evades a no-go theorem by deriving a quantum fluctuation theorem that recovers the Jarzynski equality in classical scenarios. We incorporate imperfect control and quantify corrections to work statistics. Extending the framework to open quantum systems, we obtain an operatorial first law in which work, heat, and internal-energy changes are represented by distinct operators on the reduced system. Together, these results settle the long-standing debate over whether work is a quantum observable: it is, once the controlling system is included in the description rather than treated as external.
\end{abstract}

\maketitle

\section{Introduction}
One of the methodological cornerstones of thermodynamics is the use of time-dependent Hamiltonians to describe physical systems driven out of equilibrium. For classical systems, and for quantum systems controlled by an effectively classical agent, this framework is well motivated. From a foundational perspective, however, it also exposes one of the central challenges in formulating thermodynamics at the quantum scale~\cite{Goold2016,Vinjanampathy2016,Campisi2011}. Time-dependent Hamiltonians can be accurately approximated by large classical control fields, but generating and steering these fields requires work and inevitably dissipates entropy, often on scales orders of magnitude larger than those of the quantum processes being controlled. This calls into question the operational meaning of \enquote{quantum work} as defined within this framework, and partly explains why so many conflicting accounts of work in the quantum regime coexist.

Examples range from ergotropy~\cite{Allahverdyan2004}, understood as the energy extractable by unitary operations generated by time-dependent Hamiltonians, to the cyclic work output of thermal engines, where the Hamiltonian is explicitly driven in time, and to attempts to generalize stochastic work to the quantum regime, most notably through the two-point measurement (TPM) scheme~\cite{Tasaki2000,Kurchan2001}, which requires projective measurements before and after a time-dependent change of the Hamiltonian. In all these cases, the semiclassical treatment prevents a self-contained accounting of energy and entropy. As a consequence, the true efficiency, cost, and fundamental limitations of thermodynamic protocols at the quantum scale remain inaccessible.

One potential workaround is to consider continuously operating autonomous thermal machines~\cite{Mitchison19, Niedenzu19}, which do not suffer from these issues because their Hamiltonians are inherently time independent. This, however, severely limits the range of accessible operations: it precludes the study of agent-controlled energy changes and cyclic engines driven by externally varied parameters, and fails to capture many quantum-technological experiments. Another route is the autonomization of time-dependent Hamiltonian processes through quantum clocks and carefully designed interaction Hamiltonians~\cite{Unruh1989,Horodecki2013,Freznel2014,Skrzypczyk2014,Malabarba_2015,Colla24,Cavina26}. The simplest implementations, however, quickly run into problems of either Hamiltonian unboundedness or limited tractability. While exact unitary dynamics generated by time-dependent Hamiltonians can, in general, be induced autonomously only by infinite-dimensional ancillas~\cite{Aberg2014,Woods2018, Woods2023}, finite-dimensional approximations can nevertheless be constructed that converge to perfectly unitary dynamics in the infinite-dimensional limit. Fluctuation theorems have recently been derived within this autonomous-work framework for classical systems~\cite{Jarzynski25} and subsequently extended to the quantum regime using projective measurements~\cite{Zhao26}.

In this paper, we follow this conceptual approach and show how the energy-changing unitaries associated with time-dependent Hamiltonian evolutions of arbitrary quantum systems emerge from a fully autonomous description. This agent--free construction singles out an operator associated with the energy change induced on the system, providing a natural notion of a work observable without postulating one at the outset~\cite{Bochkov77,Yukawa00,Chernyak04,Allahverdyan05, Weimer2008,Talkner2016}. We show that the distribution of the fundamental energetic cost required to implement the evolution induced by the time-dependent Hamiltonian $\hat H_S(t)$ between $t=t_0$ and $t=\tau$ is governed by the operator
\begin{equation}\label{eq:work_op_int}
\begin{aligned}
    \hat W(\tau,t_0) &= \hat U^\dagger(\tau,t_0)\hat H_S(\tau)\hat U(\tau,t_0) - \hat H_S(t_0) \\ 
     &= \int_{t_0}^\tau\!dt~\hat U^\dagger(t,t_0)\,\partial_t{\hat H}_S(t)\,\hat U(t,t_0),
\end{aligned}
\end{equation}
evaluated on the initial state of the system, where $\hat U(\tau,t_0)$ describes the evolution between $t_0$ and $\tau$. Our autonomous construction provides a novel foundation for this operator, which, despite earlier criticism, has also recently appeared in several independent approaches to quantum work~\cite{Beyer2020,Silva2021,Silva2024,silva2025consistency,sathe2026,Schmid26,elouard2026}.

Furthermore, by deriving an exact expression for the moments of the induced energy-change distribution, we obtain a remarkable result: Jarzynski's equality~\cite{Jarzynski1997,Crooks1999,Seifert2012,Jarzynski25} is recovered for all classical protocols, while genuinely quantum cases acquire an explicit closed-form correction, expressible as the expectation value of a commutator. This circumvents a well-known no-go theorem for defining fluctuating work in quantum processes~\cite{Perarnau17} and removes the need for two-point measurements in the estimation of work~\cite{micadei2020,Lostaglio2018}.

This has profound implications for four reasons. First, a definition of work that does not require measurement is crucial, because measurement constitutes one of the major challenges for quantum thermodynamics: although its energetic cost is not yet known in full generality, in practice it often lies far beyond the quantum energy scale~\cite{Guryanova20,Schwarzhans25}. Second, it shows that a consistent definition of work for agent-controlled quantum processes is possible, contrary to long-held convictions in the field~\cite{Allahverdyan05,Talkner2007,Perarnau17,Baumer2018}. Third, unlike previous attempts~\cite{ Talkner2007, Perarnau17,  Silva2024}, this definition is not obtained by imposing consistency requirements on work itself, but emerges as a consequence of autonomizing controlled quantum evolutions and proper energy bookkeeping. Finally, this operatorial perspective also provides a natural route to a first law of thermodynamics at the operator level. Once work is identified with autonomous energy transfer rather than with a measurement record, the same framework can be applied to open controlled systems, where one can ask how work, heat, and internal-energy changes are represented on the reduced
system alone.

The rest of the manuscript is organized as follows. Section~\ref{sec:autonomous} introduces the autonomous picture of driven quantum dynamics. Section~\ref{sec:work_operator} derives the work operator from the autonomous energy transfer. Section~\ref{sec:fluctuation_theorem} establishes the corresponding quantum fluctuation theorem. Section~\ref{sec:imprecise_control} examines the effects of imprecise control. Section~\ref{sec:oqs} extends the framework to open quantum systems and uses the work operator to formulate the first law of thermodynamics at the
operatorial level. Section~\ref{sec:discussion} addresses the existing debate in the literature and the shortcomings of arguments against the work operator, and finally Section~\ref{sec:conclusion} concludes.

\section{The Autonomous Picture of a Controlled Quantum System}\label{sec:autonomous}
The evolution of a closed quantum system $S$, defined over a Hilbert space $\mathcal{H}_S$, from a time $t_0$ to time $t$ is governed by unitary dynamics,
\begin{equation}\label{eq:driven_dynamics}
    \hat \rho_S(t) = \hat U(t,t_0)\hat \rho_S(t_0)\hat U^\dagger(t,t_0),
\end{equation}
where $\hat\rho_S(t)$ is the state of the system $S$ at time $t$ and $\hat U(t,t_0)$ is the time-ordered unitary operator
\begin{equation}
    \hat U(t,t_0)=\mathcal{T}\exp\!\left(-i\int_{t_0}^t \!dt'\, \hat H_S(t')\right),
\end{equation}
with $\hat H_S(t)$ the Hamiltonian of the system at time $t$. A time-dependent Hamiltonian allows one to model the action of an external agent on a system’s control parameters, such as the displacement of a piston or the switching of an external field. From a thermodynamic standpoint, this is a powerful framework, as any operation an agent can perform on a system can be interpreted as a change in the Hamiltonian, provided that the description includes all relevant degrees of freedom except those associated with the external agent. However, it is important to note that this description does not correspond to a fully energy-preserving framework, as energy must be supplied by the external agent implementing the changes.

Here, we propose an alternative picture in which the explicit time dependence of the system Hamiltonian is removed, by incorporating the external control into the quantum description of the dynamics. To this end, we introduce a quantum clock that interacts with the system through a time-independent total Hamiltonian. Although the joint dynamics of system and clock is autonomous, the interaction induces an effective time-dependent Hamiltonian on the system alone: for each clock state, the Hamiltonian acting on $S$ is selected through its interaction with the clock. Importantly, this \emph{autonomous picture} exactly recovers the desired time-dependent Hamiltonian unitary dynamics, without ever invoking an external agent.

Concretely, we introduce an auxiliary clock system $C$ with Hilbert space $\mathcal{H}_C=L^2(\mathbb{R})$, spanned by the continuous family of (generalized) vectors $\{\ket{t}\}_{t\in\mathbb R}$ satisfying $\braket{t|t'}=\delta(t-t')$. We now consider the evolution of the composite system on $\mathcal{H}_S\otimes\mathcal{H}_C$, governed by the total Hamiltonian
\begin{equation}
    \hat H_{\text{tot}}=\mathbbm 1\otimes \hat H_C+\hat H_{\text{int}},
\end{equation}
where $\hat H_{C}$ is the clock Hamiltonian and $\hat H_{\text{int}}$ couples the clock to the system
\begin{equation}\label{eq:int_ham}
    \hat H_{\text{int}}=\int_{\mathbb{R}}\! dt'\, \hat H_S(t')\otimes \ket{t'}\!\bra{t'},
\end{equation}
The clock Hamiltonian is taken to be\footnote{This construction is idealized, as it requires the clock Hamiltonian $\hat H_C$ to be unbounded from below. We adopt it here for simplicity; two analogous constructions based on bounded Hamiltonians are given in Appendix~\ref{app:autonomous_physical}.}
\begin{equation}\label{eq:clock_ham}
    \hat H_C = \frac{1}{2\pi}\int_{\mathbb R}\!dE\,dt\,dt'E \, e^{iE(t-t')}\ket{t\,}\!\bra{t'},
\end{equation}
which generates translations in the clock basis,
\begin{equation}
    e^{-i\hat H_C s}\ket{t}=\ket{t+s}.
\end{equation}
The evolution of the composite system is then
\begin{equation}
    \hat\rho_{tot}(\Delta t) = e^{-i\hat H_{tot}\Delta t}\,\hat\rho_{tot}(0) \, e^{i\hat H_{tot}\Delta t}.
\end{equation}
For an arbitrary $\hat\rho_{tot}(0)$ this leads to a highly non-trivial evolution. However, when there is no initial correlation between system and clock we can recover the driven dynamics on the system $S$. We formulate this insight as 
\begin{theorem}[Autonomization of Quantum Systems]\label{th:autonomization}
    Given an initial product state $\hat\rho_{tot}(0) = \hat \rho_S \otimes \hat \rho_C $, the reduced dynamics on the system take the form of a mixture of unitary evolutions weighted by the initial clock distribution
    \begin{equation}\label{eq:autonomous_evolution}
    \Tr_C[\hat\rho_{tot}(\Delta t)]\! =\! \int_{\mathbb{R}}dt'\,p_C(t')\,\hat U(t'+\Delta t,t')\,\hat\rho_S\,\hat U^\dagger(t'+\Delta t,t'),
    \end{equation}
    where $p_C(t) = \braket{t|\hat\rho_C|t}$.

    In particular, for the initial clock state $\hat\rho_C= \ket{t_0}\!\bra{t_0}$, the reduced dynamics coincide with the original non-autonomous evolution from $t_0$ to $t = t_0 + \Delta t$:
    \begin{equation}\label{eq:autonomous_vs_nonautonomous}
        \hat \rho_S(t) = \Tr_C[\hat\rho_{tot}(\Delta t)] = \hat U(t,t_0) \, \hat \rho_S \, \hat U^\dagger(t,t_0).
    \end{equation}
\end{theorem}
\begin{proof}
    The proof consists of applying the Lie--Trotter formula and exploiting the specific form of $\hat H_{\text{int}}$ to simplify the evolution. Full details are provided in Appendix~\ref{app:autonomous}.
\end{proof}

\begin{figure*}[t]
    \centering
    \begin{minipage}[c]{0.35\textwidth}
        \centering
    \includegraphics[width=\linewidth]{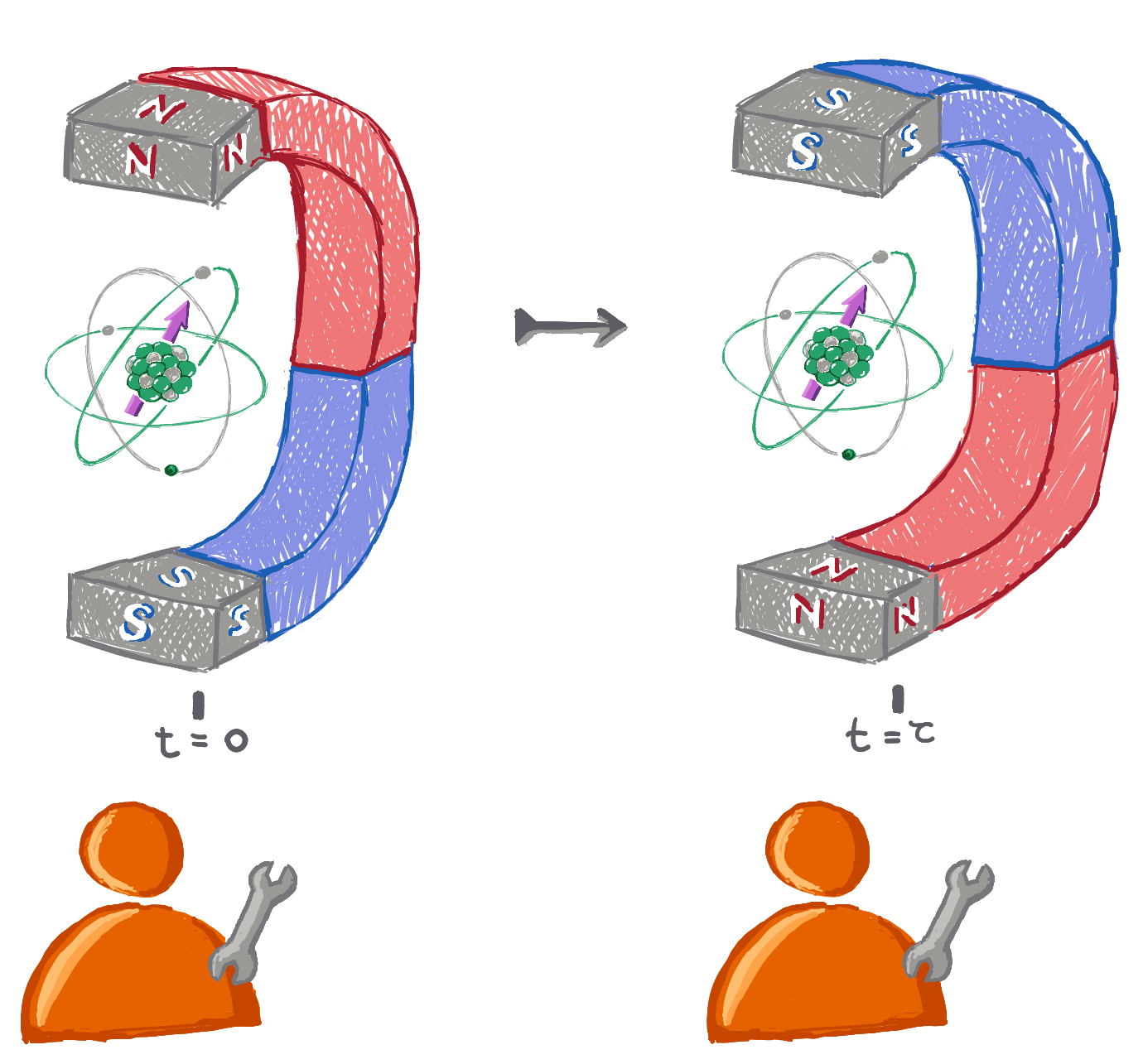}
    \end{minipage}
    \hspace{0.1\textwidth}
    \begin{minipage}[c]{0.35\textwidth}
        \centering        \includegraphics[width=\linewidth]{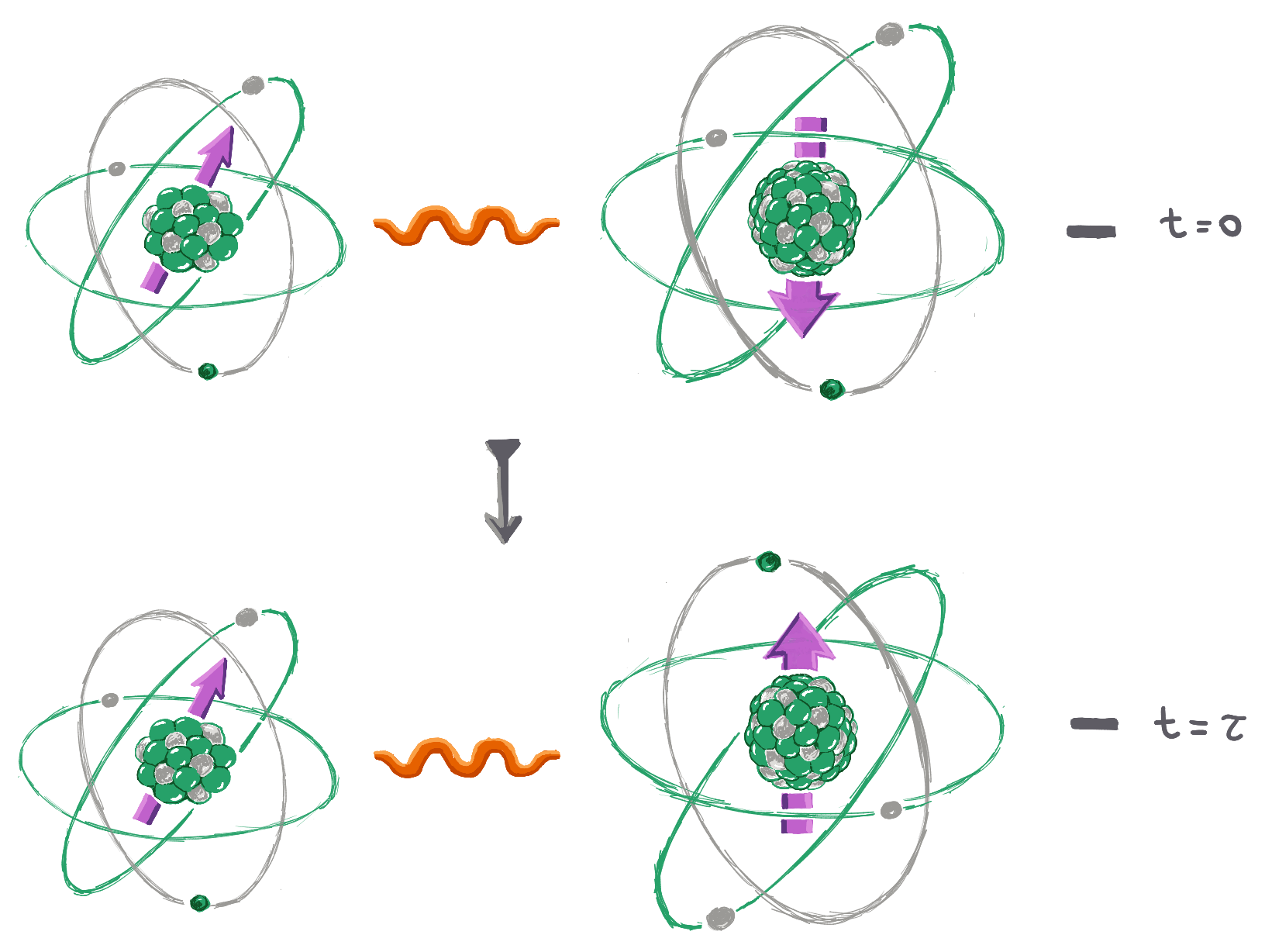}
    \end{minipage}
    \caption{\label{fig:illustration}
    Illustration of the autonomization procedure. Left: In the usual non-autonomous description, a qubit is driven by a time-dependent Hamiltonian $\hat H_S(t)$, represented here by an externally controlled magnetic field that is changed by a classical agent. The energetic cost and fluctuations of this control are not included in the model. Right: In the autonomous description, the control device is promoted to a quantum system. The qubit interacts with a second system, here depicted as a rotating qudit, whose state selects the effective Hamiltonian acting on the system. The joint evolution is generated by a time-independent Hamiltonian, while the reduced dynamics of the qubit reproduces the original driven evolution.}
\end{figure*}

The theorem’s main statement shows that any dynamics generated by a time-dependent Hamiltonian admits an exact autonomous realization: it can be reproduced by an energy-conserving, time-independent dynamics on a larger system; see the illustration in Figure~\ref{fig:illustration}. It is worth emphasizing that this is stronger than the standard Stinespring dilation of quantum channels~\cite{stinespring}, which states that every CPTP map can be realized by coupling the system to an ancilla, evolving the composite system unitarily, and then discarding part of it. Indeed, for each time pair $(t,t')$, the unitary $\hat U(t,t')$ can be realized as the reduced action of some unitary $\hat V_{t,t'}$ on an enlarged system. However, this does not guarantee that the unitaries corresponding to different times arise from a single time-independent Hamiltonian. The nontrivial content of the present construction is precisely that the entire family $\{\hat U(t,t')\}_{t,t'}$ is recovered by evolving with $\hat H_{tot}$ for a duration $t-t'$ and choosing the initial clock state to be $\hat\rho_C = \ket{t'}\!\bra{t'}$. 

Moreover, beyond reproducing the exact driven dynamics, this construction provides a natural way to model imprecise time control: one simply chooses a different initial state for the clock; see Eq.~\eqref{eq:autonomous_evolution}. In the remainder of this section and throughout the next, we restrict attention to the exact driven dynamics, corresponding to the initial clock state $\hat \rho_C=\ket{t_0}\!\bra{t_0}$. We return to the general case in Section~\ref{sec:imprecise_control}.

Having established the equivalence of the two pictures at the level of states, we now turn to the correspondence between observables.
\begin{corollary}\label{th:observables}
    Consider an observable on the composite system of the form
    \begin{equation}
        \hat O_{\text aut}
        =
        \int_{\mathbb{R}} dt^\prime \, \hat O_{t^\prime}\otimes \ket{t^\prime}\!\bra{t^\prime}
    \end{equation}
    and the clock initial state to be pure $\hat \rho_C=\ket{t_0}\!\bra{t_0}$. Then, for every time $t=t_0+\Delta t$, $\hat O_t$ is the unique time-dependent observable on the system whose moments reproduce those of $\hat O_\text{aut}$ on the composite system:
    \begin{equation}\label{eq:observable_statistics}
        \Tr_{SC}\!\!\left[
            \hat O_\text{aut}^n\, e^{-i \hat H_\text{tot} \Delta t}
            (\hat \rho_S \otimes \ket{t_0}\!\bra{t_0}\bigr)
            e^{i \hat H_\text{tot} \Delta t}
        \right]
        =
        \Tr_S\!\left[
            \hat O_{t}^n\, \hat \rho_S(t)
        \right].
    \end{equation}
\end{corollary}
\begin{proof}
    It can be proved by noting that $\hat O^n_\text{aut} = \int_{\mathbb{R}} dt^\prime \,\hat O^n_{t^\prime} \otimes \ket{t^\prime}\!\bra{t^\prime}$ and applying the same reasoning of the proof of Theorem~\ref{th:autonomization}, detailed in Appendix~\ref{app:autonomous}. Uniqueness follows by contradiction, using the fact that the equality must hold for every $\hat \rho_S$.
\end{proof}
This observation allows us to relate observables on the autonomous composite system to observables on the system alone. As we will see in the next section, this leads to a natural definition of a work operator for the system. 

\section{Work Operator}\label{sec:work_operator}
Characterizing work in a closed but non-isolated system is subtle, because the energy cost of implementing the classical control encoded in the time-dependent Hamiltonian is not explicitly accounted for. In the previous section, we reformulated the time-dependent scenario entirely in terms of a closed, isolated autonomous system. In this energy-conserving setting, work is unambiguously defined: it is the energy lost by the clock and transferred to the system. At the quantum level, this transfer is captured, as an operator, by the change in the clock Hamiltonian in the Heisenberg picture (see Section~\ref{sec:discussion} for an extensive discussion). We thus identify the autonomous work operator
\begin{equation}\label{eq:work_autonomous}
    \hat W_{\text{aut}}(\Delta t)
    =
    -\left(
    e^{i\hat H_{\text{tot}}\Delta t} \,\hat H_C \,e^{-i\hat H_{\text{tot}}\Delta t}
    -
    \hat H_C
    \right).
\end{equation}
Beyond the total energy transfer encoded by
$\hat W_{\mathrm{aut}}(\Delta t)$, the autonomous description also identifies its
instantaneous flow. The energy current from the clock to the system is, in the Schrödinger picture,
\begin{equation}\label{eq:aut_current}
    \hat w_{\mathrm{aut}}
    :=
    -i[\hat H_{\mathrm{tot}},\hat H_C].
\end{equation}

In the Heisenberg picture, currents at different times are all represented as observables referred to the same initial state and can therefore be consistently integrated to obtain the work operator \mbox{$\hat W_{\mathrm{aut}}(\Delta t) =\int_0^{\Delta t}\!ds\,e^{i\hat H_{\mathrm{tot}}s} \hat w_{\mathrm{aut}} e^{-i\hat H_{\mathrm{tot}}s}$}.

This representation also clarifies the role of $\Delta t$: although $\hat W_{\mathrm{aut}}(\Delta t)$ is constructed by integrating a Heisenberg-picture current, it does not represent work measured at time $\Delta t$. Rather, it is an observable evaluated on the initial state, with $\Delta t$ parameterizing the duration of the process over which the energy current is accumulated. Thus, $\hat W_{\mathrm{aut}}(\Delta t)$ answers the question: ``How much work would be performed, or equivalently how much energy would be transferred to the system, if the autonomous dynamics were implemented from time $t_0$ to time $t=t_0+\Delta t$?'' Owing to the operatorial nature of energy, this transfer is encoded not only in changes of energy populations, but also in coherences. One could alternatively define a \emph{retroactive} work operator, with the unitary operators placed on the opposite sides, corresponding to the energy transfer inferred at the final time, after the process has occurred. However, since the definition of work does not involve an intermediate measurement, these two descriptions are physically equivalent.

Our autonomization construction shows that there is a unique observable on the system that reproduces the same statistics as this work operator. Owing to its importance, we state this result as a theorem.
\begin{theorem}[Work Operator]\label{th:work}
    The work operator
    \begin{equation}\label{eq:work_operator_system}
        \hat W(t,t_0) = \hat U^\dagger(t,t_0)\,\hat H_S(t)\,\hat U(t,t_0) - \hat H_S(t_0)
    \end{equation}
    is the unique operator on $S$ that has the same statistics on the initial state of the system as the autonomous work operator $\hat W_\text{aut}(\Delta t)$ on the composite system, when $\hat \rho_C=\ket{t_0}\!\bra{t_0}$, i.e.
    \begin{equation}
        \Tr_{SC}\!\left[\hat W^n_\text{aut}(\Delta t)(\hat \rho_S \otimes \ket{t_0}\!\bra{t_0})\right] = \Tr_S\!\left[\hat W^n(t,t_0) \,\hat \rho_S\right]\!.
    \end{equation}
\end{theorem}
\begin{proof}
    The theorem can be proved by noting that $\hat W_\text{aut}(\Delta t) = e^{i\hat H_\text{tot}\Delta t} \hat H_\text{int} e^{-i\hat H_\text{tot}\Delta t} - \hat H_\text{int}$ due to energy conservation, and by applying Corollary \ref{th:observables} to $\hat H_\text{int}$ at the initial time $t_0$ and final time $t = t_0+\Delta t$.
\end{proof}

Likewise, the correspondence extends to energy currents: the autonomous current of the composite system, defined in Eq.~\eqref{eq:aut_current}, corresponds uniquely to the non-autonomous work-current observable in the Schrödinger picture
\begin{equation}
    \hat w(t) = \partial_t \hat H_S(t).
\end{equation}

Theorem~\ref{th:work} shows that the work operator in Eq.~\eqref{eq:work_operator_system} follows directly from the autonomous, energy-conserving formulation, rather than being postulated for the driven system. In this formulation, work is identified with the energy exchanged between the clock and the system. The construction is therefore self-consistent and principled: the corresponding observable is uniquely determined by the requirement that it characterizes the energy transferred to the system. In this way, the definition avoids the ambiguities and ad hoc choices that arise in other approaches to quantum work. Moreover, it leads to genuine predictions that can differ from those obtained with other characterizations of work, such as the TPM scheme~\cite{Tasaki2000, Kurchan2001} and quasi-probability-based approaches~\cite{Lostaglio2018}. To see this, consider the following example.

\par\medskip
\noindent\textbf{Example 1.} Consider a qubit in the initial state $\ket{\psi}=\ket{0}$, undergoing evolution under the instantaneous quench
\begin{equation} \label{eq:example_hamiltonian}
  \hat H_S(t)=
  \begin{cases}
    \hat \sigma_x, & t=t_0,\\[2mm]
    \hat \sigma_z+\hat \sigma_x, & t>t_0.
  \end{cases}
\end{equation}
We can calculate the work to be
\begin{equation}
  \hat W(t,t_0)=
  \begin{cases}
    0, & t=t_0,\\[2mm]
    \hat \sigma_z, & t>t_0.
  \end{cases}
\end{equation}
Since the initial state is an eigenstate of the work operator, it follows that, for any $t>0$, the work performed on the system is deterministically $w=1$. By contrast, both the TPM and the QP schemes assign a nonzero probability to outcomes with $w>1$; see Appendix~\ref{app:example} for details. This difference has direct operational consequences. Suppose, for instance, that the energy supply is prepared so as to provide at most one unit of work, then the process will fail whenever more energy is required. The work operator then predicts that the process always succeeds, since exactly one unit of work is supplied. On the contrary, the TPM scheme or quasi-probability-based approaches assign instead a nonzero probability to failure. The origin of this discrepancy is that these quantities have different physical meanings: TPM and quasi-probability approaches are tied to measurements performed on the system, which modify the process and, in this example, affect the very resources that enable it. These resources are coherences in energy, which are necessary for the process but are destroyed by the TPM measurement and altered by the quasi-probability construction. The work operator derived here, by contrast, quantifies the energy transfer without introducing such measurement back-action. Hence the work operator is not only the unique operator compatible with energy transfer in the autonomous, energy-conserving setting, but also yields physical predictions that are not captured by measurement-based approaches to quantum work.
\section{Quantum Fluctuation Theorem}\label{sec:fluctuation_theorem}
We have shown that the work operator arises naturally from energy transfer in the autonomous, energy-conserving description. We now turn to the corresponding fluctuation relation. In particular, we show that this framework admits a quantum fluctuation theorem that generalizes the Jarzynski equality~\cite{Jarzynski1997} to coherent work processes, and that the resulting expression takes a remarkably simple form.
\begin{theorem}[Quantum Fluctuation Theorem]\label{th:fluctuation_theorem}
    For processes where the state is initialized in a thermal state
    $\hat\rho_S(t_0)=\hat\pi_S(t_0)=\frac{1}{Z_{t_0}}e^{-\beta\hat H_S(t_0)}$
    at inverse temperature $\beta=(k_BT)^{-1}$, with
    $Z_t=\Tr[e^{-\beta\hat H_S(t)}]$, the work statistics satisfy
    \begin{equation}\label{eq:fluctuation_theorem}
        \braket{e^{-\beta\hat W(t,t_0)}}_{\hat\pi_S(t_0)}
        =
        e^{-\beta\Delta F}
        +
        \Delta_{t,t_0},
    \end{equation}
    where $\Delta F=k_BT\ln(Z_{t_0}/Z_t)$ is the change in equilibrium free energy, and
    \begin{equation}\label{eq:fluctuation_correction}
        \Delta_{t,t_0}
        \!=
        \left\langle
            e^{-\beta\hat W(t,t_0)}
            \!-
            \!e^{-\beta\hat U^\dagger(t,t_0)\hat H_S(t)\hat U(t,t_0)}
            e^{\beta\hat H_S(t_0)}
        \right\rangle_{\hat\pi_S(t_0)} \!.
    \end{equation}
    In particular, for all $t$ and $t_0$, $\Delta_{t,t_0}\geq0$, and
    \begin{equation}\label{eq:classicality}
        \Delta_{t,t_0}=0
        \quad\Longleftrightarrow\quad
        \left[
            \hat U^\dagger(t,t_0)\hat H_S(t)\hat U(t,t_0),
            \hat H_S(t_0)
        \right]
        =0.
    \end{equation}
    For explicit evaluation, the correction can equivalently be written as
    \begin{align}\label{eq:fluctuation_correction_explicit}
        \Delta_{t,t_0}
        &=
        \frac{\beta^2}{Z_{t_0}}
        \int_0^1\!d\lambda
        \int_0^1\!dx\,x\,
        \Tr\!\left[
            \hat\Xi_{\lambda,x}
            e^{-\lambda\beta\hat H_S(t_0)}
        \right],
        \\
        \hat\Xi_{\lambda,x}
        &=\
        e^{-x\beta\hat K_\lambda}
        \!\left[
            \hat U^\dagger(t,t_0)\hat H_S(t)\hat U(t,t_0),
            \hat H_S(t_0)
        \right]\!
        e^{-(1-x)\beta\hat K_\lambda},
        \nonumber
    \end{align}
    with $\hat K_\lambda
    =\hat U^\dagger(t,t_0)\hat H_S(t)\hat U(t,t_0)
    -\lambda\hat H_S(t_0)$.
\end{theorem}
\begin{proof}
    The proof is given in Appendix~\ref{app:fluctuation_theorem}. It follows from the Duhamel identity together with straightforward algebraic manipulations.
\end{proof}

The significance of Theorem~\ref{th:fluctuation_theorem} is twofold. First, it shows that the work operator derived from autonomous energy transfer satisfies a genuine quantum fluctuation theorem. The departure from the classical Jarzynski equality is entirely captured by the non-commutative correction $\Delta_{t,t_0}$, which leaves the average work unchanged but contributes to the exponential work average, and hence to the higher moments of the work statistics. Notably, whereas the equilibrium contribution $e^{-\beta\Delta F}$ depends only on the initial and final Hamiltonians, the quantum correction is generally protocol-dependent through the propagator. Thus, unlike the classical Jarzynski equality, the quantum fluctuation relation is not, in general, determined by equilibrium endpoint data alone, as an effect of non-commutativity.

Second, the theorem clarifies what the appropriate classical limit is. The no-go theorem of Ref.~\cite{Perarnau17} takes classicality to mean agreement with the TPM statistics whenever the initial state is diagonal in the eigenbasis of $\hat H(t_0)$. From the present perspective, however, this condition is too weak to identify a genuinely classical process. Even if the initial state contains no energy coherences, coherent effects can still be generated whenever the Hamiltonians at different times do not commute. Thus, agreement with TPM on energy-diagonal initial states does not by itself characterize the absence of quantum effects in the dynamics (see Ref.~\cite{Silva2024} for similar arguments). 

In contrast, this formulation shows that the quantum effects are quantified by the non-commutative correction $\Delta_{t,t_0}$: when this term vanishes, the quantum fluctuation theorem reduces to the Jarzynski equality; when it does not, the deviation is not a failure of the work operator, but a signature of genuinely quantum dynamics. In particular, in the classical regime, where the Hamiltonians commute throughout the protocol, our relation reduces to the standard Jarzynski equality. Outside this regime, one should not expect the classical relation to hold. This explains why our construction evades the no-go theorem: we do not impose TPM agreement as the definition of classicality, but instead derive the classical limit from the absence of non-commutative contributions to energy transfer. The autonomous framework therefore singles out the physically relevant condition for classicality, rather than building it in through agreement with TPM statistics.

Finally, the connection between the work operator and TPM can be characterized more precisely. For an arbitrary initial state, the first two TPM moments satisfy
\begin{equation}\label{eq:TPM_first_two_moments}
    \braket{W^n}_{\mathrm{TPM}}
    =
    \Tr\!\left[
        \hat W^n(t,t_0)
        \mathcal D_{t_0}[\hat\rho]
    \right],
    \qquad n=1,2,
\end{equation}
where
$\mathcal D_{t_0}[\hat\rho]
:=\sum_n\hat\Pi_n(t_0)\hat\rho\,\hat\Pi_n(t_0)$
denotes dephasing in the eigenspaces of the initial Hamiltonian. Consequently, for energy-diagonal initial states, the work operator and TPM yield the same mean and variance, while their higher moments generally differ because of the non-commutative operator ordering absent from TPM energy differences. It is precisely these higher-order differences that generate the correction $\Delta_{t,t_0}$ in the fluctuation theorem. Furthermore, the vanishing of the correction implies more than the recovery of the Jarzynski equality. By Eq.~\eqref{eq:classicality}, $\Delta_{t,t_0}=0$ implies that the work operator commutes with the initial Hamiltonian. Consequently, the work-operator and TPM statistics coincide not only for an initial thermal state, as we see in Theorem~\ref{th:fluctuation_theorem}, but for any initial state:
\begin{equation}\label{eq:full_TPM_agreement}
    \Delta_{t,t_0}=0
    \quad\Longrightarrow\quad
    \braket{W^n}_{\mathrm{TPM}}
    =
    \Tr\!\left[
        \hat\rho\,\hat W^n(t,t_0)
    \right], \quad \forall \hat \rho.
\end{equation}

\section{Imprecise Control}\label{sec:imprecise_control}
The autonomous formulation provides an energy-conserving description in which all energetic contributions are explicitly accounted for. Work can therefore be identified with energy exchanged between quantum systems, without appealing to an externally prescribed classical drive whose energetic cost lies outside the description. One remaining idealization concerns the preparation of the clock. The exact reproduction of the driven dynamics relies on preparing an auxiliary system in the pure state $\hat \rho_C=\ket{t_0}\!\bra{t_0}$. Exact pure-state preparation is an idealization: in realistic settings, it typically requires ideal cooling, perfect isolation, or otherwise unphysical resources~\cite{Misra2016,Masanes2017}. Consequently, we consider here the more realistic mixed-state case, which captures finite-temperature, imperfectly isolated, and experimentally accessible preparations. This allows us to study how imperfect control affects the resulting work operator and the associated energetic description, in a fully self-consistent way.

As follows from Theorem~\ref{th:autonomization}, only the diagonal components of the clock state in the time basis affect the reduced system dynamics; coherences between different time states do not contribute. Consider, for example, an intended preparation of the clock in the state $\ket{t_0}$ subject to a small timing uncertainty described by a narrow probability distribution $G_\varepsilon$ centered at zero, with variance $\varepsilon^2$ and well-behaved higher moments. The resulting clock state is
\begin{equation}
    \hat \rho_C=\int ds \,G_\varepsilon(s) \,\ket{t_0+s} \!\bra{t_0+s}.
\end{equation}
The autonomous work operator of Eq.~\eqref{eq:work_autonomous} is independent of the initial clock state, so the notion of energy exchange between the system and the clock remains unchanged. What depends on the clock preparation are the effective work statistics induced on the system. Consider a protocol of duration $\Delta t$ intended to start at $t_0$ and end at $t=t_0 + \Delta t$. For the mixed clock state above, these statistics are a classical mixture of those associated with protocols starting at $t_0+s$ and ending at $t+s$. In general, this mixture cannot be represented by a single system observable. Instead, each moment is reproduced by a distinct moment operator,
\begin{equation}\label{eq:imprecise_work}
    \hat W_{G_\varepsilon}^{(n)}(t,t_0)
    =
    \int ds\,G_\varepsilon(s)\,
    \hat W^n(t+s,t_0+s),
\end{equation}
where $\hat W$ is defined in Eq.~\eqref{eq:work_op_int}. These operators do not, in general, arise as powers of a single effective work observable, since
$\hat W_{G_\varepsilon}^{(n)}\neq
\left(\hat W_{G_\varepsilon}^{(1)}\right)^n$.

This expression shows that finite control precision can be incorporated in a transparent way: while the autonomous energy-exchange observable remains unchanged, the work statistics induced on the system become coarse grained and are described by the hierarchy of moment operators. This coarse graining also propagates to the fluctuation relation, which acquires an additional correction quantifying the energetic effect of imperfect timing control. To evaluate this correction, we consider the small-variance limit, corresponding to a small imprecision in the preparation of the clock state. We can then expand around the intended initial time $t_0$, evaluating the zeroth-order fluctuation relation with respect to the intended initial state $\hat\pi_S(t_0)$ and the corresponding free-energy difference. We obtain
\begin{theorem}[Imprecise control]
\label{th:imprecise_fluctuation_theorem}
Given the hierarchy of moment operators in Eq.~\eqref{eq:imprecise_work}, the finite-control quantum fluctuation theorem reads
\begin{multline}
    \sum_{n=0}^{\infty}
        \frac{(-\beta)^n}{n!}\,
        \braket{\hat W_{G_\varepsilon}^{(n)}(t,t_0)}_{\hat\pi_S(t_0)}
    =\\
    =e^{-\beta\Delta F}
    +\Delta_{t,t_0}
    +\varepsilon^2 I_{t,t_0}
    +o(\varepsilon^2).
\end{multline}
where
\begin{equation}
    I_{t,t_0}
    =
    \frac{1}{2}
    \frac{d^2}{ds^2}
    \braket{
        e^{-\beta\hat W(t+s,t_0+s)}
    }_{\hat\pi_S(t_0)}
    \bigg|_{s=0},
\end{equation}
and analogous definitions to Theorem~\ref{th:fluctuation_theorem}.
\end{theorem}
\begin{proof}
    The proof follows by expanding the averaged fluctuation quantity around the intended initial time. The vanishing mean of $G_\varepsilon$ eliminates the linear contribution, while its variance $\varepsilon^2$ determines the leading correction. The zeroth-order term is given by Theorem~\ref{th:fluctuation_theorem}. Further details are provided in Appendix~\ref{app:precise_control_theorem}.
\end{proof}
The correction $I_{t,t_0}$ quantifies how the work fluctuations are modified by finite control of the auxiliary system. Unlike the term $\Delta_{t,t_0}$ in Theorem~\ref{th:fluctuation_theorem}, which is purely non-commutative and vanishes in the classical regime, $I_{t,t_0}$ contains both a quantum and a classical contribution. In particular, it need not vanish even when the Hamiltonians commute throughout the protocol: finite control already modifies the work statistics at the classical level.

In Appendix~\ref{app:precise_control_example}, we analyze an exactly solvable example in which both the intrinsic quantum correction and the finite-control correction can be evaluated explicitly. We then compare their magnitudes in the near-classical, weak-imprecision regime at fixed characteristic work scale and identify when timing uncertainty dominates over quantum effects. This illustrates how imperfect control and noncommutativity contribute separately to the fluctuation relation.

Note furthermore that imprecise control need not increase the total work variance and may instead decrease it. Denoting by $\mu(s)$ and $\sigma^2(s)$ the mean and variance of the shifted protocol, timing uncertainty adds fluctuations by mixing protocols with different mean work, while it may simultaneously reduce the intrinsic variance by averaging over nearby protocols. To leading order in the timing imprecision, the total work variance decreases when \mbox{$\bigl[\partial_s^2\sigma^2(s)+2\bigl(\partial_s\mu(s)\bigr)^2\bigr]_{s=0}<0$}.

More generally, the framework developed in this section applies to other forms of finite control. For instance, one could consider a situation in which the agent aims to implement the protocol $\hat H_S(t)$, but the Hamiltonian realized in each run is $\hat H_S(t)+\delta\hat H_\xi(t)$, where $\xi$ labels an uncontrolled error distributed according to a probability density $p(\xi)$. Each realization $\xi$ defines its own work operator, and the lack of control over $\xi$ induces a classical coarse graining of the corresponding work statistics, analogous to that generated above by the imprecise clock preparation. These statistics are described by a hierarchy of moment operators rather than by a single effective work observable and lead to a fluctuation relation with noise corrections analogous to those in Theorem~\ref{th:imprecise_fluctuation_theorem}.

A qualitatively different form of non-ideality arises when the work source is itself finite and dynamical. Unlike the imprecise control considered above, this cannot generally be represented as classical uncertainty over a family of prescribed protocols: the system and work source become dynamically correlated, and back-action on the source can prevent it from reproducing the intended system dynamics exactly. Autonomous fluctuation relations that explicitly account for back-action on work sources have recently been derived in the classical setting~\cite{Jarzynski25} and subsequently extended to the quantum regime using successive projective measurements~\cite{Zhao26}. Extending our measurement-independent fluctuation relation to finite, back-reacting quantum work sources is an interesting direction for future work.

\section{Open Quantum Systems and First Law} \label{sec:oqs}
We now extend the construction to a driven subsystem of interest $A$ interacting with an environment $E$. To this end, the system $S$ considered in the previous sections is decomposed as $\mathcal H_S=\mathcal H_A\otimes\mathcal H_E$, where $A$ is accessible to the agent and $E$ is inaccessible, with Hamiltonian
\begin{equation}
    \hat H_S(t)
    =
    \hat H_A(t)\otimes\mathbbm 1_E
    +
    \mathbbm 1_A \otimes \hat H_E
    +
    \hat V(t).
\end{equation}
The explicit time dependence specifies the agent's control: the agent drives $\hat H_A(t)$ and may also modulate the interaction $\hat V(t)$, whereas $\hat H_E$ remains fixed and uncontrolled.

As in the closed-system scenario, work is the energy supplied to the composite system by the agent through the controlled time dependence of its Hamiltonian. It is therefore defined at the level of the composite system $S$, independently of its partition into $A$ and $E$. Starting from this global definition allows us to avoid a notorious subtlety encountered when extending thermodynamics to open quantum systems: when an interacting composite is divided into a subsystem of interest and an environment, the split of the total energy into local and interaction contributions is not unique, and with it the very definitions of work, heat, and internal energy become ambiguous~\cite{Talkner2022,Colla22,Seegebrecht2024, Schrauwen2025}. In our construction, the system–environment partition does not redefine work; rather, energy conservation decomposes the already-defined energy transfer according to where it is stored. Thus, which degrees of freedom the agent can access affects the description of the energy transfer, but not the energetic cost of the underlying operation. Following the canonical definitions, we identify the heat $\hat Q$ with the energy transferred to the inaccessible environment $E$, $\Delta\hat E$ with the change in the internal energy of the system of interest $A$, and $\Delta\hat V$ with the change in the interaction energy between $A$ and $E$. This gives the exact operator identity
\begin{equation}\label{eq:work_decomposition}
    \hat W(t,t_0)
    =
    \Delta\hat E(t,t_0)
    +
    \hat Q(t,t_0)
    +
    \Delta\hat V(t,t_0),
\end{equation}
where
\begin{equation}\label{eq:energy_operators}
\begin{aligned}
    \Delta\hat E(t,t_0)
    &=
    \hat U^\dagger(t,t_0)\,
    \hat H_A(t)\,
    \hat U(t,t_0)
    -
    \hat H_A(t_0),
    \\
    \hat Q(t,t_0)
    &=
    \hat U^\dagger(t,t_0)\,
    \hat H_E\,
    \hat U(t,t_0)
    -
    \hat H_E,
    \\
    \Delta\hat V(t,t_0)
    &=
    \hat U^\dagger(t,t_0)\,
    \hat V(t)\,
    \hat U(t,t_0)
    -
    \hat V(t_0).   
\end{aligned}
\end{equation}

The system-environment interaction contribution is not specific to quantum mechanics, but arises whenever an interacting composite is divided into subsystems. If the interaction-energy change vanishes, for instance because the interaction is switched off at the protocol boundaries, or if it is negligible compared with the other relevant energy scales, Eq.~\eqref{eq:work_decomposition} reduces to the operatorial first law, that we formulate as 
\begin{theorem}[First Law of Thermodynamics]
For protocols where the system and its environment are disconnected at the beginning and end of the protocol $(\hat V(t)=\hat V(t_0) = 0)$
\begin{equation}\label{eq:first_law}
    \hat W(t,t_0)=\Delta\hat E(t,t_0)+\hat Q(t,t_0).
\end{equation}
Furthermore, for a weak interaction, Eq.~\eqref{eq:first_law} is recovered at leading order $\mathcal O(\|\hat V\|^0)$.
\end{theorem}
\begin{proof}
    The first law follows trivially from Eq.~\eqref{eq:work_decomposition}.
\end{proof}
While operatorial formulations of the first law have previously been considered in specific contexts~\cite{Quan2005,Koide2016, luo2026}, our construction provides a general derivation as a direct consequence of the autonomization of arbitrary time-dependent Hamiltonian processes. Note that this simplification need not hold for the corresponding instantaneous currents, since a negligible net change in interaction energy does not imply a negligible interaction-energy current.

Having established this global formulation on $S$, we now ask whether the statistics of work can be reproduced within a reduced description involving operators on $A$ alone. We assume an initially factorized state $\hat\rho_S(t_0)=\hat\rho_A(t_0)\otimes\hat\sigma_E$ and denote the reduced dynamics by $\hat\rho_A(t)=\Phi_t[\hat\rho_A(t_0)]$. Such a description cannot, in general, be provided by a single work observable $\hat W_A(t,t_0)$ whose powers reproduce all moments of the global work observable. Instead, each moment is represented by a distinct operator,
\begin{equation}\label{eq:work_reduction_moment}
    \hat W_A^{(n)}(t,t_0)
    :=
    \Tr_E\!\left[
        (\mathbbm 1_A\otimes\hat\sigma_E)
        \hat W^n(t,t_0)
    \right],
\end{equation}
satisfying
\begin{equation}\label{eq:work_reduction_equality}
    \Tr_A\!\left[
        \hat\rho_A(t_0)\hat W_A^{(n)}(t,t_0)
    \right]
    =
    \Tr_S\!\left[
        \hat\rho_S(t_0)\hat W^n(t,t_0)
    \right].
\end{equation}
Crucially, in general,
\begin{equation}
    \hat W_A^{(n)}
    \neq
    \bigl(\hat W_A^{(1)}\bigr)^n.
\end{equation}
The reduced work statistics are therefore encoded in a hierarchy of moment operators rather than in a single observable on $A$. Although the operators $\hat W_A^{(n)}$ generally do not commute for different values of $n$, they arise as moments of a single POVM on $A$.\footnote{Write $\hat W=\sum_k w_k\hat\Pi_k$ and define $\hat E_k:=\Tr_E[(\mathbbm 1_A\otimes\hat\sigma_E)\hat\Pi_k]$. The operators $\hat E_k$ form a POVM and $\hat W_A^{(n)}=\sum_k w_k^n\hat E_k$.} The outcome distribution of this POVM determines all moments, making the entire hierarchy jointly accessible within a single measurement scheme. Analogous hierarchies describe the reduced statistics of $\Delta\hat E$ and $\hat Q$.

Applying the first-moment reduction to each term in Eq.~\eqref{eq:first_law} gives the reduced first law
\begin{equation}\label{eq:reduced_first_law}
    \hat W_A^{(1)}(t,t_0)
    =
    \Delta\hat E_A^{(1)}(t,t_0)
    +
    \hat Q_A^{(1)}(t,t_0).
\end{equation}
For an undriven interaction, $\partial_t\hat V(t)=0$, the first-moment operators take the form (cf. Appendix~\ref{app:oqs_undriven})
\begin{equation}\label{eq:OQS_first_moments}
\begin{aligned}
    \hat W_A^{(1)}(t,t_0)
    &=
    \int_{t_0}^{t}\!ds\,
    \Phi_s^\dagger[\partial_s\hat H_A(s)],\\
    \Delta\hat E_A^{(1)}(t,t_0) &= \int_{t_0}^t\!ds~\Phi_s^\dagger[\partial_s \hat H_A(s) + \mathcal L_s^\dagger[\hat H_A(s)]],\\
     \hat Q_A^{(1)}(t,t_0)
    &=
    -\int_{t_0}^{t}\!ds\,
    \Phi_s^\dagger[\mathcal L_s^\dagger[\hat H_A(s)]],
\end{aligned}
\end{equation}
where $\Phi_t^\dagger$ denotes the adjoint of the dynamical map\footnote{Defined by \mbox{$\Phi_t^\dagger[\hat O_A] = \Tr_{E}[(\mathbbm 1_A\otimes\hat\sigma_{E})\hat U^\dagger(t,t_0)(\hat O_A\otimes\mathbbm 1_{E})\hat U(t,t_0)]$}, such that 
$\Tr[\Phi_t[\hat\rho_A]\hat O_A]
=
\Tr[\hat\rho_A\Phi_t^\dagger[\hat O_A]]$, $\forall \hat\rho_A, \hat O_A$.}, and $\mathcal L_t := \dot\Phi_t\circ\Phi_t^{-1}$ is the generator of the reduced dynamics\footnote{Under the common assumption that the inverse map $\Phi_t^{-1}$ exists, albeit not necessarily as a CPTP map.}.
Their expectation values recover known expressions for mean work, $\langle\hat W\rangle=\int_{t_0}^{t}\!ds\,\Tr[\hat\rho_A(s)\partial_s\hat H_A(s)]$, and heat, $\langle\hat Q\rangle=\int_{t_0}^{t}\!ds\,\Tr[\mathcal L_s[\hat\rho_A(s)]\hat H_A(s)]$ (see, e.g., Refs.~\cite{Alicki76,Spohn78,Colla26}).

Notably, the first-moment operators in Eq.~\eqref{eq:OQS_first_moments} are entirely determined by the reduced dynamical law $\Phi_t$ (or, equivalently, by its generator $\mathcal L_t$) together with the known system Hamiltonian, without requiring explicit knowledge of the environment or the joint system--environment evolution. This simplification is specific to the first moments: higher moments involve multi-time system--environment correlations and cannot, in general, be expressed solely in terms of the reduced dynamical map without further assumptions. 

For the remainder of this section we consider the Born--Markov limit, in which a driven interaction, $\partial_t\hat V(t)\neq0$, can also be treated within the reduced description. In this regime, the dynamics of $A$ is governed by a GKLS generator whose coherent part contains the local Hamiltonian $\hat H_A(t)$ and the Lamb-shift Hamiltonian $\hat H_{\mathrm{LS}}(t)$, while its dissipative part is determined by the time-dependent jump operators $\hat L_{k\omega}(t)$. A brief review of this construction and its connection with the microscopic dynamics is given in Appendix~\ref{app:Markov-Lindblad}.

The first-moment work operator is
\begin{equation}\label{eq:OQS_work_V}
    \hat W_A^{(1)}(t,t_0)
    =
    \int_{t_0}^{t}\!ds\,
    \Phi_s^\dagger[
        \partial_s\hat H_A(s)+\hat w_V(s)
    ],
\end{equation}
where we defined the interaction work current $\hat w_V(t)$ as a function of the Lindblad rates $\gamma$ and Lamb-shift (cf. Appendix~\ref{app:OQS_first_moment_driven})
\begin{equation}
    \hat w_V(t)
    =
    \partial_t\hat H_{\mathrm{LS}}^{(T)}(t)
    +
    \frac{i}{2}
    \sum_{kk'\omega}
    \gamma_{kk'}(\omega)
    \hat L_{k\omega}^\dagger(t)
    \overleftrightarrow{\partial_t}
    \hat L_{k'\omega}(t),
\end{equation}
with
$\hat f\,\overleftrightarrow{\partial_t}\,\hat g
:=
\hat f(\partial_t \hat g)
-
(\partial_t \hat f)\hat g$.

The corresponding reduced currents for internal energy, heat, and interaction energy are also derived in Appendix~\ref{app:OQS_first_moment_driven}, showing how the terms entering $\hat w_V(t)$ are distributed among the different channels of the energy balance. In particular, the Lamb-shift contribution generates a nonvanishing interaction-energy current. Consequently, for a driven interaction, the operatorial energy-conservation law at the level of currents cannot be cast as a first law involving only work, internal energy, and heat: the interaction-energy current must also be retained. By contrast, this term vanishes for an undriven interaction, yielding an operatorial first law at the level of currents.

More broadly, the assignment of these energy contributions to work, heat, internal energy, and interaction energy can be regarded as a choice of thermodynamic gauge~\cite{Colla22,Nicacio2023, Schrauwen2025}. Accordingly, the reduced operators in Eqs.~\eqref{eq:OQS_first_moments} and~\eqref{eq:OQS_work_V} are, in principle, gauge dependent. In our construction, however, the gauge is fixed by the microscopic description. Although a gauge transformation of the jump operators and the Lamb-shift Hamiltonian leaves the reduced dynamics unchanged, it changes the underlying microscopic system--environment coupling and therefore describes a different physical implementation. Hence, different gauge choices are not equivalent descriptions of the same process, but correspond to genuinely different thermodynamic scenarios.

As in the exact undriven case, Eq.~\eqref{eq:OQS_work_V} depends only on the reduced dynamical law: the reduced propagator and the coefficients of the GKLS generator determine the first-moment operator without requiring explicit knowledge of the joint system--environment evolution. Strikingly, within the Born--Markov approximation, this reduction is not confined to the first moment but extends to the entire hierarchy of work moments. Indeed, applying a generalization of the quantum regression theorem, derived in Appendix~\ref{app:generalized_QRT}, to the multi-time correlations entering $\hat W_A^{(n)}$, we obtain each work-moment operator entirely from reduced dynamical quantities, as a combinatorial sum of nested adjoint propagators acting on the driving terms $\partial_t\hat H_A(t)$ and $\partial_t\hat L_{k\omega}(t)$. For Gaussian environments, the reduced work moments organize into a regular contribution, a direct environment-noise term, and an out-of-time-order correction. This gives the following result.
\begin{theorem}[GKLS work moments]
\label{th:oqs_work_moments}
For a driven system described in the GKLS limit, the reduced $n$-th work moment is
\begin{multline}
\label{eq:gkls_work_moments}
    \hat W_A^{(n)}(t,t_0) \simeq\left(\int_{t_0}^{t}\!ds\,\Phi_s^\dagger\circ\mathfrak W_s \circ(\Phi_s^\dagger)^{-1}\right)^n[\mathbbm 1_A]\\
    + \hat W_N^{(n)}(t,t_0) + \hat W_{\rm OO}^{(n)}(t,t_0).
\end{multline}
Furthermore, if the interaction is undriven, $\partial_s\hat V(s)=0$ for $t_0\leq s\leq t$, then $\mathfrak W_s$ reduces to left multiplication by $\partial_s\hat H_A(s)$ and the direct environment-noise contribution $\hat W_{N}^{(n)}$ vanishes.
\end{theorem}

\begin{proof}
Applying the generalized regression theorem (Appendix~\ref{app:generalized_QRT}) to the $n$-fold expansion of the work operator separates the result into the regular, direct environment-noise, and out-of-time-order contributions. The complete derivation is given in Appendix~\ref{app:OQS_higher_moments}.
\end{proof}

Here, $\mathfrak W_s$ denotes the \emph{work-current superoperator}, whose action on the identity yields the reduced instantaneous work current entering Eq.~\eqref{eq:OQS_work_V}, $\mathfrak W_s[\mathbbm 1_A] = \hat w_A(s) := \partial_s\hat H_A(s)+\hat w_V(s)$.
More generally, it incorporates the regular higher-moment contributions generated by the driving of the local Hamiltonian and the jump operators. The term $\hat W_N^{(n)}$ accounts for the additional correlation-induced noise associated with a driven coupling, while $\hat W_{\rm OO}^{(n)}$ add out-of-time order corrections. The explicit definitions of these quantities are given in Appendix~\ref{app:OQS_higher_moments}.

Importantly, Theorem~\ref{th:oqs_work_moments} shows that, in the GKLS limit, the complete hierarchy of reduced work moments can be reconstructed from the reduced propagator and the corresponding work, noise, and out-of-time-order contributions. Since these terms are expressed through system maps and the bath correlation functions entering the GKLS description, the work moments can be obtained without explicitly evaluating the joint system--environment evolution. The system-only description available exactly for the first moment therefore extends, within this approximation, to the entire hierarchy, even when the interaction is driven.

Combining Theorems~\ref{th:fluctuation_theorem} and \ref{th:oqs_work_moments} yields a quantum fluctuation relation expressed in terms of the reduced work moments. Consider an initially factorized thermal state $\hat\pi_A(t_0)\otimes\hat\pi_E$ and assume that the interaction vanishes at the protocol boundaries. Since $\hat H_E$ is unchanged, its partition function cancels from the free-energy difference, giving $\Delta F_S=\Delta F_A$, with $\Delta F_A=k_BT\ln[Z_A(t_0)/Z_A(t)]$. From the definition in Eq.~\eqref{eq:work_reduction_moment}, and Theorem~\ref{th:fluctuation_theorem}, we obtain
\begin{align}
\label{eq:reduced_moment_fluctuation_relation}
    \braket{e^{-\beta\hat W(t,t_0)}}_{\hat\pi_A(t_0)\otimes\hat\pi_E}
    &= \sum_{n=0}^{\infty} \frac{(-\beta)^n}{n!}
    \braket{\hat W_A^{(n)}(t,t_0)}_{\hat\pi_A(t_0)}
    \nonumber\\
    &=e^{-\beta\Delta F_A}+\Delta_{t,t_0}.
\end{align}
By Theorem~\ref{th:oqs_work_moments}, every reduced moment in this
series is determined by the reduced description, while
$\Delta F_A$ depends only on reduced-system quantities.
Consequently, $\Delta_{t,t_0}$ can be evaluated without
explicitly constructing either the global work operator or the
joint system--environment dynamics. This establishes a quantum fluctuation theorem for open quantum systems entirely within the reduced GKLS description.

Finally, as a simple illustration of the higher-moment
construction, we consider the work variance for an undriven
interaction and an initial thermal state.
Section~\ref{sec:fluctuation_theorem} showed that the work operator
and the TPM scheme yield the same first two moments for
energy-diagonal initial states. Consequently, under the
weak-coupling assumptions of Ref.~\cite{Miller2019}, the variance
of the work operator coincides with their result for the TPM work
variance, as we verify explicitly in
Appendix~\ref{app:OQS_second_moment_driven}. This agreement
implies that, in the slow-driving regime
\cite{Speck04,Cavina17,Miller2019,Scandi20}, the quantum
fluctuation--dissipation relation of Ref.~\cite{Miller2019} also
holds for the work operator and reads
$\beta\sigma_W^2/2=W_{\mathrm{diss}}+\mathcal Q$, where
$W_{\mathrm{diss}}=\braket{\hat W}-\Delta F_A$ and
$\mathcal Q\geq0$ is the quantum contribution arising from the
non-commutativity of the driving. Expanding
Eq.~\eqref{eq:reduced_moment_fluctuation_relation} consistently
to leading nontrivial order in linear response and using this
identity gives
\begin{equation}
\label{eq:Delta_Q}
    \Delta_{t,t_0}
    \simeq
    \beta e^{-\beta\Delta F_A}
    \left(
        \frac{\beta}{2}\sigma_W^2-W_{\mathrm{diss}}
    \right)
    =
    \beta e^{-\beta\Delta F_A}\mathcal Q.
\end{equation}
Thus, at this order, the correction to the Jarzynski equality is
governed by the same non-commutative contribution that modifies
the classical fluctuation--dissipation relation.

\section{Discussion: Work is an Observable}\label{sec:discussion}
The autonomous construction above gives the work operator a precise physical meaning as the energy transferred during a fixed quantum process. We now revisit the main objections to treating work as an observable and show how they are resolved once the control system and the full energetic implementation are included in the description.

\paragraph*{Work is a process observable, not a state function.}
A common objection to the work operator is that work characterizes a process, rather than an intrinsic property of the state~\cite{Talkner2007,Talkner2016}. We agree with the premise, but not with the conclusion. The work operator does not attempt to represent work as a state function. Its definition in Eq.~\eqref{eq:work_autonomous}
depends explicitly on the whole protocol through both the Hamiltonian path and the corresponding unitary dynamics. The operator is evaluated on the initial state, but it is parametrized by the process. There is therefore no contradiction between work being a process quantity and being represented, once the process is fixed, by an observable. What must be avoided is the classical interpretation that assigns simultaneous definite values to the initial and final energies when the corresponding operators do not commute. The work operator is not such a joint assignment; it is the quantum observable associated with the energy displacement induced by the process.

\paragraph*{Work eigenvalues are not necessarily energy differences.}
A connected criticism~\cite{Baumer2018} is that the eigenvalues of $\hat W(t)$ are not, in general, differences between eigenvalues of $\hat H(t)$ and $\hat H(0)$. This is true, but it is not a pathology. In quantum mechanics, the spectrum of a difference of two non-commuting observables is not the set of pairwise differences of their spectra. Requiring work values to be energy-eigenvalue differences amounts to assuming that the initial energy and the final Heisenberg-evolved energy admit a joint distribution. Such a joint distribution exists in the classical, commuting case, but not in general. The fact that the work spectrum is not given by TPM-like energy differences is precisely what allows the work operator to retain the contribution of energy coherences, rather than destroying them through an initial projective measurement. In the commuting limit of Eq.~\eqref{eq:classicality}, the two energy observables can be jointly diagonalized, and the work eigenvalues reduced to the corresponding classical energy changes.

The right interpretation of $\hat W(t)$ is instead that of an energy displacement induced by the process. This is analogous to other displacement observables in quantum mechanics. For instance, a finite displacement can be represented by $\hat x(t)-\hat x(0)$, even though $\hat x(t)$ and $\hat x(0)$ need not commute, and its spectrum need not be the set of pairwise differences between two position eigenvalues. Similarly, the velocity operator is operationally meaningful even though it is not obtained by assigning simultaneous sharp values to the position at two different times. Its eigenstates are states of definite change of  position, despite initial and final position not being well-defined. The same applies to work: eigenstates of $\hat W(t)$ are states of definite energetic displacement, even though they need not have a well-defined initial or final energy. 

\paragraph*{The TPM scheme is the operational definition of work.}
The two-point measurement scheme becomes natural if one insists that fluctuating work must be defined as the difference between two measured energies~\cite{Talkner2007}. One then has to measure the energy at the beginning and at the end of the process, and work becomes a property of the resulting measurement record. This is also the viewpoint behind the criticism that, without such a measurement record, the work operator is merely counterfactual: no agent has actually measured the initial and final energies, so there would be no operational fact of the matter about the work performed.

On the contrary, in the autonomous formulation, work is not introduced as the difference between two measurement outcomes, but as the energy transferred from the clock or battery to the system in an energy-conserving dynamics. This allows a definition of work that is completely independent of the intervention of an agent, and measurement need not be part of the definition of work. It therefore also avoids the infinite-resource requirement associated with the ideal measurements underlying the TPM scheme~\cite{Guryanova20}. Importantly, this definition is operationally testable in the same sense as any other quantum observable: once the process is fixed, so is the corresponding work operator, and its statistics can be obtained by performing a projective measurement in its eigenbasis on the initial state. Although our construction derives this observable from energy conservation rather than postulating a measurement prescription, Ref.~\cite{silva2025consistency} showed that the associated protocol is uniquely selected by four broad physical requirements: conservation laws at the level of the full probability distribution, reality, independence from the initial state, and no-signaling. This provides independent axiomatic support for the work-operator measurement and, whenever these requirements are imposed, singles it out over the TPM scheme.

\paragraph*{Jarzynski equality and the no-go theorem.}
It is often emphasized that the work operator does not satisfy the classical Jarzynski equality for arbitrary coherent quantum protocols~\cite{Talkner2007}. This, however, should not be regarded as a failure of the construction. The Jarzynski equality is a classical fluctuation relation, and therefore should only be expected to hold in a genuinely classical regime.  The work operator satisfies the quantum fluctuation Theorem~\ref{th:fluctuation_theorem}, whose correction term vanishes precisely in the commuting, and hence classical, limit. In that regime the standard Jarzynski equality is recovered. The no-go theorem of Ref.~\cite{Perarnau17} is therefore evaded because it adopts a weaker notion of classicality, based only on the absence of initial coherences, while non-commutative effects generated by the protocol itself may still be present: even an initially energy-diagonal state can undergo a genuinely quantum process if the Hamiltonians at different times do not commute. 

\paragraph*{Gauge dependence.}
Finally, one may object that the work operator is gauge-dependent~\cite{Campisi2011}: adding a time-dependent scalar $f(t)\,\mathbbm 1$ to the system Hamiltonian leaves the reduced dynamics unchanged, but shifts the work by $f(\tau)-f(0)$ (see Appendix~\ref{app:gauge}). In the non-autonomous description, this may appear to be an arbitrary gauge freedom. In the autonomous description, instead, the shift has a direct physical interpretation. Different choices of $f(t)$ correspond to different couplings between the system and the clock, and therefore to different energetic implementations of the same reduced unitary dynamics. The reduced dynamics of the system alone is not sufficient to determine thermodynamic work: the same channel can be generated by physically distinct drives with different energetic costs. Work measures precisely this cost, and is therefore a property of the full implementation, not only of the induced dynamics on the system.

\section{Conclusion}\label{sec:conclusion}

In this manuscript we have constructed an autonomous formulation of driven quantum dynamics. Starting from an arbitrary time-dependent Hamiltonian, we showed how to embed the corresponding evolution into a larger, time-independent dynamics on a composite system. In this picture, the external control is no longer treated as a classical primitive, but is included explicitly as part of the quantum description. This addresses a central limitation of the standard framework: the energetic cost of classical control may be orders of magnitude larger than the work exchanged by the quantum system itself, yet it is normally left outside the description. Autonomization overcomes this difficulty by reformulating the driven process as an energy-conserving evolution of a closed composite system.

This allowed us to derive a work operator in a natural way. Since the composite dynamics is energy conserving, work is unambiguously identified with the energy transferred between the system and the auxiliary degrees of freedom. The corresponding observable on the system alone is then uniquely determined by the requirement that its statistics reproduce those of the autonomous energy exchange. Thus, the work operator is not postulated, nor introduced through a measurement prescription, but follows from the autonomous energy-transfer description. In this sense, the construction avoids the ambiguities of other approaches to quantum work: all energetic contributions are explicitly accounted for, and the work operator characterizes the undisturbed process rather than a measurement-modified one.

We then derived a quantum fluctuation theorem for this operator. The resulting relation reduces to the usual Jarzynski equality whenever the relevant Hamiltonians commute, while genuinely quantum protocols give rise to an explicit non-commutative correction. This identifies the obstruction to the classical fluctuation relation directly in terms of non-commutativity. It also clarifies why our construction evades existing no-go theorems for quantum work: we do not impose agreement with the TPM scheme as a defining classicality condition. Instead, the autonomous formulation identifies the relevant classical regime directly: the Jarzynski equality is recovered exactly when the non-commutative correction to the energy-transfer fluctuations vanishes.

The same framework provides a particularly transparent way to study finite control. Since the control system is part of the physical description, imperfect preparations or uncertain implementations can be incorporated directly, leading to coarse-grained work operators and corresponding corrections to the fluctuation theorem. This makes explicit the different energetic contributions involved in the experiment, including both the work exchanged with the system and the resources required to realize the control. 

Finally, we extended the construction to open quantum systems and formulated the first law directly at the operator level, identifying distinct operators for work, heat, internal-energy change, and interaction-energy change. The global work statistics induce a hierarchy of moment operators on the accessible subsystem, which cannot generally be represented as powers of a single reduced observable. For an undriven interaction, the first-moment operators are determined entirely by the reduced dynamics, while in the GKLS regime this extends to the full hierarchy of moments, including for driven interactions.

The autonomous formulation therefore provides a self-contained setting in which work, heat, internal-energy changes, fluctuations, coherence, finite control, and control costs can be treated within a unified energy-conserving description.

\section*{Acknowledgments} 

The authors warmly thank Paolo Abiuso, Pharnam Bakhshinezhad, \v{C}aslav Brukner, Alessandra Colla, Mohammad Mehboudi, Martí Perarnau-Llobet, and Michalis Skotiniotis for insightful discussions.

D.C., M.H., and A.R. acknowledge funding from the European Research Council (Consolidator grant ‘Cocoquest’ 101043705). A.R. acknowledges funding from the Swiss National Science Foundation (Postdoc.Mobility grant `CATCH' 225461). This research was funded in part by the Austrian Science Fund (FWF) [10.55776/F71] and
[10.55776/COE1]. This publication was made possible through the
financial support of WOST (WithOutSpaceTime) grant
from the John Templeton Foundation. The opinions expressed in this publication are those of the authors and
do not necessarily reflect the views of the John Templeton Foundation
For the purpose of open access, the author(s) have applied a CC BY license to any author-accepted manuscript version arising from this submission.

\bibliography{mybib.bib}

\widetext
\newpage
\appendix

\section{Autonomous picture}
\subsection{Construction of the autonomous picture}\label{app:autonomous}
We give here the proof of Theorem~\ref{th:autonomization}. We begin by applying Trotter formula to the full evolution, obtaining
\begin{equation}
    e^{-i \hat H_\text{tot} t} = \lim_{n\to\infty} \left(e^{-i \hat H_C t/n} e^{-i \hat H_\text{int} t/n}\right)^n.
\end{equation}
Then, since 
\begin{equation}
    \hat H_\text{int} \ket{\psi} \ket{t_0} = \hat H_S(t_0) \ket{\psi} \ket{t_0}
\end{equation}
it follows that 
\begin{equation}
    e^{-i \hat H_C t/n} e^{-i \hat H_\text{int} t/n} \ket{\psi}_S \ket{t_0}_C = e^{-i \hat H_S(t_0) t/n} \ket{\psi} \ket{t_0 + t/n}.
\end{equation}

We then apply the evolution $n$ times and we find
\begin{equation}
    \left(e^{-i \hat H_C t/n} e^{-i \hat H_\text{int} t/n}\right)^n \ket{\psi}_S \ket{t_0}_C = \prod_{k=1}^n e^{-i \hat H_S(t_0+k\,t/n)\, t/n} \ket{\psi} \ket{t_0+t} \xrightarrow[n\to\infty]{} \mathcal{T}\exp\left(i \int_{t_0}^{t_0+t} dt^\prime \hat H_S(t^\prime)\right) \ket{\psi} \ket{t_0+t},
\end{equation}
meaning that
\begin{equation}
    e^{-i \hat H_\text{tot} t} \ket{\psi}\ket{t_0}= \hat U(t+t_0,t_0)\ket{\psi} \ket{t_0+t}.
\end{equation}
Finally, we replace this expression in
\begin{equation}
    e^{-i \hat H_\text{tot} t}  \,\hat\rho_S \otimes \hat \rho_C \,e^{i \hat H_\text{tot} t} = \int dt_0^\prime dt_0^{\prime \prime}  \braket{t_0^\prime|\hat \rho_C |t_0^{\prime \prime}} e^{-i \hat H_\text{tot} t} \,\hat \rho_S \otimes \ket{t_0^\prime}\!\bra{t_0^{\prime \prime}} \,e^{i\hat H_\text{tot}t}.
\end{equation}
and find 
\begin{equation}
    e^{-i \hat H_\text{tot} t} \, \hat \rho_S \otimes\hat  \rho_C \, e^{i \hat H_\text{tot} t} = \int dt_0^\prime dt_0^{\prime \prime}  \braket{t_0^\prime|\hat \rho_C |t_0^{\prime \prime}} \, \hat U(t+t_0',t_0')\, \hat\rho_S \,\hat U^\dagger(t+t_0'',t_0'') \otimes \ket{t+t_0^\prime}\!\bra{t+t_0^{\prime \prime}},
\end{equation}
which yields Eq.~\eqref{eq:autonomous_evolution} by taking a partial trace and a shift in the integration variable. Specializing this expression for $\hat \rho_C=\ket{t_0}\!\bra{t_0}$ gives Eq.~\eqref{eq:autonomous_vs_nonautonomous}.
\subsection{Beyond the unbounded Hamiltonian}\label{app:autonomous_physical}
We give here two models for the autonomization that do not require an unbounded Hamiltonian for the clock, and show how analogous results to Theorem~\ref{th:autonomization} can be derived. The first model involves a bounded continuous spectrum, while the second one has a finite spectrum.
\subsubsection{Bounded continuous spectrum}
Consider a generic clock Hamiltonian $\hat H_C$ with bounded continuous spectrum $\text{spec}(\hat H_C)\subset [E_-,E_+]$. Let $\{\ket{E}\}_{E\in[E_-,E_+]}$ denote its energy eigenstates. We define the band-limited time states as
\begin{equation}
\ket{t}:=\frac{1}{\sqrt{\Delta E}}\int_{E_-}^{E_+} dE\, e^{-iEt}\ket{E},
\qquad
\Delta E=E_+ - E_-.
\end{equation}
These states are translated correctly by $\hat H_C$, namely
\begin{equation}
e^{-i\hat H_C s}\ket{t}=\ket{t+s}.
\end{equation}
However, because the spectrum is bounded, they are not orthogonal in general:
\begin{equation}
\braket{t|t'}
=
\frac{1}{\Delta E}\int_{E_-}^{E_+} dE\, e^{-iE(t'-t)}
=
e^{-i\bar E(t'-t)}
\frac{\sin\!\left[\frac{\Delta E}{2}(t'-t)\right]}{\frac{\Delta E}{2}(t'-t)},
\qquad
\bar E= \frac{E_+ + E_-}{2}.
\end{equation}
Nevertheless, they become orthogonal when restricted to the discrete time lattice
\begin{equation}
t_n = n\,\Delta t,
\qquad
\Delta t = \frac{2\pi}{\Delta E},
\end{equation}
so that
\begin{equation}
\braket{t_n|t_m} = \delta_{nm}.
\end{equation}

We now define the total Hamiltonian by
\begin{equation}
\hat H_\text{tot} = \hat H_C + \hat H_\text{int},
\end{equation}
with interaction term
\begin{equation}
\hat H_\text{int}=\sum_n \hat H_S(t_n)\,\ket{t_n}\!\bra{t_n}.
\end{equation}

Let the initial clock state be $\hat \rho_C = \ket{0}\!\bra{0}$. In Theorem~\ref{th:autonomization}, we showed that the autonomization construction in the unbounded case reproduces a non-autonomous system dynamics generated by a time-dependent Hamiltonian. In the present bounded setting, this is no longer true in general: the reduced system evolution is instead described by a CPTP map,
\begin{equation}
\hat \rho_S(t) = \Tr_C\!\left[e^{-i\hat H_\text{tot}t}\, (\hat \rho_S \otimes \ket{0}\!\bra{0})\, e^{i \hat H_\text{tot}t}\right].
\end{equation}
However, when the evolution is evaluated at the lattice times $t_n$, one recovers the desired unitary dynamics
\begin{equation}
\hat \rho_S(t_n) = \hat U(t_n,0)\,\hat \rho_S\,\hat U^\dagger(t_n,0),
\end{equation}
which can be proved following the same steps of Appendix~\ref{app:autonomous}.
Therefore, the conclusion of Theorem~\ref{th:autonomization} remains valid at this discrete set of times, as a consequence of the bounded spectrum of the clock Hamiltonian.

This admits a simple physical interpretation: the clock has a finite time resolution set by $\Delta t = \frac{2\pi}{\Delta E}$. Hence, the clock becomes more precise as the bandwidth $\Delta E$ increases. In particular, in every practical scenario where time differences cannot be resolved below a finite time resolution $\Delta t$, it is sufficient to choose the clock bandwidth such that
\begin{equation}
\Delta E > \frac{2\pi}{\Delta t},
\end{equation}
so that the induced dynamics appears unitary at all times. In the unbounded limit $\Delta E \to \infty$, the time lattice becomes continuous, and we find the results in the main text.
\subsubsection{Finite spectrum}

Consider now the clock to be a finite-dimensional system with Hilbert space
$
\mathcal H_C=\mathbb C^d
$
and orthonormal basis
$
\{\ket{j}\}_{j\in\mathbb Z_d}.
$
Let $\tau>0$ be the total duration of the experiment, and define the uniform grid
\begin{equation}
    t_j = j\,\frac{\tau}{d},
    \qquad j=0,1,\dots,d.
\end{equation}
For each interval $[t_j,t_{j+1}]$, define the exact system propagator
\begin{equation}
    \hat U_j
    :=
    \hat U(t_{j+1},t_j)
    =
    \mathcal T \exp\!\left(
        -i\int_{t_j}^{t_{j+1}} dt'\, \hat H_S(t')
    \right),
    \qquad j=0,\dots,d-1.
\end{equation}
Then, introduce the unitary operator
\begin{equation}
    \hat U_d
    :=
    \sum_{j\in\mathbb Z_d}
    \hat U_j \otimes \ket{j+1}\!\bra{j},
\end{equation}
where $j+1$ is understood modulo $d$. Since $\hat U_d$ is unitary on the finite-dimensional space
$
\mathcal H_S\otimes\mathcal H_C,
$
there exists a Hermitian operator $\hat H_\text{tot}$ such that
\begin{equation}
    e^{-i\hat H_\text{tot}\Delta t}=\hat U_d,
    \qquad
    \Delta t := \frac{\tau}{d}.
\end{equation}
One convenient choice is
\begin{equation}
    \hat H_\text{tot}
    =
    \frac{i}{\Delta t}\,\log \hat U_d
    =
    \frac{id}{\tau}\,\log \hat U_d.
\end{equation}
The corresponding autonomous evolution, generated by this time-independent Hamiltonian, is the usual
\begin{equation}
    \hat U_\text{tot}(t)=e^{-i \hat H_\text{tot}t}.
\end{equation}
At the discrete times $t_k=k\Delta t$, we then have
\begin{equation}
    \hat U_\text{tot}(t_k)
    =
    e^{-i \hat H_\text{tot}t_k}
    =
    \left(e^{-i \hat H_\text{tot}\Delta t}\right)^k
    =
    \hat U_d^k.
\end{equation}
Hence, if the clock is initialized in $\ket{0}$, then for $k=0,\dots,d$,
\begin{equation}
    \hat U_d^k\bigl(\ket{\psi}\otimes\ket{0}\bigr)
    =
    \left(
        \hat U_{k-1}\cdots \hat U_1\hat U_0\,\ket{\psi}
    \right)\otimes \ket{k \!\!\!\pmod d}.
\end{equation}
Using the composition law for propagators, this becomes
\begin{equation}
    \hat U_d^k\bigl(\ket{\psi}\otimes\ket{0}\bigr)
    =
    \bigl(
        \hat U(t_k,0)\ket{\psi}
    \bigr)\otimes \ket{k \!\!\!\pmod d}.
\end{equation}
Therefore, at the grid times $\{t_k\}_k$, the reduced system state is
\begin{equation}
    \hat \rho_S(t_k)
    :=
    \Tr_C\!\left[
        e^{-i \hat H_\text{tot}t_k}
        (\hat \rho_S\otimes\ket{0}\!\bra{0})
        e^{i \hat H_\text{tot}t_k}
    \right]
    =
    \hat U(t_k,0)\,\hat \rho_S\,\hat U^\dagger(t_k,0).
\end{equation}
For intermediate times $t\notin\{t_k\}_k$, the reduced dynamics is a CPTP map, which in general does not coincide exactly with the driven evolution generated by $\hat H_S(t)$. However, as in the previous section, given a physical time resolution $\Delta t$, it is sufficient to choose
\begin{equation}
    d > \frac{\tau}{\Delta t}
\end{equation}
to obtain an evolution that, for all practical purposes, is undistinguishable from the unitary dynamics generated by $\hat H_S(t)$. In the limit $d\to\infty$, the grid becomes dense and the exact stroboscopic reconstruction at the times $t_k$ approaches the continuous-time driven evolution of the main text.
\section{Illustrative example}\label{app:example}
Consider the example discussed in the main text: a qubit initially prepared in
the state $\ket{\psi}=\ket{0}$ and subject to the instantaneous quench
\begin{equation}
  \hat H_S(t)=
  \begin{cases}
    \hat \sigma_x, & t=0,\\[2mm]
    \hat \sigma_z+\hat \sigma_x, & t>0.
  \end{cases}
\end{equation}
For this protocol, the work operator of Eq.~\eqref{eq:work_operator_system} is
\begin{equation}
  \hat W(t)=
  \begin{cases}
    0, & t=0,\\[2mm]
    \hat \sigma_z, & t>0.
  \end{cases}
\end{equation}
Since the initial state $\ket{0}$ is an eigenstate of $\hat \sigma_z$, the work
is deterministic:
\begin{equation}
    w=\braket{0|\hat\sigma_z|0}=1 .
\end{equation}

Let us now compare this prediction with the TPM scheme~\cite{Tasaki2000,Kurchan2001}. Immediately after the
quench, at $t=0^+$, the TPM work distribution is
\begin{equation}
  p_{\rm TPM}(w,0^+)
  =
  \sum_{E,E'}
  \delta\bigl(w-(E'-E)\bigr)\,
  p^{0^+}(E'|E)\,p^0(E).
\end{equation}
The first measurement is performed in the eigenbasis of the initial Hamiltonian
$\hat\sigma_x$, whose eigenstates we denote by $\ket{\pm}$ with eigenvalues
$E=\pm1$. Since
$ \bigl|\braket{+|0}\bigr|^2 = \bigl|\braket{-|0}\bigr|^2 = \frac{1}{2},$
we have
\begin{equation}
    p^0(E=\pm1)=\frac{1}{2}.
\end{equation}
The second measurement is performed immediately after the quench, in the
eigenbasis of the final Hamiltonian $\hat\sigma_x+\hat\sigma_z$. Its normalized
eigenstates are
\begin{equation}\label{eq:H_1_eig}
  \ket{E'_+}
  =
  \frac{1}{\sqrt{2(2+\sqrt{2})}}
  \begin{pmatrix}
    1+\sqrt{2}\\[1mm]
    1
  \end{pmatrix},
  \qquad
  \ket{E'_-}
  =
  \frac{1}{\sqrt{2(2-\sqrt{2})}}
  \begin{pmatrix}
    1-\sqrt{2}\\[1mm]
    1
  \end{pmatrix},
\end{equation}
with eigenvalues $E'_\pm=\pm\sqrt{2}$. Therefore,
\begin{equation}
  p^{0^+}(E'|E)=
  \begin{cases}
    \dfrac{2+\sqrt{2}}{4}, &
    (E'=\sqrt{2},\,E=1)
    \ \vee\
    (E'=-\sqrt{2},\,E=-1),\\[3mm]
    \dfrac{2-\sqrt{2}}{4}, &
    (E'=-\sqrt{2},\,E=1)
    \ \vee\
    (E'=\sqrt{2},\,E=-1).
  \end{cases}
\end{equation}
Combining these transition probabilities with $p^0(E)=1/2$, one obtains
\begin{equation}
  p_{\rm TPM}(w,0^+)=
  \begin{cases}
    \dfrac{2+\sqrt{2}}{8}, &
    w=\sqrt{2}-1
    \ \vee\
    w=1-\sqrt{2},\\[2mm]
    \dfrac{2-\sqrt{2}}{8}, &
    w=\sqrt{2}+1
    \ \vee\
    w=-\sqrt{2}-1.
  \end{cases}
\end{equation}
Thus, while the work operator predicts the deterministic value $w=1$, the TPM
scheme assigns a nonzero probability
\begin{equation}
    p_{\rm TPM}(w=1+\sqrt{2})=\frac{2-\sqrt{2}}{8}
\end{equation}
to a work value larger than one.

A similar conclusion is obtained from the quasi-probability approach~\cite{Lostaglio2018}, based on two weak energy measurements, one before and one after the quench. In this case, the work quasi-probability at $t=0^+$ is
\begin{equation}
  p_{\rm QP}(w,0^{+})
  =
  \sum_{E,E'}
  \delta\bigl(w-(E'-E)\bigr)\,
  \bra{E'}\hat U(0^{+})\ket{E}\,
  \braket{E|\psi}\braket{\psi|E'} .
\end{equation}
Using the eigenstates and eigenvalues introduced above, and noting that immediately after the quench $\hat U(0^+)=\mathbbm{1}$, we obtain
\begin{equation}
  p_{\rm QP}(w,0^{+}) =
  \begin{cases}
    \dfrac{1+\sqrt{2}}{4}, & w=\sqrt{2}-1,\\[2mm]
    \dfrac{1}{4}, & w=\sqrt{2}+1\ \ \vee\ \ w=1-\sqrt{2},\\[2mm]
    \dfrac{1-\sqrt{2}}{4}, & w=-\sqrt{2}-1.
  \end{cases}
\end{equation}
Thus, also in this approach, a nonzero weight is assigned to a work value larger than one:
\begin{equation}
    p_{\rm QP}(w=\sqrt{2}+1)=\frac{1}{4}.
\end{equation}
\section{Quantum fluctuation theorem}\label{app:fluctuation_theorem}

We prove Theorem~\ref{th:fluctuation_theorem}. For notational simplicity, we set $\beta=1$; equivalently, one may replace $\hat H\mapsto\beta\hat H$ throughout. Let
\begin{equation}
    \hat H_{t,t_0}
    :=
    \hat U^\dagger(t,t_0)\hat H_S(t)\hat U(t,t_0),
    \qquad
    \hat H_{t_0}
    :=
    \hat H_S(t_0),
\end{equation}
so that
\begin{equation}
    \hat W(t,t_0)=\hat H_{t,t_0}-\hat H_{t_0}.
\end{equation}
We define
\begin{equation}
    \Delta_{t,t_0}
    :=
    \braket{e^{-\hat W(t,t_0)}}_{\hat\rho_S(t_0)}
    -
    \frac{Z_t}{Z_{t_0}},
    \qquad
    \hat\rho_S(t_0)
    =
    \frac{e^{-\hat H_{t_0}}}{Z_{t_0}}.
\end{equation}
Since $\Tr[e^{-\hat H_{t,t_0}}]
    =
    \Tr[e^{-\hat H_S(t)}]
    =
    Z_t$,
this can be written as
\begin{equation}
    \Delta_{t,t_0}
    =
    \frac{1}{Z_{t_0}}
    \left(
        \Tr[
            e^{-(\hat H_{t,t_0}-\hat H_{t_0})}
            e^{-\hat H_{t_0}}
        ]
        -
        \Tr[e^{-\hat H_{t,t_0}}]
    \right).
\end{equation}
Introduce
\begin{equation}
    \hat K_\lambda(t,t_0)
    =
    \hat H_{t,t_0}
    -
    \lambda\hat H_{t_0},
    \qquad
    F_{t,t_0}(\lambda)
    =
    \Tr\!\left[
        e^{-\hat K_\lambda(t,t_0)}
        e^{-\lambda\hat H_{t_0}}
    \right].
\end{equation}
Then $F_{t,t_0}(1)
    =
    \Tr\!\left[
        e^{-\hat W(t,t_0)}
        e^{-\hat H_{t_0}}
    \right]$,
    and 
    $F_{t,t_0}(0)=Z_t$, and hence
\begin{equation}
    \Delta_{t,t_0}
    =
    \frac{1}{Z_{t_0}}
    \int_0^1 d\lambda\,
    \frac{dF_{t,t_0}(\lambda)}{d\lambda}.
\end{equation}
Using the Duhamel identity,
\begin{equation}
    \frac{d}{d\lambda}e^{\hat B(\lambda)}
    =
    \int_0^1 ds\,
    e^{(1-s)\hat B(\lambda)}
    \frac{d\hat B(\lambda)}{d\lambda}
    e^{s\hat B(\lambda)},
\end{equation}
one obtains
\begin{equation}
\begin{aligned}
    \frac{dF_{t,t_0}(\lambda)}{d\lambda}
    =
    \int_0^1 ds\,
    \Tr\!\left[
        e^{-(1-s)\hat K_\lambda(t,t_0)}
        [\hat H_{t_0},e^{-s\hat K_\lambda(t,t_0)}]
        e^{-\lambda\hat H_{t_0}}
    \right].
\end{aligned}
\end{equation}
We now use the commutator form of the same identity,
\begin{equation}
\begin{aligned}
    [\hat H_{t_0},e^{-s\hat K_\lambda(t,t_0)}]
    =
    \int_0^s du\,
    &e^{-(s-u)\hat K_\lambda(t,t_0)}
    [\hat K_\lambda(t,t_0),\hat H_{t_0}]
    e^{-u\hat K_\lambda(t,t_0)}.
\end{aligned}
\end{equation}
Since $[\hat K_\lambda(t,t_0),\hat H_{t_0}]
    =
    [\hat H_{t,t_0},\hat H_{t_0}]$, this gives
\begin{equation}
\begin{aligned}
    \frac{dF_{t,t_0}(\lambda)}{d\lambda}
    =
    \int_0^1 dr\,(1-r)\,
    \Tr\!\Big[
        &e^{(r-1)\hat K_\lambda(t,t_0)}
        [\hat H_{t,t_0},\hat H_{t_0}]
        e^{-r\hat K_\lambda(t,t_0)}
        e^{-\lambda\hat H_{t_0}}
    \Big].
\end{aligned}
\end{equation}
Substituting this expression into the integral for
$\Delta_{t,t_0}$ gives
\begin{equation}
\begin{aligned}
    \Delta_{t,t_0}
    =
    \frac{1}{Z_{t_0}}
    \int_0^1 d\lambda
    \int_0^1 dr\,(1-r)\,
    \Tr\!\Big[
        &e^{(r-1)\hat K_\lambda(t,t_0)}
        [\hat H_{t,t_0},\hat H_{t_0}]
        e^{-r\hat K_\lambda(t,t_0)}
        e^{-\lambda\hat H_{t_0}}
    \Big].
\end{aligned}
\end{equation}
Setting $x=1-r$ and defining
\begin{equation}
    \hat\Xi_{\lambda,x}
    :=
    e^{-x\hat K_\lambda(t,t_0)}
    [\hat H_{t,t_0},\hat H_{t_0}]
    e^{-(1-x)\hat K_\lambda(t,t_0)},
\end{equation}
we obtain
\begin{equation}
    \Delta_{t,t_0}
    =
    \frac{1}{Z_{t_0}}
    \int_0^1 d\lambda
    \int_0^1 dx\,x\,
    \Tr\!\left[
        \hat\Xi_{\lambda,x}
        e^{-\lambda\hat H_{t_0}}
    \right],
\end{equation}
which is the claimed expression for $\beta=1$. Restoring
$\beta$ produces the factor $\beta^2$ in
the main text.

Finally, the Golden--Thompson inequality gives
\begin{equation}
    \Tr\!\left[
        e^{-(\hat H_{t,t_0}-\hat H_{t_0})}
        e^{-\hat H_{t_0}}
    \right]
    \geq
    \Tr[e^{-\hat H_{t,t_0}}]
    =
    Z_t,
\end{equation}
with equality if and only if
\begin{equation}
    [\hat H_{t,t_0}-\hat H_{t_0},\hat H_{t_0}]
    =[\hat H_{t,t_0},\hat H_{t_0}]=
    0,
\end{equation}
or, equivalently,
\begin{equation}
    \left[
        \hat U^\dagger(t,t_0)\hat H_S(t)\hat U(t,t_0),
        \hat H_S(t_0)
    \right]
    =
    0.
\end{equation}
Therefore $\Delta_{t,t_0}\geq0$, and the correction vanishes precisely when the above commutator vanishes. In that case, the fluctuation theorem reduces to the Jarzynski equality.

\section{Imprecise control} 
In this appendix we give the proof of Theorem~\ref{th:imprecise_fluctuation_theorem} and provide an example where all the terms in the fluctuation theorem are explicitly calculated. 
\subsection{Proof of the theorem}
\label{app:precise_control_theorem}

For notational convenience, define $t=t_0+\Delta t$ and
\begin{equation}
    \hat W_s
    :=
    \hat W(t+s,t_0+s),
    \qquad
    f(s)
    :=
    \braket{e^{-\beta\hat W_s}}_{\hat\pi_S(t_0)}.
\end{equation}
By the definition of the moment operators in
Eq.~\eqref{eq:imprecise_work}, we have
\begin{align}
        \sum_{n=0}^{\infty}
        \frac{(-\beta)^n}{n!}\,
        \braket{\hat W_{G_\varepsilon}^{(n)}(t,t_0)
    }_{\hat\pi_S(t_0)}
    =
    \int ds\,G_\varepsilon(s)\,
    \braket{e^{-\beta\hat W_s}}_{\hat\pi_S(t_0)}
    =\int ds\,G_\varepsilon(s)\,f(s).
\end{align}
Expanding around the intended initial time gives
\begin{equation}
    f(s)
    =
    f(0)
    +s f'(0)
    +\frac{s^2}{2}f''(0)
    +o(s^2).
\end{equation}
Using
\begin{equation}
    \int ds\,sG_\varepsilon(s)=0,
    \qquad
    \int ds\,s^2G_\varepsilon(s)=\varepsilon^2,
\end{equation}
we obtain
\begin{equation}
    \int ds\,G_\varepsilon(s)\,f(s)
    =
    f(0)
    +\frac{\varepsilon^2}{2}f''(0)
    +o(\varepsilon^2).
\end{equation}
The zeroth-order term follows from
Theorem~\ref{th:fluctuation_theorem},
\begin{equation}
    f(0)
    =
    e^{-\beta\Delta F}
    +\Delta_{t,t_0},
\end{equation}
while
\begin{equation}
    I_{t,t_0}
    =
    \frac{1}{2}f''(0)
    =
    \frac{1}{2}
    \frac{d^2}{ds^2}
    \braket{
        e^{-\beta\hat W(t+s,t_0+s)}
    }_{\hat\pi_S(t_0)}
    \bigg|_{s=0}.
\end{equation}
Combining these expressions proves
Theorem~\ref{th:imprecise_fluctuation_theorem}.

\subsection{Example}\label{app:precise_control_example}
We illustrate finite control with a smoothly driven, exactly solvable
qubit. Consider
\begin{equation}
    \hat H_S(t)
    =
    -\dot\phi(t)\hat\sigma_y
    +
    g(t)\hat R(t)\hat\sigma_x\hat R^\dagger(t),
    \qquad
    \hat R(t)=e^{i\phi(t)\hat\sigma_y}.
\end{equation}
Since $i\dot{\hat R}(t)\hat R^\dagger(t)
=-\dot\phi(t)\hat\sigma_y$, the propagator is
\begin{equation}
    \hat U(t,t_0)
    =
    \hat R(t)
    \exp\left[
        -i\hat\sigma_x\int_{t_0}^{t}dr\,g(r)
    \right]
    \hat R^\dagger(t_0).
\end{equation}
We specialize to $\phi(t)=\Omega t/2$ and $g(t)=gt^3$, obtaining
\begin{equation}
    \hat H_S(t)
    =
    -\frac{\Omega}{2}\hat\sigma_y
    +
    gt^3
    \left[
        \cos(\Omega t)\hat\sigma_x
        +
        \sin(\Omega t)\hat\sigma_z
    \right],
\end{equation}
and the propagator
\begin{equation}
    \hat{U}(t,t_0) = e^{i \frac{\Omega t}{2} \hat\sigma_y} e^{-i \frac{g}{4}(t^4-t_0^4) \hat\sigma_x} e^{-i \frac{\Omega t_0}{2} \hat\sigma_y} .
\end{equation}
When $\Omega=0$, the Hamiltonians commute at different times, whereas
for $g=0$ the Hamiltonian is time independent. The initial state $\hat{\pi}_S(t_0) = e^{-\beta \hat{H}(t_0)}/Z_{t_0}$ can be written in an useful form by using the Bloch vector $\bold{h}$ of the Hamiltonian $\hat{H}_S(t) = \bold{h}(t) \cdot \hat{\boldsymbol{\sigma}}$, with $\bold{h}(t) = (gt^3 \cos(\Omega t), - \frac{\Omega}{2}, gt^3 \sin(\Omega t))^T$, the Pauli vector $\hat{\boldsymbol{\sigma}} = (\hat\sigma_x, \hat\sigma_y, \hat\sigma_z)^T$, and length $E_t = |\bold{h}(t)| = \sqrt{g^2 t^6 + \frac{\Omega}{2}}$. This gives us 

\begin{equation}
    e^{- \beta \, \bold{h}(t) \cdot \hat{\boldsymbol{\sigma}}} = \cosh(\beta E_t) \, \mathbb{I} - \frac{\sinh(\beta E_t)}{E_t} \, \bold{h}(t) \cdot \hat{\boldsymbol{\sigma}} ,
\end{equation}

and $Z_t = 2 \cosh(\beta E_t)$. The initial state at $t_0$ then reads

\begin{equation}
    \hat\pi_S(t_0) = \frac{1}{2} \left( \mathbb{I} - \frac{\tanh(\beta E_{t_0})}{E_{t_0}} \, \bold{h}(t) \cdot \hat{\boldsymbol{\sigma}} \right) .
\end{equation}

Consider a protocol evolving from $t_0$ to $t=t_0+\Delta t$. For a timing error $s$,
define $\hat W_s:=\hat W(t+s,t_0+s)$,
$D_s:=(t+s)^3-(t_0+s)^3$, and
$\Phi_s:=\frac{g}{2}[(t+s)^4-(t_0+s)^4]$. The work operator can be written in terms of Pauli matrices
\begin{equation}
    \hat W_s
    \equiv 
    \bold{w}_s \cdot \hat{\boldsymbol{\sigma}}
    =
    X_s\hat\sigma_x
    +
    Y_s\hat\sigma_y
    +
    Z_s\hat\sigma_z,
\end{equation}
with the Bloch vector of the work operator $\bold{w}_s = (X_s,Y_s,Z_s)^T$, and
\begin{align}
    X_s
    &=
    gD_s\cos\bigl(\Omega(t_0+s)\bigr)
    -
    \frac{\Omega}{2}
    \sin\Phi_s\sin\bigl(\Omega(t_0+s)\bigr),
    \\
    Y_s
    &=
    \frac{\Omega}{2}(1-\cos\Phi_s),
    \\
    Z_s
    &=
    gD_s\sin\bigl(\Omega(t_0+s)\bigr)
    +
    \frac{\Omega}{2}
    \sin\Phi_s\cos\bigl(\Omega(t_0+s)\bigr).
\end{align}
The norm, given by $\chi_s = |\bold{w}_s| = \sqrt{X_s^2 + Y_s^2 + Z_s^2}$, is then
\begin{equation}
    \chi_s
    =
    \sqrt{
        g^2D_s^2
        +
        \Omega^2\sin^2\left(\frac{\Phi_s}{2}\right)
    }.
\end{equation}
Analogously to before, we write
\begin{equation}
    e^{-\beta \hat{W}_s} = \cosh(\beta \chi_s) \, \mathbb{I} - \frac{\sinh(\beta \chi_s)}{\chi_s} \, \bold{w}_s \cdot \hat{\boldsymbol{\sigma}} ,
\end{equation}
where we find
\begin{equation}
    \braket{e^{-\beta\hat W_s}}_{\hat\pi_S(t_0)}
    =
    \cosh(\beta\chi_s)
    +
    \frac{\tanh(\beta E_{t_0})}{E_{t_0}}\,
    K_s \, \frac{\sinh(\beta\chi_s)}{\chi_s} ,
    \label{eq:smooth_example_shifted_fluctuation}
\end{equation}
with $K_s \equiv \bold{h}(t_0) \cdot \bold{w}_s$, resulting in
\begin{equation}
    K_s
    =
    gt_0^3
    \left(
        X_s \cos(\Omega t_0)
        +
        Z_s \sin(\Omega t_0)
    \right)
    -
    \frac{\Omega}{2}Y_s.
\end{equation}
The quantum correction using $e^{-\beta \Delta F} = Z_t/Z_{t_0}$, is therefore
\begin{equation}
    \Delta_{t,t_0}
    =
    \braket{e^{-\beta\hat W_0}}_{\hat\pi_S(t_0)}
    -
    e^{-\beta\Delta F}
    =
    \cosh(\beta\chi_0)
    +
    \frac{\tanh(\beta E_{t_0})}{E_{t_0}}\,
    K_0\frac{\sinh(\beta\chi_0)}{\chi_0}
    -
    \frac{\cosh(\beta E_t)}
    {\cosh(\beta E_{t_0})}.
    \label{eq:example_quantum_correction}
\end{equation}

To calculate the finite-control correction, we denote derivatives with
respect to $s$, evaluated at $s=0$, by a dot and a subscript $0$. In
particular, $\dot D_0=3(t^2-t_0^2)$,
$\ddot D_0=6(t-t_0)$,
$\dot\Phi_0=2gD_0$, and
$\ddot\Phi_0=6g(t^2-t_0^2)$. It follows that
\begin{align}
    \dot\chi_0
    &=
    \frac{
        g^2D_0\dot D_0
        +
        \frac{\Omega^2}{4}\sin\Phi_0\,\dot\Phi_0
    }{\chi_0},
    \\
    \ddot\chi_0
    &=
    \frac{
        g^2(\dot D_0^2+D_0\ddot D_0)
        +
        \frac{\Omega^2}{4}
        \left(
            \cos\Phi_0\,\dot\Phi_0^2
            +
            \sin\Phi_0\,\ddot\Phi_0
        \right)
    }{\chi_0}
    -
    \frac{\dot\chi_0^2}{\chi_0},
    \\
    \dot K_0
    &=
    gt_0^3
    \left(
        g\dot D_0
        -
        \frac{\Omega^2}{2}\sin\Phi_0
    \right)
    -
    \frac{\Omega^2}{4}\sin\Phi_0\,\dot\Phi_0,
    \\
    \ddot K_0
    &=
    gt_0^3
    \left(
        g\ddot D_0
        -
        \Omega^2gD_0
        -
        \Omega^2\cos\Phi_0\,\dot\Phi_0
    \right)
    -
    \frac{\Omega^2}{4}
    \left(
        \cos\Phi_0\,\dot\Phi_0^2
        +
        \sin\Phi_0\,\ddot\Phi_0
    \right).
\end{align}

Defining $h_\beta(\chi):=\sinh(\beta\chi)/\chi$, with $h_\beta'(\chi) \equiv \frac{\partial}{\partial \chi} h_\beta(\chi)$, gives
\begin{equation}
    h_\beta'(\chi)
    =
    \frac{
        \beta\chi\cosh(\beta\chi)-\sinh(\beta\chi)
    }{\chi^2},
    \qquad
    h_\beta''(\chi)
    =
    \frac{
        (\beta^2\chi^2+2)\sinh(\beta\chi)
        -
        2\beta\chi\cosh(\beta\chi)
    }{\chi^3},
\end{equation}
where we obtain
\begin{align}
    I_{t,t_0}
    &=
    \frac{\beta}{2}\ddot\chi_0
    \sinh(\beta\chi_0)
    +
    \frac{\beta^2}{2}\dot\chi_0^2
    \cosh(\beta\chi_0)
    +
    \frac{\tanh(\beta E_{t_0})}{2E_{t_0}}
    \left(
        \ddot K_0\,h_\beta(\chi_0)
        +
        2\dot K_0\,h_\beta'(\chi_0)\dot\chi_0
        +
        K_0
        \left(
            h_\beta''(\chi_0)\dot\chi_0^2
            +
            h_\beta'(\chi_0)\ddot\chi_0
        \right)
    \right).
    \label{eq:example_finite_control_correction}
\end{align}

In the commuting limit $\Omega=0$, this reduces to
\begin{equation}
    I_{t,t_0}
    =
    \frac{
        3\beta g(t-t_0)\sinh(\beta gt^3)
        +
        \frac{9}{2}\beta^2g^2(t^2-t_0^2)^2\cosh(\beta gt^3)
    }{
        \cosh(\beta gt_0^3)
    }.
\end{equation}
Thus, imperfect timing control produces a nonvanishing correction even
in the classical limit, when the Hamiltonians commute at different
times.

It is interesting to compare the two corrections in the near-classical,
weak-imprecision regime. We set $t_0=0$ and consider
$\varepsilon/t\ll1$ and $|\Omega|\ll gt^3$, while keeping
$x:=\beta gt^3$ fixed. Indeed, in the commuting limit $\Omega=0$, the
work operator is $\hat W(t,0)=gt^3\hat\sigma_x$, so $gt^3$ is its
characteristic work scale and $x$ measures this scale in units of the
thermal energy $1/\beta$. The relative magnitude of the two corrections
is
\begin{equation}
    \frac{\varepsilon^2I_{t,0}}{\Delta_{t,0}}
    \simeq
    8\frac{(\varepsilon/t)^2}{(\beta\Omega)^2}
    \frac{
        3x\sinh x+\frac{9}{2}x^2\cosh x
    }{
        \cosh x-\frac{\sinh x}{x}
    }.
\end{equation}
In the weak-driving regime $x\ll1$, both corrections depend quadratically
on the characteristic work scale. The quantum contribution dominates
when $\varepsilon/t\ll\beta|\Omega|$, while the imprecise-control correction dominates otherwise. By contrast, for $x\gg1$ the ratio grows quadratically with $x$, so resolving the quantum correction requires the increasingly stringent condition $\varepsilon/t\ll\beta|\Omega|/x$. Stronger thermodynamic driving therefore makes the fluctuation relation increasingly sensitive to timing imprecision relative to intrinsic quantum effects.

\section{Work in open quantum systems}\label{app:OQS}
In this section we derive the expressions for the reduced work operator used in the main text. In particular, we will give an exact form in the case of the first momentum operator $\hat W_A^{(1)}$ in the case of a time-independent interaction, and the full hierarchy of momenta $\hat W_A^{(n)}$ for general time-dependent interactions, in the GKLS limit. The calculations rely on a generalized quantum regression theorem we prove in Appendix~\ref{app:generalized_QRT}.

\subsection{First moment for time-independent interaction}\label{app:oqs_undriven}

Let's start from the time-independent case. Throughout, we assume the factorized initial state $\hat\rho_S(t_0) = \hat\rho_A(t_0)\otimes\hat\sigma_E$ with $[\hat\sigma_E,\hat H_E]=0$, and decompose the total Hamiltonian as
\begin{equation}\label{eq:app_Hsplit}
    \hat H_S(t) = \hat H_A(t)\otimes\mathbbm 1_E + \mathbbm 1_A\otimes \hat H_E + \hat V,
\end{equation}
with only $\hat H_A(t)$ time-dependent.

The work operator $\hat W(t,t_0) = \hat U^\dagger(t,t_0)\,\hat H_S(t)\,\hat U(t,t_0) - \hat H_S(t_0)$ admits the integral representation
\begin{equation}\label{eq:app_W_int}
    \hat W(t,t_0) = \int_{t_0}^t\!ds\,\hat U^\dagger(s,t_0)\big(\partial_s \hat H_A(s)\otimes\mathbbm 1_E\big)\hat U(s,t_0),
\end{equation}
obtained by Eq.~\eqref{eq:work_op_int} and Eq.~\eqref{eq:app_Hsplit}, or equivalently by $    \hat W(t,t_0)
    =
    \int_{t_0}^{t}\!ds\,
    \frac{d}{ds}
    \left[
        \hat U^\dagger(s,t_0)
        \hat H_S(s)
        \hat U(s,t_0)
    \right]$.

Consider the evolution map on the state of $A$ only
\begin{equation}
    \hat \rho_A(t) = \Phi_t[\hat \rho_A] = \text{Tr}_E\left[\hat U(t,t_0) \, \hat\rho_A(t_0) \otimes \hat \sigma_E \,\hat U^\dagger(t,t_0)\right].
\end{equation}
The Heisenberg-picture adjoint of this map is defined by $\Tr_A[\hat O_A\,\Phi_t[\hat\rho_A]] = \Tr_A[\Phi_t^\dagger[\hat O_A]\,\hat\rho_A]$ for all $\hat\rho_A,O_A$, and reads
\begin{equation}\label{eq:app_Phi_dag}
    \Phi_t^\dagger[\hat O_A] = \Tr_E\!\big[(\mathbbm 1_A\otimes\hat\sigma_E)\,\hat U^\dagger(t,t_0)(\hat O_A\otimes\mathbbm 1_E)\hat U(t,t_0)\big].
\end{equation}

The first-moment reduced work operator on the system $A$ alone can be conveniently written in terms of this adjoint map. In fact, by definition,
\begin{equation}\label{eq:app_W_A_1}
    \hat W_A^{(1)}(t) = \Tr_E\!\big[(\mathbbm 1_A\otimes\hat\sigma_E)\,\hat W(t,t_0)\big] = \int_{t_0}^tds\,\Phi_s^\dagger[\partial_s\hat H_A(s)],
\end{equation}
where we recognized each term in the integrand as the action of $\Phi_s^\dagger$ on the $A$-operator $\partial_s \hat H_A(s)$. We can apply the same reasoning to $\Delta\hat E(t,t_0)$ from Eq.~\eqref{eq:energy_operators}, obtaining
\begin{align}
    \Delta\hat E(t,t_0)
    &=
    \int_{t_0}^{t}\!ds\,
    \frac{d}{ds}
    \left[
        \hat U^\dagger(s,t_0)
        \bigl(\hat H_A(s)\otimes\mathbbm 1_E\bigr)
        \hat U(s,t_0)
    \right]
    \nonumber\\
    &=
    \int_{t_0}^{t}\!ds\,
    \hat U^\dagger(s,t_0)
    \left(
        \partial_s\hat H_A(s)\otimes\mathbbm 1_E
        +
        i[\hat H_S(s),\hat H_A(s)\otimes \mathbbm 1_E]
    \right)
    \hat U(s,t_0).
    \label{eq:DeltaE_integral}
\end{align}
Taking the first-moment reduction therefore yields
\begin{align}
    \Delta\hat E_A^{(1)}(t,t_0)
    &=
    \int_{t_0}^{t}\!ds\,
    \left(
        \Phi_s^\dagger[\partial_s\hat H_A(s)]
        +
        \dot\Phi_s^\dagger[\hat H_A(s)]
    \right)=
    \int_{t_0}^{t}\!ds\,
    \Phi_s^\dagger\!\left[
        \partial_s\hat H_A(s)
        +
        \mathcal L_s^\dagger[\hat H_A(s)]
    \right],
\end{align}
where in the last equality we used
$\dot\Phi_s^\dagger=\Phi_s^\dagger\circ\mathcal L_s^\dagger$. By comparing the expressions for $\hat W_A^{(1)}$ and $\Delta\hat E_A^{(1)}$, and applying the reduced first law of Eq.~\eqref{eq:reduced_first_law}, we find the heat operator
\begin{equation}
         \hat Q_A^{(1)}(t,t_0)= -\int_{t_0}^{t}\!ds\,
    \Phi_s^\dagger[\mathcal L_s^\dagger[\hat H_A(s)]].
\end{equation}
In the more general case where $\Delta \hat V$ is not negligible, other contributions would be present.

We now turn to the higher-order moment operators
\begin{equation}\label{eq:app_W_A_n}
    \hat W_A^{(n)}(t) = \Tr_E\!\big[(\mathbbm 1_A\otimes\hat\sigma_E)\,\hat W^n(t,t_0)\big].
\end{equation}
Expanding $\hat W^n$ using Eq.~\eqref{eq:app_W_int} gives
\begin{equation}\label{eq:app_W_n_expanded}
    \hat W_A^{(n)}(t) = \int_{[t_0,t]^n}\!\!ds_1\cdots ds_n\,\Tr_E\!\bigg[(\mathbbm 1_A\otimes\hat\sigma_E)\prod_{k=1}^n \hat U^\dagger(s_k,t_0)\big(\partial_{s_k} \hat H_A(s_k)\otimes\mathbbm 1_E\big)\hat U(s_k,t_0)\bigg],
\end{equation}
with the product ordered $k=1,2,\ldots,n$ from left to right. The integrand is an operator-valued multi-time correlation on $A$, whose expectation in $\hat\rho_A(t_0)$ gives the corresponding correlation function in $\hat\rho_S(t_0)=\hat\rho_A(t_0)\otimes\hat\sigma_E$.

Although the general expression for the higher work moments in
Eq.~\eqref{eq:app_W_n_expanded} does not admit substantial simplification, it
takes a more transparent form under standard assumptions. For this reason, we specialize
to the Markov--Lindblad limit.

\subsection{Review on GKLS limit} \label{app:Markov-Lindblad}
We now move on to the case where the external agent has control over the interaction between the $A$ and then environment $E$, so that the total Hamiltonian reads
\begin{equation}
    \hat H_S(t) = \hat H_A(t)\otimes \mathbbm 1_E + \mathbbm 1_A \otimes \hat H_E + \hat V(t).
\end{equation}

In order to study the form of the work moments, we first write down the equation of motion of the subsystem of interest $A$. We study this scenario in the more tractable Markov--Lindblad limit, starting from the initial factorized state $\hat\rho_S(t_0)=\hat\rho_A(t_0)\otimes\hat\sigma_E$ with $[\hat\sigma_E,\hat H_E]=0$. The (time-dependent) generator of the dynamical map can then be written in canonical form~\cite{Breuer07} as
\begin{equation}\label{eq:Vt_GKLS_canon}
    \mathcal L_t[\,\cdot\,] = -i[\hat H_A(t) + \hat H_{LS}(t),\,\cdot\,] + \sum_{kk'}\sum_{\omega\in\Omega(t)} \gamma_{kk'}(\omega)\left(\hat L_{k\omega}(t)\,\cdot\, \hat L_{k'\omega}^\dagger(t) - \tfrac{1}{2}\{ \hat L_{k'\omega}^\dagger(t) \hat L_{k\omega}(t),\,\cdot\,\}\right),
\end{equation}
where we defined the Lamb-shift Hamiltonian as
\begin{equation}\label{eq:Vt_HLS}
    \hat H_{LS}(t) = \sum_{kk'}\sum_{\omega\in\Omega(t)} S_{kk'}(\omega) \hat L_{k'\omega}^\dagger(t) \hat L_{k\omega}(t),
\end{equation}
and the instantaneous Bohr spectrum $\Omega(t)$ corresponding to the energy gaps in $\hat H_A(t)$: $\Omega(t):= \{\varepsilon'-\varepsilon\,|\,\varepsilon',\varepsilon \in \mathcal E(\hat H_A(t))\}$ with $\mathcal E(\hat H_A(t))$ the spectrum of $\hat H_A(t)$.

We now proceed to show how the rates $\gamma_{kk'}$, the shifts $S_{kk'}$, and jump operators $\hat L_{k\omega}$ are fixed by the microscopic derivation.

Without loss of generality, we expand the interaction on a time-independent basis, allowing us to carry all the time dependence into scalar couplings
\begin{equation}\label{eq:Vt_V}
    \hat V(t)=\sum_k\lambda_k(t)\,\hat A_k\otimes \hat E_k.
\end{equation}
We further make the standard assumption
\begin{equation}\label{eq:zeroMeanE}
    \qquad \Tr_E[\hat\sigma_E \hat E_k]=0,
\end{equation}
since a nonzero bath mean can always be absorbed into $\hat H_A(t)$. 

Both the rates and the shifts are defined from the bath correlation functions. Specifically, the interaction-picture bath operators are $\tilde E_k(t)= e^{i\hat H_Et}\hat E_k e^{-i\hat H_Et}$, and their correlation functions are $C_{kk'}(t-s)=\Tr_E[\hat\sigma_E\,\tilde E_k(t)\tilde E_{k'}(s)]$. Their half-Fourier transform are
\begin{equation}\label{eq:Vt_Gamma}
    \Gamma_{kk'}(\omega)=\int_0^\infty\!\!dt\,e^{i\omega t}\,C_{kk'}(t),
\end{equation}
and rates and shifts are defined as
\begin{align}
    \gamma_{kk'}(\omega)
    &=
    \Gamma_{kk'}(\omega)
    +
    \Gamma_{k'k}^{*}(\omega),
    \\
    S_{kk'}(\omega)
    &=
    \frac{1}{2i}
    \left[
        \Gamma_{kk'}(\omega)
        -
        \Gamma_{k'k}^{*}(\omega)
    \right],
\end{align}
where $\gamma_{kk'}(\omega)$ and $S_{kk'}(\omega)$ are Hermitian w.r.t. the indices $k$ and $k'$, and $\gamma_{kk'}(\omega)$ is positive semi-definite.

Finally, the jump operators are defined with the resolution of the system operators $\hat A_k$ in the instantaneous eigenspaces of $\hat H_A(t)$, 
\begin{equation}\label{eq:Vt_L}
    \hat L_{k\omega}(t)=\lambda_k(t)\hat A_k(\omega)~,\qquad \hat A_k(\omega) := \!\!\!\!\!\! \sum_{\substack{\varepsilon,\varepsilon'\in \mathcal E(\hat H_A(t))\\~\varepsilon'-\varepsilon = \omega}} \!\!\!\!\!\! \hat\Pi_\varepsilon \hat A_k\hat\Pi_{\varepsilon'}~,
\end{equation}
where $\hat\Pi_\varepsilon$ is the projector onto the eigenspace of $\varepsilon$, so that $\hat A_k=\sum_\omega \hat A_k(\omega)$ for all $t$ and $[\hat H_A(t),\hat A_k(\omega)]=-\omega \,\hat A_k(\omega)$. It is important to note that the values of $\omega$ and the operators $\hat A_k(\omega)$ are implicitly time-dependent when $\hat H_A(t)$ is driven. However, their contributions to work are all of higher order in bath correlation times, which means they can be neglected in this limit. Therefore, for our purposes, $\dot L_{k\omega}(t) = \partial_t \hat L_{k\omega}(t) = \dot \lambda_k(t) \hat A_k(\omega)$.

In practice the canonical generator~\eqref{eq:Vt_GKLS_canon} is used in its diagonalized form. Since the bath alone supplies $\gamma_{kk'}(\omega)$, this matrix carries no time dependence; diagonalizing it at each Bohr frequency,
\begin{equation}\label{eq:Vt_diag}
    \gamma_{kk'}(\omega)=\sum_\alpha g_\alpha(\omega)\,v^\alpha_k(\omega)\,v^{\alpha*}_{k'}(\omega)~,\qquad g_\alpha(\omega)\ge0,
\end{equation}
and defining $\hat J_{\alpha\omega}(t)=\sqrt{g_\alpha(\omega)}\sum_k v^\alpha_k(\omega)\,\hat L_{k\omega}(t)$, the dissipator becomes a single sum $\sum_{\alpha\omega}(\hat J_{\alpha\omega}\,\cdot\,\hat J_{\alpha\omega}^\dagger-\tfrac12\{\cdots\})$. The generator reads
\begin{equation}\label{eq:Vt_GKLS_diag}
    \mathcal L_t[\,\cdot\,]=-i[\hat H_A(t)+\hat H_{LS}(t),\,\cdot\,]+\sum_{\alpha\omega}\Big(\hat J_{\alpha\omega}(t)\,\cdot\,\hat J_{\alpha\omega}^\dagger(t)-\tfrac12\{\hat J_{\alpha\omega}^\dagger(t)\hat J_{\alpha\omega}(t),\,\cdot\,\}\Big).
\end{equation}
In principle this generator has a gauge symmetry: one could make a transformation on the jump operators that would change the Lamb-shift Hamiltonian (here it remains unchanged by the diagonalization), thus changing the Hamiltonian-dissipative split without changing the dynamics. However, once the microscopic description is fixed, this split is not a matter of convention: each gauge corresponds to a different physical setup where the system and environment interact differently, but the reduced dynamics are indistinguishable. Here, the gauge is fixed by the microscopic description of the interaction $\hat V(t) = \sum_k\lambda_k(t) \, \hat A_k\otimes \hat E_k$ which uniquely defines the jump operators $\hat L_{k\omega}$ (for a given Bohr-frequency) and therefore the corresponding diagonalization. 

\subsection{First moment for time-dependent interactions in the GKLS limit}\label{app:OQS_first_moment_driven}
For a time-dependent interaction, the work operator of Eq.~\eqref{eq:app_W_int} contains an additional contribution that in the GKLS regime reads
\begin{equation}
    \hat W_V(t)=\int_{t_0}^tds\,\hat U^\dagger(s,t_0)\partial_s \hat V(s)\hat U(s,t_0) = \sum_k \int_{t_0}^t ds \, \dot \lambda_k(s)\,\hat U^\dagger(s,t_0) \,(\hat A_k \otimes \hat E_k)\,U^\dagger(s,t_0).
\end{equation}

Its reduced first-moment operator is
\begin{equation}\label{eq:Vt_WV}
    \hat W_{V,A}^{(1)}(t) = \sum_k\int_{t_0}^t\!\!ds\,\dot\lambda_k(s)\,\Tr_E\!\big[(\mathbbm 1_A\otimes\hat\sigma_E)\,\hat U^\dagger(s,t_0)(\hat A_k\otimes \hat E_k)\hat U(s,t_0)\big].
    \end{equation}
such that $\Tr[\hat\rho_S(t)\hat W_V(t)] = \Tr[\hat\rho_A(t)\hat W_{V,A}^{(1)}(t)]$.

The partial trace in the integral has a very similar structure to the adjoint of the reduced dynamical map in Eq.~\eqref{eq:app_Phi_dag}, the only difference being the bath operator $\hat E_k$. We introduce the bath-tagged map
\begin{align}
    \Phi_{k,t}[\hat \rho_A] &= \Tr_E\!\big[\hat U(t,t_0)(\hat \rho_A\otimes\hat\sigma_E)\,\hat U^\dagger(t,t_0)(\mathbbm 1_A\otimes \hat E_k)\big], \label{eq:mod_evolution}\\
    \Phi_{k,t}^\dagger[\hat O_A] &= \Tr_E\!\big[(\mathbbm 1_A\otimes\hat\sigma_E)\,\hat U^\dagger(t,t_0)(\hat O_A\otimes \hat E_k)\hat U(t,t_0)\big].\label{eq:mod_evolution_adjoint}
\end{align}
Eq.~\eqref{eq:Vt_WV} then reads
\begin{equation}
    \hat W_{V,A}^{(1)}(t)
    =
    \sum_k\int_{t_0}^t ds\,
    \dot\lambda_k(s)\Phi_{k,s}^\dagger[\hat A_k].
\end{equation}

We now study the bath-tagged map $\Phi_{k,s}$ in the GKLS limit, and we will show that in the limit $\Phi^\dagger_{k,t} = \Phi^\dagger_{t}\circ \mathcal K^\dagger_{k,t}$, with $\mathcal K_{k,t}$ defined as 
\begin{align}
    \mathcal K_{k,t}[\hat \rho_A] &= -i\sum_{k'\omega} \lambda_{k'}(t)\big(\Gamma_{kk'}(\omega)\hat A_{k'}(\omega)\hat \rho_A - \Gamma^*_{kk'}(\omega)\hat \rho_A \hat A_{k'}^\dagger(\omega) \big),\\
    \mathcal K_{k,t}^\dagger[\hat O_A] &= i\sum_{k'\omega} \lambda_{k'}(t)\big(\Gamma_{kk'}^*(\omega) \hat A_{k'}^\dagger(\omega)\hat O_A - \Gamma_{kk'}(\omega)\hat O_A \hat A_{k'}(\omega) \big).
\end{align}

To prove this relation, we work in the interaction picture with respect to the free Hamiltonian
\begin{equation}
    \hat H_0(t)=\hat H_A(t)\otimes\mathbbm 1_E+\mathbbm 1_A\otimes\hat H_E.
\end{equation}
Its propagator factorizes as
\begin{equation}
    \hat U_0(t,t_0)=\hat U_A(t,t_0)\otimes\hat U_E(t,t_0),
\end{equation}
where
\begin{equation}
    \hat U_A(t,t_0)=\mathcal T e^{-i\int_{t_0}^t ds\,\hat H_A(s)},
    \qquad
    \hat U_E(t,t_0)=e^{-i\hat H_E(t-t_0)}.
\end{equation}
The full propagator is written as
\begin{equation}
    \hat U(t,t_0)=\hat U_0(t,t_0)\hat U_I(t,t_0),
    \qquad
    \hat U_I(t,t_0)=\mathcal T e^{-i\int_{t_0}^t ds\,\tilde V(s)}.
\end{equation}

For any joint operator $\hat O(t)$ and state $\hat\rho_S(t)$, their interaction-picture forms are
\begin{equation}
    \tilde O(t)=\hat U_0^\dagger(t,t_0)\hat O(t)\hat U_0(t,t_0),
    \qquad
    \tilde\rho_S(t)=\hat U_0^\dagger(t,t_0)\hat\rho_S(t)\hat U_0(t,t_0).
\end{equation}
In particular,
\begin{equation}
    \tilde O_A(t)=\hat U_A^\dagger(t,t_0)\hat O_A(t)\hat U_A(t,t_0),
    \qquad
    \tilde E_k(t)=\hat U_E^\dagger(t,t_0)\hat E_k\hat U_E(t,t_0).
\end{equation}
The interaction-picture state then satisfies
\begin{equation}
    \frac{d}{dt}\tilde\rho_S(t)=-i[\tilde V(t),\tilde\rho_S(t)],
    \qquad
    \tilde\rho_S(t)=\hat U_I(t,t_0)\hat\rho_S(t_0)\hat U_I^\dagger(t,t_0).
\end{equation}

Since $\hat U_A(t_0,t_0)=\mathbbm 1_A$ and  $[\hat H_E,\hat\sigma_E]=0$, in interaction picture the maps $\Phi_t$ and $\Phi_{k,t}$ become 
\begin{align}
    \tilde\Phi_t[\hat\rho_A]
    &=
    \Tr_E\!\big[\hat U_I(t,t_0)(\hat\rho_A\otimes\hat\sigma_E)\hat U_I^\dagger(t,t_0)\big]=\hat U_A^\dagger(t,t_0)\Phi_t[\hat\rho_A]\hat U_A(t,t_0),\\
    \tilde\Phi_{k,t}[\hat\rho_A]
    &=\Tr_E\!\big[\hat U_I(t,t_0)(\hat\rho_A\otimes\hat\sigma_E)\hat U_I^\dagger(t,t_0)(\mathbbm 1_A\otimes\tilde E_k(t))\big]=\hat U_A^\dagger(t,t_0)\Phi_{k,t}[\hat\rho_A]\hat U_A(t,t_0).
\end{align}

We now perform the Born approximation. Because the interaction generates system--environment correlations, in general $\tilde\rho_S(t)\neq\tilde\rho_A(t)\otimes\hat\sigma_E$. We define the projective superoperators
\begin{equation}
    \mathbb P[\,\cdot\,]=\Tr_E[\,\cdot\,]\otimes\hat\sigma_E,
    \qquad
    \mathbb Q=\mathcal I-\mathbb P,
\end{equation}
so that
\begin{equation}
    \mathbb P[\tilde\rho_S(t)]=\tilde\rho_A(t)\otimes\hat\sigma_E,
    \qquad
    \tilde\rho_S(t)=\mathbb P[\tilde\rho_S(t)]+\mathbb Q[\tilde\rho_S(t)].
\end{equation}
Applying $\mathbb Q$ to the interaction-picture equation of motion gives the exact relation
\begin{align}
    \frac{d}{dt}\mathbb Q[\tilde\rho_S(t)]
    ={}&-i\mathbb Q[\tilde V(t),\tilde\rho_A(t)\otimes\hat\sigma_E]\nonumber\\
    &-i\mathbb Q[\tilde V(t),\mathbb Q[\tilde\rho_S(t)]].
\end{align}
The Born approximation consists in neglecting the second term. Moreover, the centering condition~\eqref{eq:zeroMeanE} implies
\begin{equation}
    \mathbb P[\tilde V(t),\tilde\rho_A(t)\otimes\hat\sigma_E]=0,
\end{equation}
and hence $\mathbb Q$ acts as the identity on the remaining commutator. Since the initial state is factorized, $\mathbb Q[\tilde\rho_S(t_0)]=0$, and therefore
\begin{equation}\label{eq:evolution_Q}
    \mathbb Q[\tilde\rho_S(t)]\simeq-i\int_{t_0}^t ds\,[\tilde V(s),\tilde\rho_A(s)\otimes\hat\sigma_E].
\end{equation}

We can now evaluate the bath-tagged map directly. Its contribution from $\mathbb P[\tilde\rho_S(t)]$ vanishes because the bath is centered:
\begin{equation}\label{eq:E_cancellation}
    \Tr_E\!\big[\mathbb P[\tilde\rho_S(t)](\mathbbm 1_A\otimes\tilde E_k(t))\big]
    =\tilde\rho_A(t)\Tr_E[\hat\sigma_E\tilde E_k(t)]=0.
\end{equation}
Thus, for the initial state $\hat\rho_A(t_0)$,
\begin{equation}
    \tilde\Phi_{k,t}[\hat\rho_A(t_0)]
    =\Tr_E\!\big[\mathbb Q[\tilde\rho_S(t)](\mathbbm 1_A\otimes\tilde E_k(t))\big].
\end{equation}
Substituting Eq.~\eqref{eq:evolution_Q} gives
\begin{align}
    \tilde\Phi_{k,t}[\hat\rho_A(t_0)] &= -i\int_{t_0}^t ds\,\Tr_E\!\big[[\tilde V(s),\tilde\rho_A(s)\otimes\hat\sigma_E](\mathbbm 1_A\otimes\tilde E_k(t))\big]\\
    &= -i\sum_{k'}\int_{t_0}^t ds\,\lambda_{k'}(s)\Tr_E\!\big[\big[\tilde A_{k'}(s)\otimes\tilde E_{k'}(s),\tilde\rho_A(s)\otimes\hat\sigma_E\big](\mathbbm 1_A\otimes\tilde E_k(t))\big]\\
    &= -i\sum_{k'}\int_{t_0}^t ds\,\lambda_{k'}(s)\big(C_{kk'}(t-s)\tilde A_{k'}(s)\tilde\rho_A(s)-C_{kk'}^*(t-s)\tilde\rho_A(s)\tilde A_{k'}(s)\big)\\
    &\simeq -i\sum_{k'\omega}\lambda_{k'}(t)\int_0^\infty dr\,e^{i\omega r}\big(C_{kk'}(r)\tilde A_{k'}(\omega)\tilde\rho_A(t)-C_{kk'}^*(r)\tilde\rho_A(t)\tilde A_{k'}(\omega)\big)\\
    &= -i\sum_{k'\omega}\lambda_{k'}(t)\big(\Gamma_{kk'}(\omega)\tilde A_{k'}(\omega)\tilde\rho_A(t)-\Gamma_{kk'}^*(\omega)\tilde\rho_A(t)\tilde A_{k'}^\dagger(\omega)\big)\\
    &\equiv \tilde{\mathcal K}_{k,t}[\tilde\rho_A(t)].
\end{align}
where we have employed the Markov approximation. Since $\tilde\rho_A(t)=\tilde\Phi_t[\hat\rho_A(t_0)]$ and the initial state is arbitrary, returning to the Schroedinger picture gives
\begin{align}
    \Phi_{k,t}&\simeq\mathcal K_{k,t}\circ\Phi_t,\\
    \Phi_{k,t}^\dagger&\simeq\Phi_t^\dagger\circ\mathcal K_{k,t}^\dagger.\label{eq:mod_QRT}
\end{align}

Therefore, the reduced operator for the expected work value in the GKLS limit becomes
\begin{equation}\label{eq:work_first_moment_OQS}
    \hat W_A^{(1)}(t) = \int_{t_0}^t ds \,\Phi_s^\dagger \left[\partial_s \hat H_A(s) + \mathcal{K}_{k,s}^\dagger[\dot \lambda_k(s) \hat A_k]\right],
\end{equation}
and correspondingly the work current is
\begin{equation}\label{eq:power_GKLS}
\hat w^{(1)}(t) = \partial_t \hat H_A(t) +\sum_k\mathcal K_{k,t}^\dagger[\dot\lambda_k(t) \hat A_k].
\end{equation}
We now expand $\mathcal{K}_{k,s}^\dagger$ to obtain a more explicit form of the work operator:
 \begin{align}
	\hat W_{V,A}^{(1)}(t) &= \sum_k\int_{t_0}^t\!ds~\dot\lambda_{k}(s)\,\Phi^\dagger_{k,s}[\hat A_k],\\
	&= \sum_k\int_{t_0}^t\!ds~ \dot\lambda_{k}(s)\,\Phi^\dagger_{s}[\mathcal K_{k,s}^\dagger[\hat A_k]],\\
	&\approx\int_{t_0}^t\!ds~ \Phi^\dagger_{s}[i\sum_{kk'\omega} \dot\lambda_{k}\lambda_{k'}\big(\Gamma_{kk'}^*(\omega) \hat A_{k'}^\dagger(\omega)\hat A_k(\omega) - \Gamma_{kk'}(\omega)\hat A_k^\dagger(\omega) \hat A_{k'}(\omega) \big)],\\
	&=\int_{t_0}^t\!ds~ \Phi^\dagger_{s}[i\sum_{kk'\omega} \big(\dot\lambda_{k'}\lambda_{k}\Gamma_{k'k}^*(\omega) -\dot\lambda_{k}\lambda_{k'} \Gamma_{kk'}(\omega)\big) \hat A_{k}^\dagger(\omega)\hat A_{k'}(\omega)],\\
	&=\int_{t_0}^t\!ds~ \Phi^\dagger_{s}[\sum_{kk'\omega} \big(\tfrac{i}{2}\gamma_{kk'}(\omega)(\dot\lambda_{k'}\lambda_{k} - \dot\lambda_{k}\lambda_{k'}) + S_{kk'}(\omega)(\dot\lambda_{k'}\lambda_{k} + \dot\lambda_{k}\lambda_{k'})\big) \hat A_{k}^\dagger(\omega)\hat A_{k'}(\omega)],\\
	&=\int_{t_0}^t\!ds~ \Phi^\dagger_{s}[\partial_s \hat H^{(T)}_{LS}(s) + \tfrac{i}{2}\!\sum_{kk'\omega}\gamma_{kk'}(\omega)(\hat L_{k\omega}^\dagger(s) (\partial_s \hat L_{k'\omega}(s)) - (\partial_s \hat L_{k\omega}^\dagger(s)) \hat L_{k'\omega}(s))],
\end{align}
where the third line follows from the Secular approximation, with the time arguments of the $\lambda_k$s suppressed for readability, while the last line is expressed in terms of the diagonalized jump operators defined in Eq.~\eqref{eq:Vt_GKLS_diag}, and 
\begin{equation}
\hat H_{LS}^{(T)}(t) = \sum_{kk'\omega} S_{k'k}(\omega)\hat L_{k'\omega}^\dagger(t)\hat L_{k\omega}(t)  
\end{equation}
is the $S$-transposed Lamb-shift Hamiltonian.

Therefore, the work operator in Eq.~\eqref{eq:work_first_moment_OQS} becomes
\begin{equation}
    \hat W_A^{(1)}(t) = \int_{t_0}^t\!ds~\Phi_s^\dagger\left[\partial_s \hat H_A(s) + \partial_s \hat H_{LS}^{(T)}(s) + \frac{i}{2}\!\sum_{kk'\omega}\gamma_{kk'}(\omega)\!\left(\hat L_{k\omega}^\dagger(s) (\partial_s \hat L_{k'\omega}(s)) - (\partial_s \hat L_{k\omega}^\dagger(s)) \hat L_{k'\omega}(s)\right)\right],
\end{equation}
and consequently the work current is
\begin{equation}\label{eq:app_work_GKLS}
    \hat{w}^{(1)}(t) = \partial_t \hat H_A(t) + \partial_t \hat H_{LS}^{(T)}(t) + \frac{i}{2}\!\sum_{kk'\omega}\gamma_{kk'}(\omega)\!\left(\hat L_{k\omega}^\dagger(s) (\partial_s \hat L_{k'\omega}(s)) - (\partial_s \hat L_{k\omega}^\dagger(s)) \hat L_{k'\omega}(s)\right).
\end{equation}

We can apply the same reasoning to calculate the other currents
\begin{align}\nonumber
    \hat e(t) &:= \partial_t\hat H_A(t)\otimes\mathbbm 1_{E} + i[\hat H_S(t),\hat H_A(t)\otimes\mathbbm 1_{E}],\\
    \nonumber
    \hat q(t) &:=  i[\hat H_S(t),\mathbbm 1_{A}\otimes\hat H_E],\\
    \nonumber
    \hat v(t) &:= \partial_t\hat V(t) + i[\hat H_S(t), \hat V(t)].
\end{align}
obtaining
\begin{equation}\label{eq:reduced_currents_GKLS}
\begin{aligned}
    \hat e_A^{(1)}(t) &= \partial_t \hat H_A(t) + \mathcal L_t^\dagger[\hat H_A(t)], \\
    \hat q_A^{(1)}(t)\! &=\! -\mathcal L_t^\dagger[\hat H_A(t)] -\partial_t \hat H_{LS}^{(T)}(t) + \frac{i}{2}\!\sum_{{kk'\omega}}\!\!\gamma_{kk'}(\omega)\left(\hat L_{k\omega}^\dagger(s) (\partial_s \hat L_{k'\omega}(s)) - (\partial_s \hat L_{k\omega}^\dagger(s)) \hat L_{k'\omega}(s)\right),\\
    \hat v_A^{(1)}(t) &= 2\partial_t\hat H_{LS}^{(T)}(t).
\end{aligned}
\end{equation}

Naturally, the conservation of energy current holds
\begin{equation}
    \hat w^{(1)}_A(t)= \hat e^{(1)}_A(t) + \hat q^{(1)}_A(t) + \hat v^{(1)}_A(t).
\end{equation}
The decomposition in Eq.~\eqref{eq:reduced_currents_GKLS} shows that the first term of the work current in Eq.~\eqref{eq:app_work_GKLS} contributes to the internal energy, the Lamb-shit goes into the interaction and heat currents and the $\gamma$ term into heat only. A first-law at the level of current can be found only in the undriven case $\partial_t \hat V(t)=0$, when $\hat v_A^{(1)}(t)=0$.

\subsection{Variance for time-dependent interactions in the GKLS limit}\label{app:OQS_second_moment_driven}
We are now interested in finding the work second moment. By
Eq.~\eqref{eq:work_reduction_moment}, its expression is 
\begin{align}
    \hat W_A^{(2)}(t,t_0)
    &=
    \Tr_E\!\left[
        (\mathbbm 1_A\otimes\hat\sigma_E)
        \hat W(t,t_0)^2
    \right]
    \nonumber\\
    &=
    \int_{t_0}^{t}\!dt_1dt_2\,
    \Big(
        \hat w_{AA}(t_1,t_2)
        +
        \hat w_{AV}(t_1,t_2)
        +
        \hat w_{VA}(t_1,t_2)
        +
        \hat w_{VV}(t_1,t_2)
    \Big),
\end{align}
where
\begin{equation}
    \hat w_{XY}(t_1,t_2)
    =
    \Tr_E\!\left[
        (\mathbbm 1_A\otimes\hat\sigma_E)
        \hat U^\dagger(t_1,t_0)
        \hat G_X(t_1)
        \hat U(t_1,t_0)
        \hat U^\dagger(t_2,t_0)
        \hat G_Y(t_2)
        \hat U(t_2,t_0)
    \right],
\end{equation}
with
\begin{equation}
    \hat G_A(t)
    =
    \partial_t\hat H_A(t)\otimes\mathbbm 1_E,
    \qquad
    \hat G_V(t)
    =
    \sum_k
    \dot\lambda_k(t)\hat A_k\otimes\hat E_k.
\end{equation}
In order to calculate this trace of ordered product of operators $\hat U_i \hat G_i \hat U_i^\dagger$, we need a generalized version of the quantum regression theorem. We state and prove this generalization as Theorem~\ref{thm:generalized_QRT} in Appendix~\ref{app:generalized_QRT}.

For compactness, in the following we use $\hat h_i:=\partial_{\tau_i}\hat H_A(\tau_i)$, and $\hat a_{k,i}:=\dot\lambda_k(\tau_i)\hat A_k$. We split the integration square into its two chronologically ordered sectors,
\begin{equation}
    \int_{t_0}^{t}\!dt_1dt_2\,F(t_1,t_2)
    =
    \int_{t_0}^t\!d\tau_2\!\int_{t_0}^{\tau_2}\!d\tau_1
    \left[
        F(\tau_1,\tau_2)
        +
        F(\tau_2,\tau_1)
    \right].
\end{equation}
The two terms correspond to the two possible product orders of the operators evaluated at the chronologically ordered times $\tau_1<\tau_2$.

We first consider the $AA$ contribution. In the chronologically ordered product sector, the ordinary QRT (as a particular case of Theorem~\ref{thm:generalized_QRT}) gives
\begin{equation}
\label{eq:w_AA_ordered}
    \hat w_{AA}(\tau_1,\tau_2)
    \simeq
    \Phi_{\tau_1,t_0}^\dagger
    \left[
        \hat h_1
        \Phi_{\tau_2,\tau_1}^\dagger[\hat h_2]
    \right].
\end{equation}
Consider now the opposite product order,
$\hat w_{AA}(\tau_2,\tau_1)$. The corresponding recursion
intervals overlap over $[\tau_1,\tau_2]$.
The modified QRT gives an ordinary regression term and an
out-of-time-order boundary correction.  Here the only overlapping
pair is $(q,r)=(1,2)$, with
$I_1\cap I_2=[\tau_1,\tau_2]$ and
$\epsilon_1\epsilon_2=-1$.  Therefore,
\begin{align}
\label{eq:w_AA_reversed_final}
    \hat w_{AA}(\tau_2,\tau_1)
    &\simeq
    \Phi_{\tau_2,t_0}^\dagger
    \left[
        \hat h_2
        \Phi_{\tau_1,\tau_2}^\dagger[\hat h_1]
    \right] -
    \left.
    \Phi_{s,t_0}^\dagger
    \left[
        \Phi_{\tau_2,s}^\dagger[\hat h_2]
        \Phi_{\tau_1,s}^\dagger[\hat h_1]
    \right]
    \right|_{s=\tau_1}^{s=\tau_2}
    \nonumber\\
    &=
    \Phi_{\tau_1,t_0}^\dagger
    \left[
        \Phi_{\tau_2,\tau_1}^\dagger[\hat h_2]
        \hat h_1
    \right].
\end{align}
The upper-boundary contribution exactly cancels the ordinary regression term, while the lower-boundary contribution produces the oppositely oriented regression chain. Thus, the finite-overlap correction converts the formally backward ordinary regression chain into a forward-propagating expression with right multiplication by the earlier power operator.

Combining Eqs.~\eqref{eq:w_AA_ordered} and
\eqref{eq:w_AA_reversed_final}, we finally obtain
\begin{align}
\label{eq:AA_second_moment}
    \hat W_{AA}^{(2)}(t,t_0)
    \simeq{}&
    \int_{t_0}^{t}\!d\tau_2
    \int_{t_0}^{\tau_2}\!d\tau_1\,
    \Phi_{\tau_1,t_0}^\dagger
    \left[
        \hat h_1
        \Phi_{\tau_2,\tau_1}^\dagger[\hat h_2]
        +
        \Phi_{\tau_2,\tau_1}^\dagger[\hat h_2]
        \hat h_1
    \right]
    \nonumber\\
    ={}&
    \int_{t_0}^{t}\!d\tau_2
    \int_{t_0}^{\tau_2}\!d\tau_1\,
    \Phi_{\tau_1,t_0}^\dagger
    \left[
        \left\{
            \hat h_1,
            \Phi_{\tau_2,\tau_1}^\dagger[\hat h_2]
        \right\}
    \right].
\end{align}

It is worth noting that the second product-order sector could equivalently have been obtained without explicitly evaluating the out-of-order (from now on, OO) term. Indeed, for $\tau_1<\tau_2$, one may start the QRT recursion from the rightmost operator in the product, which is the operator at the earliest physical time $\tau_1$, rather than from the leftmost one. The resulting propagation then remains forward in time throughout and directly gives $\hat w_{AA}(\tau_2,\tau_1) \simeq \Phi_{\tau_1,t_0}^\dagger \left[\Phi_{\tau_2,\tau_1}^\dagger[\hat h_2] \hat h_1 \right]$. The explicit calculation above shows that the OO correction is precisely what makes the left-started and right-started regression procedures equivalent.

We next consider the two mixed contributions.  With $\hat w_V(t) = \sum_k \mathcal K_{k,t}^\dagger [\dot\lambda_k(t)\hat A_k]$, for the chronologically ordered product sectors, the generalized QRT gives
\begin{align}
    \hat w_{AV}(\tau_1,\tau_2)
    &\simeq
    \Phi_{\tau_1,t_0}^\dagger
    \left[
        \hat h_1
        \Phi_{\tau_2,\tau_1}^\dagger[
            \hat w_V(\tau_2)
        ]
    \right],
    \label{eq:w_AV_ordered}\\
    \hat w_{VA}(\tau_1,\tau_2)
    &\simeq
    \Phi_{\tau_1,t_0}^\dagger
    \Bigg[
        \sum_k
        \mathcal K_{k,\tau_1}^\dagger
        \left[
            \hat a_{k,1}
            \Phi_{\tau_2,\tau_1}^\dagger[\hat h_2]
        \right]
        +
        \sum_k
        \hat a_{k,1}
        \mathcal R_{k,\tau_1}^\dagger
        \left[
            \Phi_{\tau_2,\tau_1}^\dagger[\hat h_2]
        \right]
    \Bigg],
    \label{eq:w_VA_ordered}
\end{align}
where we used $\mathcal R_{k,t}^\dagger[\mathbbm 1_A]=0$.

Consider now the reversed product order
$\hat w_{VA}(\tau_2,\tau_1)$. The product-order sequence is
$(k_1,k_2)=(2,1)$, and the unique environment insertion occurs at
position $p=1$, at the physical time $t_{k_p}=\tau_2$. Since
$n=2$, the sets entering the $m=1$ out-of-time-order correction are
empty: $\mathcal O_{<}^{(1)} = \varnothing$ and $\mathcal O_{>}^{(1)} = \varnothing$.
There is therefore no separate out-of-time-order contribution in
this sector. Starting the recursion from the leftmost operator gives
the formally backward expression
\begin{align}
    \hat w_{VA}(\tau_2,\tau_1)
    \simeq
    \Phi_{\tau_2,t_0}^\dagger
    \Bigg[
        \sum_k
        \mathcal K_{k,\tau_2}^\dagger
        \left[
            \hat a_{k,2}
            \Phi_{\tau_1,\tau_2}^\dagger[\hat h_1]
        \right]
        +
        \sum_k
        \hat a_{k,2}
        \mathcal R_{k,\tau_2}^\dagger
        \left[
            \Phi_{\tau_1,\tau_2}^\dagger[\hat h_1]
        \right]
    \Bigg].
    \label{eq:w_VA_reversed_backward}
\end{align}
The backward recursion interval is adjacent to the environment
insertion and is therefore already accounted for by the ordinary
$\mathcal R^\dagger$ branch.

For the reversed word $\hat w_{AV}(\tau_2,\tau_1)$, the
product-order sequence is again $(k_1,k_2)=(2,1)$, but the unique
environment insertion now occurs at position $p=2$, at the physical
time $t_{k_p}=\tau_1$. The set
$\mathcal O_{>}^{(2)}$ is empty, whereas the first recursion interval
contains the environment-insertion time: $\mathcal O_{<}^{(2)} = \{1\}$ and $\mathcal O_{>}^{(2)} = \varnothing$.
Indeed, the corresponding physical recursion interval is
$I_1=[t_0,\tau_2]$, which contains $\tau_1$. The ordinary backward
regression chain must therefore be supplemented by the
$\mathcal K^{\mathrm{OO}\dagger}$ contribution of the $m=1$
generalized regression theorem:
\begin{align}
    \hat w_{AV}(\tau_2,\tau_1)
    &\simeq
    \Phi_{\tau_2,t_0}^\dagger
    \left[
        \hat h_2
        \Phi_{\tau_1,\tau_2}^\dagger[
            \hat w_V(\tau_1)
        ]
    \right]
    +
    \delta\hat w_{AV}^{\rm OO}(\tau_2,\tau_1),
    \label{eq:w_AV_reversed_before_OO}
\end{align}
with 
\begin{align}
    \delta\hat w_{AV}^{\rm OO}(\tau_2,\tau_1)
    =
    -i\sum_{kj\omega}
    \Phi_{\tau_1,t_0}^\dagger
    \Big[
        \gamma_{kj}(\omega)
        \Phi_{\tau_2,\tau_1}^\dagger[\hat h_2]
        \hat a_{k,1}
        \hat L_{j\omega}(\tau_1)
        -
        \gamma_{kj}^{*}(\omega)
        \hat L_{j\omega}^\dagger(\tau_1)
        \Phi_{\tau_2,\tau_1}^\dagger[\hat h_2]
        \hat a_{k,1}
    \Big].
    \label{eq:w_AV_OO_expanded}
\end{align}
However, we can avoid the out-of-time-order contribution by starting the recursion from the rightmost operator in the product. This keeps all reduced propagations forward in physical time.

More generally, the two possible positions of an interaction-power
insertion relative to an arbitrary future system operator $\hat O$
lead to two operator identities. When the interaction insertion lies
to the left of $\hat O$, the generalized regression rule gives
\begin{equation}
    \mathcal V_{L,t}^\dagger[\hat O]
    :=
    \sum_k
    \mathcal K_{k,t}^\dagger
    \left[
        \hat a_k(t)\hat O
    \right]
    +
    \sum_k
    \hat a_k(t)
    \mathcal R_{k,t}^\dagger[\hat O].
    \label{eq:V_left_insertion_identity}
\end{equation}
When instead the interaction insertion lies to the right of
$\hat O$, starting the recursion from the rightmost operator gives
\begin{equation}
    \mathcal V_{R,t}^\dagger[\hat O]
    :=
    \sum_k
    \mathcal K_{k,t}^\dagger
    \left[
        \hat O\hat a_k(t)
    \right]
    -
    \sum_k
    \hat a_k(t)
    \mathcal R_{k,t}^\dagger[\hat O].
    \label{eq:V_right_insertion_identity}
\end{equation}
The opposite signs of the $\mathcal R^\dagger$ contributions keep
track of whether the interaction insertion lies to the left or to
the right of the future system word. Both maps reduce to the
interaction-work operator when acting on the identity:
\begin{equation}
    \mathcal V_{L,t}^\dagger[\mathbbm 1_A]
    =
    \mathcal V_{R,t}^\dagger[\mathbbm 1_A]
    =
    \sum_k
    \mathcal K_{k,t}^\dagger[\hat a_k(t)]
    =
    \hat w_V(t),
    \label{eq:V_insertion_identity_action}
\end{equation}
where we used
$\mathcal R_{k,t}^\dagger[\mathbbm 1_A]=0$. Moreover, summing the
two insertion identities eliminates the explicit
$\mathcal R^\dagger$ terms:
\begin{equation}
    \mathcal V_{L,t}^\dagger[\hat O]
    +
    \mathcal V_{R,t}^\dagger[\hat O]
    =
    \sum_k
    \mathcal K_{k,t}^\dagger
    \left[
        \left\{
            \hat a_k(t),\hat O
        \right\}
    \right].
    \label{eq:V_left_right_sum_identity}
\end{equation}

These identities can now be applied directly to the two reversed
mixed sectors. For $\hat w_{VA}(\tau_2,\tau_1)$, the rightmost
operator is the system insertion $\hat h_1$. Starting the recursion
from this operator gives
\begin{equation}
    \hat w_{VA}(\tau_2,\tau_1)
    \simeq
    \Phi_{\tau_1,t_0}^\dagger
    \left[
        \Phi_{\tau_2,\tau_1}^\dagger[
            \hat w_V(\tau_2)
        ]
        \hat h_1
    \right].
    \label{eq:w_VA_reversed}
\end{equation}
For $\hat w_{AV}(\tau_2,\tau_1)$, the rightmost operator is instead
the interaction-power insertion at $\tau_1$. The later system
operator is first propagated back to $\tau_1$, after which the
interaction insertion acts in the right slot. Equation
\eqref{eq:V_right_insertion_identity} therefore gives directly
\begin{align}
    \hat w_{AV}(\tau_2,\tau_1)
    &\simeq
    \Phi_{\tau_1,t_0}^\dagger
    \Bigg[
        \sum_k
        \mathcal K_{k,\tau_1}^\dagger
        \left[
            \Phi_{\tau_2,\tau_1}^\dagger[\hat h_2]
            \hat a_{k,1}
        \right]
        -
        \sum_k
        \hat a_{k,1}
        \mathcal R_{k,\tau_1}^\dagger
        \left[
            \Phi_{\tau_2,\tau_1}^\dagger[\hat h_2]
        \right]
    \Bigg].
    \label{eq:w_AV_reversed}
\end{align}
The left- and right-started recursions provide two equivalent
representations of the same reversed product-order sectors. In
particular,
\begin{align}
    \Phi_{\tau_2,t_0}^\dagger
    \left[
        \hat h_2
        \Phi_{\tau_1,\tau_2}^\dagger[
            \hat w_V(\tau_1)
        ]
    \right]
    +
    \delta\hat w_{AV}^{\rm OO}(\tau_2,\tau_1)\simeq
    \Phi_{\tau_1,t_0}^\dagger
    \Bigg[
        \sum_k
        \mathcal K_{k,\tau_1}^\dagger
        \left[
            \Phi_{\tau_2,\tau_1}^\dagger[\hat h_2]
            \hat a_{k,1}
        \right]
        -
        \sum_k
        \hat a_{k,1}
        \mathcal R_{k,\tau_1}^\dagger
        \left[
            \Phi_{\tau_2,\tau_1}^\dagger[\hat h_2]
        \right]
    \Bigg],
    \label{eq:w_AV_reversal_identity}
\end{align}
whereas
\begin{align}
    \Phi_{\tau_1,t_0}^\dagger
    \left[
        \Phi_{\tau_2,\tau_1}^\dagger[
            \hat w_V(\tau_2)
        ]
        \hat h_1
    \right] \simeq \Phi_{\tau_2,t_0}^\dagger
    \Bigg[
        \sum_k
        \mathcal K_{k,\tau_2}^\dagger
        \left[
            \hat a_{k,2}
            \Phi_{\tau_1,\tau_2}^\dagger[\hat h_1]
        \right]
        +
        \sum_k
        \hat a_{k,2}
        \mathcal R_{k,\tau_2}^\dagger
        \left[
            \Phi_{\tau_1,\tau_2}^\dagger[\hat h_1]
        \right]
    \Bigg].
    \label{eq:w_VA_reversal_identity}
\end{align}
These identities express the equivalence of the two recursion
directions. Now, the four product-order sectors can now be recombined pairwise. First, Eqs.~\eqref{eq:w_AV_ordered} and \eqref{eq:w_VA_reversed} contain the interaction insertion at the later time $\tau_2$ and give
\begin{align}
    &\hat w_{AV}(\tau_1,\tau_2)
    +
    \hat w_{VA}(\tau_2,\tau_1) \simeq
    \Phi_{\tau_1,t_0}^\dagger
    \left[
        \hat h_1
        \Phi_{\tau_2,\tau_1}^\dagger[
            \hat w_V(\tau_2)
        ]
        +
        \Phi_{\tau_2,\tau_1}^\dagger[
            \hat w_V(\tau_2)
        ]
        \hat h_1
    \right].
    \label{eq:mixed_later_V_pair}
\end{align}
These are the two possible positions of the earlier system operator
$\hat h_1$ relative to the propagated interaction-work operator at
$\tau_2$.

The remaining two sectors,
Eqs.~\eqref{eq:w_VA_ordered} and
\eqref{eq:w_AV_reversed}, contain the interaction insertion at the
earlier time $\tau_1$. In terms of the left- and right-slot maps
defined in Eqs.~\eqref{eq:V_left_insertion_identity} and
\eqref{eq:V_right_insertion_identity}, they read
\begin{align}
    \hat w_{VA}(\tau_1,\tau_2)
    \simeq
    \Phi_{\tau_1,t_0}^\dagger
    \left[
        \mathcal V_{L,\tau_1}^\dagger
        \left[
            \Phi_{\tau_2,\tau_1}^\dagger[\hat h_2]
        \right]
    \right],
    \qquad
    \hat w_{AV}(\tau_2,\tau_1) \simeq
    \Phi_{\tau_1,t_0}^\dagger
    \left[
        \mathcal V_{R,\tau_1}^\dagger
        \left[
            \Phi_{\tau_2,\tau_1}^\dagger[\hat h_2]
        \right]
    \right].
\end{align}
Using the operator identity
\eqref{eq:V_left_right_sum_identity}, their sum becomes
\begin{align}
    \hat w_{VA}(\tau_1,\tau_2)
    +
    \hat w_{AV}(\tau_2,\tau_1) \simeq
    \Phi_{\tau_1,t_0}^\dagger
    \left[
        \sum_k
        \mathcal K_{k,\tau_1}^\dagger
        \left[
            \left\{
                \hat a_{k,1},
                \Phi_{\tau_2,\tau_1}^\dagger[\hat h_2]
            \right\}
        \right]
    \right].
    \label{eq:mixed_earlier_V_pair}
\end{align}
Indeed, the positive $\mathcal R^\dagger$ contribution in
Eq.~\eqref{eq:w_VA_ordered} cancels the negative
$\mathcal R^\dagger$ contribution in
Eq.~\eqref{eq:w_AV_reversed}. Combining Eqs.~\eqref{eq:mixed_later_V_pair} and
\eqref{eq:mixed_earlier_V_pair}, we finally obtain
\begin{multline}
    \hat W_{AV}^{(2)}(t,t_0)
    +
    \hat W_{VA}^{(2)}(t,t_0)
    \simeq 
    \int_{t_0}^{t}\!d\tau_2
    \int_{t_0}^{\tau_2}\!d\tau_1\,
    \Phi_{\tau_1,t_0}^\dagger
    \Bigg[
        \left\{\hat h_1,
        \Phi_{\tau_2,\tau_1}^\dagger[
            \hat w_V(\tau_2)
        ]
        \right\}
        +
        \sum_k
        \mathcal K_{k,\tau_1}^\dagger
        \left[
            \left\{
                \hat a_{k,1},
                \Phi_{\tau_2,\tau_1}^\dagger[\hat h_2]
            \right\}
        \right]
    \Bigg].
    \label{eq:mixed_second_moment}
\end{multline}

Finally, for the $VV$ sector, Theorem~\ref{thm:generalized_QRT} yields
\begin{align}
    \hat W_{VV}^{(2)}
    \simeq{}&
    \sum_{kk'}
    \int_{t_0}^{t}\!dt_1dt_2\,
    C_{kk'}(t_1-t_2)
    \Phi_{t_1,t_0}^\dagger\!
    \left[
        \hat a_{k,1}
        \Phi_{t_2,t_1}^\dagger[
            \hat a_{k',2}
        ]
    \right].
\end{align}
Applying the Markov and secular approximations gives
\begin{equation}
\label{eq:VV_local_noise}
    \hat W_{VV}^{(2)}
    \simeq
    \int_{t_0}^{t}\!ds\,
    \Phi_{s,t_0}^\dagger[\hat N(s)],
\end{equation}
where
\begin{align}
\label{eq:work_noise_operator}
    \hat N(s) =
    \sum_{kk'\omega}
    \gamma_{kk'}(\omega)
    \dot\lambda_k(s)\dot\lambda_{k'}(s)
    \hat A_k^\dagger(\omega)
    \hat A_{k'}(\omega) =
    \sum_{\alpha\omega}
    \partial_s\hat L_{\alpha\omega}^\dagger(s)
    \partial_s\hat L_{\alpha\omega}(s)
    \geq0.
\end{align}

We can now recognize a common structure in
Eqs.~\eqref{eq:AA_second_moment} and
\eqref{eq:mixed_second_moment}. Define
\begin{align}
\label{eq:WR_after_derivation}
    \mathfrak W_{R,t}^\dagger[\hat O]
    &:=
    \partial_t\hat H_A(t)\hat O
    +
    \sum_k
    \mathcal K_{k,t}^\dagger
    [\dot\lambda_k(t)\hat A_k\hat O],
    \\
\label{eq:WL_after_derivation}
    \mathfrak W_{L,t}^\dagger[\hat O]
    &:=
    \hat O\,\partial_t\hat H_A(t)
    +
    \sum_k
    \mathcal K_{k,t}^\dagger
    [\dot\lambda_k(t)\hat O\hat A_k].
\end{align}
They coincide on the identity,
\begin{equation}
    \hat w(t)
    :=
    \mathfrak W_{R,t}^\dagger[\mathbbm 1_A]
    =
    \mathfrak W_{L,t}^\dagger[\mathbbm 1_A]
    =
    \partial_t\hat H_A(t)+\hat w_V(t).
\end{equation}
The complete second moment can therefore be written as
\begin{align}
\label{eq:oqs_second_moment}
    \hat W_A^{(2)}(t,t_0)
    \simeq
    \int_{t_0}^{t}\!d\tau_2
    \int_{t_0}^{\tau_2}\!d\tau_1\,
    \Phi_{\tau_1,t_0}^\dagger
    \left[
        \left(
            \mathfrak W_{R,\tau_1}^\dagger
            +
            \mathfrak W_{L,\tau_1}^\dagger
        \right)
        \left[
            \Phi_{\tau_2,\tau_1}^\dagger[
                \hat w(\tau_2)
            ]
        \right]
    \right] +
    \int_{t_0}^{t}\!ds\,
    \Phi_{s,t_0}^\dagger[\hat N(s)]
    +
    \mathcal O(\tau_E^2).
\end{align}
The compact form generates terms containing two
$\mathcal K^\dagger$ maps, but these are understood to be
discarded consistently as $\mathcal O(\tau_E^2)$.

For an arbitrary initial state
$\hat\rho_A(t_0)$, with
$\hat\rho_A(s)=\Phi_{s,t_0}[\hat\rho_A(t_0)]$, the work variance is
\begin{align}
\label{eq:oqs_fluctuations}
    \sigma_W^2(t,t_0)
    =&
    2\int_{t_0}^{t}\!d\tau_2
    \int_{t_0}^{\tau_2}\!d\tau_1
    \Bigg(
        \Tr\!\left[
            \hat w(\tau_2)
            \Phi_{\tau_2,\tau_1}
            [
                \mathfrak W_{\tau_1}^\dagger
                [\hat\rho_A(\tau_1)]
            ]
        \right]
        -
        \Tr[
            \hat\rho_A(\tau_1)\hat w(\tau_1)
        ]
        \Tr[
            \hat\rho_A(\tau_2)\hat w(\tau_2)
        ]
    \Bigg)
    \nonumber\\
    &+
    \int_{t_0}^{t}\!ds\,
    \Tr[\hat\rho_A(s)\hat N(s)]
    +
    \mathcal O(\tau_E^2),
\end{align}
where $\mathfrak W_t^\dagger = \mathfrak W_{R,t}^\dagger+\mathfrak W_{L,t}^\dagger$. For $\partial_t \hat V(t) = 0$, we recover the form derived in Ref.~\cite{Miller2019}.

\subsection{Higher moments for time-dependent interactions in the GKLS limit}
\label{app:OQS_higher_moments}
Let
\begin{equation}
    \hat h_i:=\partial_{t_i}\hat H_A(t_i),
    \qquad
    \hat a_{\alpha,i} := \dot\lambda_\alpha(t_i) \hat A_k = \sum_{\omega} \partial_{t_i}\hat L_{\alpha\omega}(t_i).
\end{equation}
We also define the
system-only regression strings
\begin{equation}
\label{eq:higher_moments_system_regression_string}
    \mathcal G_{a:b}^{(0)\dagger}
    =
    \Phi_{t_a,t_{a-1}}^\dagger
    \circ\mathsf L_{\hat h_a}
    \circ\cdots\circ
    \Phi_{t_b,t_{b-1}}^\dagger
    \circ\mathsf L_{\hat h_b},
    \qquad
    \mathcal G_{a:b}^{(0)\dagger}=\operatorname{id}
    \quad\text{for }a>b.
\end{equation}
The reduced operator associated with the $n$-th work moment is
\begin{equation}
\label{eq:reduced_work_nth_moment}
    \hat W_A^{(n)}(t,t_0)
    =
    \sum_{\bm x\in\{A,V\}^n}
    \int_{t_0}^t\!\!dt_1...dt_n\,
    \hat w_{\bm x}(t_1,\ldots,t_n).
\end{equation}
The generalized regression theorem shows that, up to
$\mathcal O(\tau_E^2)$, only strings containing zero, one, or two
explicit environment insertions contribute. It is therefore natural
to decompose the $n$-th moment as
\begin{equation}
\label{eq:nth_moment_decomposition}
    \hat W_A^{(n)}(t,t_0)
    \simeq
    \hat W_{\rm reg}^{(n)}(t,t_0)
    +
    \hat W_N^{(n)}(t,t_0)
    +
    \hat W_{\rm OO}^{(n)}(t,t_0)
    +
    \mathcal O(\tau_E^2).
\end{equation}

\paragraph{Regular contribution.}
The regular part contains the ordinary $m=0$ regression term and
the regular part of every $m=1$ string. Applying
Theorem~\ref{thm:generalized_QRT} gives
\begin{align}
\label{eq:nth_moment_regular_explicit}
    \hat W_{\rm reg}^{(n)}
    ={}&
    \int_{t_0}^t\!\!dt_1...dt_n\,
    \Bigg\{
        \mathcal G_{1:n}^{(0)\dagger}[\mathbbm 1_A]+
        \sum_{p=1}^n\sum_\alpha
        \mathcal G_{1:p-1}^{(0)\dagger}
        \circ\Phi_{t_p,t_{p-1}}^\dagger
        \circ
        \left(
            \mathcal K_{\alpha,t_p}^\dagger
            \circ\mathsf L_{\hat a_{\alpha,p}}
            +
            \mathsf L_{\hat a_{\alpha,p}}
            \circ\mathcal R_{\alpha,t_p}^\dagger
        \right)
        \circ\mathcal G_{p+1:n}^{(0)\dagger}
        [\mathbbm 1_A]
    \Bigg\}.
\end{align}

This suggests defining the full-hypercube right work-insertion map
as
\begin{align}
\label{eq:full_WR_definition}
    \mathfrak W_{R,s}^\dagger[\hat O_A]
    :=
    \partial_s\hat H_A(s)\hat O_A +
    \sum_\alpha\dot\lambda_\alpha(s)
    \left(
        \mathcal K_{\alpha,s}^\dagger[
            \hat A_\alpha\hat O_A
        ]
        +
        \hat A_\alpha
        \mathcal R_{\alpha,s}^\dagger[\hat O_A]
    \right).
\end{align}
Since
$\mathcal R_{\alpha,s}^\dagger[\mathbbm 1_A]=0$, its action on
the identity gives the reduced first-moment work current,
\begin{equation}
    \mathfrak W_{R,s}^\dagger[\mathbbm 1_A]
    =
    \partial_s\hat H_A(s)
    +
    \sum_\alpha
    \dot\lambda_\alpha(s)
    \mathcal K_{\alpha,s}^\dagger[\hat A_\alpha] = \hat w(s).
\end{equation}

Expanding a product of the maps in
Eq.~\eqref{eq:full_WR_definition} and retaining at most one
interaction-induced $\mathcal K/\mathcal R$ insertion reproduces
Eq.~\eqref{eq:nth_moment_regular_explicit}. Hence,
\begin{equation}
\label{eq:nth_moment_regular_chain}
    \hat W_{\rm reg}^{(n)}
    \simeq
    \int_{t_0}^t\!\!dt_1...dt_n\,
    \prod_{i=1}^n
    \left(
        \Phi_{t_i,t_{i-1}}^\dagger
        \circ\mathfrak W_{R,t_i}^\dagger
    \right)
    [\mathbbm 1_A]
    +
    \mathcal O(\tau_E^2),
\end{equation}
where the product is ordered with increasing $i$ from left to
right. The compact product generates terms containing two or more
regular $\mathcal K/\mathcal R$ insertions, but these are
consistently discarded as $\mathcal O(\tau_E^2)$.

For later use, define the pullback of the work map to $t_0$,
\begin{equation}
\label{eq:pulled_back_WR}
    \overline{\mathfrak W}_{R,s}^\dagger
    =
    \Phi_{s,t_0}^\dagger
    \circ\mathfrak W_{R,s}^\dagger
    \circ(\Phi_{s,t_0}^\dagger)^{-1}.
\end{equation}
Using $\Phi_{t_i,t_{i-1}}^\dagger = \Phi_{t_0,t_{i-1}}^\dagger \circ\Phi_{t_i,t_0}^\dagger$
and the unitality of the adjoint propagator, the factors telescope, giving
\begin{equation}
\label{eq:nth_moment_regular_power}
    \hat W_{\rm reg}^{(n)}
    \simeq
    \left(
        \int_{t_0}^t ds\,
        \overline{\mathfrak W}_{R,s}^\dagger
    \right)^n
    [\mathbbm 1_A]
    +
    \mathcal O(\tau_E^2).
\end{equation}

\paragraph{Direct environment-noise contribution.}
The $m=2$ part of Theorem~\ref{thm:generalized_QRT} gives the
direct contraction of two explicit environment insertions. If these
occur at positions $a<b$, all remaining positions contain
$\hat h_i$. Therefore,
\begin{align}
\label{eq:nth_moment_noise_explicit}
    \hat W_N^{(n)}
    ={}&
    \sum_{1\leq a<b\leq n}
    \sum_{\alpha\beta}
    \int_{t_0}^t\!\!dt_1...dt_n\,
    C_{\alpha\beta}(t_a-t_b)
    \mathcal G_{1:a-1}^{(0)\dagger}
    \circ\Phi_{t_a,t_{a-1}}^\dagger
    \circ
    \mathsf L_{
        \hat a_{\alpha,a}
        \hat h_{a+1}\cdots
        \hat h_{b-1}
        \hat a_{\beta,b}
    }
    \circ
    \mathcal G_{b+1:n}^{(0)\dagger}
    [\mathbbm 1_A].
\end{align}
The product between the two contracted vertices is an ordinary
operator product: there are no reduced propagators between
positions $a$ and $b$. This is precisely the $m=2$ term in
Theorem~\ref{thm:generalized_QRT}. In the Markov--secular limit,
the short-ranged correlation $C_{\alpha\beta}$ may equivalently
be expressed through the corresponding local
$\gamma_{\alpha\beta}(\omega)$ contraction.

\paragraph{Out-of-time-order contribution.}
The out-of-time-order part contains two distinct terms,
\begin{equation}
\label{eq:nth_moment_OO_split}
    \hat W_{\rm OO}^{(n)}
    =
    \hat W_{{\rm OO},II}^{(n)}
    +
    \hat W_{{\rm OO},EI}^{(n)}.
\end{equation}

The $II$ contribution comes from the $m=0$ sector. Let $I_i=[\min\{t_{i-1},t_i\},\max\{t_{i-1},t_i\}]$, and $\epsilon_i=\operatorname{sgn}(t_i-t_{i-1})$, where $t_0$ denotes the initial time. Then
\begin{align}
\label{eq:nth_moment_OO_II}
    \hat W_{{\rm OO},II}^{(n)}
    ={}&
    \int_{t_0}^t\!\!dt_1...dt_n\!\!\!\!
    \sum_{\substack{1\leq q<r\leq n\\
                    |I_q\cap I_r|>0}}\!\!\!\!
    \epsilon_q\epsilon_r
    \left.
    \mathcal G_{1:q-1}^{(0)\dagger}
    \circ\Phi_{s,t_{q-1}}^\dagger\!\!
    \left[
        (\Phi_{t_q,s}^\dagger\circ\mathsf L_{\hat h_q}\circ\mathcal G_{q+1:r-1}^{(0)\dagger}[\mathbbm 1_A])
        (\Phi_{t_r,s}^\dagger\circ\mathsf L_{\hat h_r}\circ\mathcal G_{r+1:n}^{(0)\dagger}[\mathbbm 1_A])
    \right]
    \right|_{s=s_{q,r}^-}^{s=s_{q,r}^+},
\end{align}
where $[s_{q,r}^-,s_{q,r}^+]=I_q\cap I_r$.

The $EI$ contribution comes from the $m=1$ sector. For an
environment insertion at position $p$, define
\begin{equation}
    \mathcal O_<^{(p)}
    =
    \left\{
        q\in\{1,\ldots,p-1\}:t_p\in I_q
    \right\},
    \qquad
    \mathcal O_>^{(p)}
    =
    \left\{
        q\in\{p+2,\ldots,n\}:t_p\in I_q
    \right\}.
\end{equation}
The corresponding contribution is
\begin{multline}
\label{eq:nth_moment_OO_EI}
    \hat W_{{\rm OO},EI}^{(n)}
    =
    \sum_{p=1}^n\sum_\alpha
    \int_{t_0}^t\!\!dt_1...dt_n
    \Bigg\{
    \sum_{q\in\mathcal O_<^{(p)}}
    \epsilon_q\,
    \mathcal G_{1:q-1}^{(0)\dagger}
    \circ\Phi_{t_p,t_{q-1}}^\dagger
    \circ\mathcal K_{\alpha,t_p}^{\rm OO\dagger}
    \circ
    \mathsf L_{
        \hat h_q\cdots
        \hat h_{p-1}\hat a_{\alpha,p}
    }
    \circ\mathcal G_{p+1:n}^{(0)\dagger}
    [\mathbbm 1_A]\\
    +
    \sum_{q\in\mathcal O_>^{(p)}}
    \epsilon_q\,
    \mathcal G_{1:p-1}^{(0)\dagger}
    \circ\Phi_{t_p,t_{p-1}}^\dagger
    \circ
    \mathsf L_{
        \hat a_{\alpha,p}
        \hat h_{p+1}\cdots
        \hat h_{q-1}
    }
    \circ\mathcal R_{\alpha,t_p}^{\rm OO\dagger}\circ
    \Phi_{t_q,t_p}^\dagger
    \circ\mathsf L_{\hat h_q}
    \circ\mathcal G_{q+1:n}^{(0)\dagger}
    [\mathbbm 1_A]
    \Bigg\}.
\end{multline}
Equations~\eqref{eq:nth_moment_regular_explicit},
\eqref{eq:nth_moment_noise_explicit},
\eqref{eq:nth_moment_OO_II}, and
\eqref{eq:nth_moment_OO_EI} contain all contributions retained by
the generalized regression theorem.

\section{Generalized quantum regression theorem}
\label{app:generalized_QRT}
We now state and prove a generalized version of a quantum regression theorem~\cite{Lax63} that we apply to calculate the moments of the reduced work operator in Appendix~\ref{app:OQS}. We employ the notations and approximations of the GKLS limit defined in Appendix~\ref{app:Markov-Lindblad}.

\begin{theorem}[Generalized quantum regression theorem for system--environment insertions]
\label{thm:generalized_QRT}
Let
\begin{equation}
    \hat Z_{\bm k}(t_0) = \prod_{i=1}^n
    \hat U^{\dagger}(t_{k_i},t_0)
    \bigl(\hat X_i\otimes\hat Y_i\bigr)
    \hat U(t_{k_i},t_0),
\end{equation}
where $\hat Y_i\in \{\mathbbm 1_E,\hat E_{k_i}\}$, and $\{t_{k_1},\ldots,t_{k_n}\}$ an arbitrary collection of potentially disordered times. We denote by $m$ the number of operators $\hat Y_i$ different from $\mathbbm 1_E$. Define the reduced operator
\begin{equation}
    \hat z_{\bm k}(t_0) = \Tr_E[(\mathbbm 1_A\otimes\hat\sigma_E)\hat Z_{\bm k}(t_0)],
\end{equation}
such that $\Tr[\hat\rho_A\hat z_{\bm k}(t_0)] = \Tr[(\hat \rho_A\otimes\hat\sigma_E)\hat Z_{\bm k}(t_0)]$ for arbitrary $\hat \rho_A$.

Then, at the working order of the Born--Markov approximation,
$\hat z_{\bm k}(t_0)$ is given by the following four cases.
\paragraph{$m=0$. (Quantum Regression Theorem)} If $\hat Y_i=\mathbbm 1_E$ for every $i$,
\begin{equation}
\label{eq:generalized_QRT_m_zero}
    \hat z_{\bm k}(t_0) \simeq \Phi_{t_{k_1},t_0}^{\dagger} \circ\mathsf L_{\hat X_1} \circ\Phi_{t_{k_2},t_{k_1}}^{\dagger} \circ\mathsf L_{\hat X_2} \circ\cdots\circ\Phi_{t_{k_n},t_{k_{n-1}}}^{\dagger} \circ\mathsf L_{\hat X_n} [\mathbbm 1_A] + \delta \hat z_{QRT}^{\rm OO},
\end{equation}
where $\mathsf L_{\hat X}[\hat O_A]=\hat X\hat O_A$ is the
left-multiplication superoperator and $\Phi_{t,t'}^\dagger$ the adjoint of the GKLS propagator for two arbitrary times. The correction
$\delta\hat z_{\rm QRT}^{\rm OO}$ can be nonzero only when two
recursion intervals overlap over a finite physical-time interval.
In particular, it vanishes for a monotonically ordered sequence of
operator times, recovering the usual quantum regression theorem. \\

\paragraph{$m=1$.}
If the unique environment insertion occurs at position $p$, $\hat Y_p=\hat E_\alpha$, and $\hat Y_i=\mathbbm 1_E$ for $i\neq p$,
\begin{equation}
\label{eq:generalized_QRT_m_one}
    \hat z_{\bm k}(t_0)
    \simeq 
    \Phi_{t_{k_1},t_0}^{\dagger}\!\!
    \circ\mathsf L_{\hat X_1}\!\!
    \circ\cdots\circ
    \Phi_{t_{k_{p-1}},t_{k_{p-2}}}^{\dagger}\!\!
    \circ\mathsf L_{\hat X_{p-1}}\!\!\circ
    \Phi_{t_{k_p},t_{k_{p-1}}}^{\dagger}\!\!
    \circ
    \mathfrak E_{p,\alpha} \circ
    \Phi_{t_{k_{p+1}},t_{k_p}}^{\dagger}\!\!
    \circ\mathsf L_{\hat X_{p+1}}
    \circ\cdots\circ
    \Phi_{t_{k_n},t_{k_{n-1}}}^{\dagger}\!\!
    \circ\mathsf L_{\hat X_n}
    [\mathbbm 1_A] + \delta \hat z_1^{\rm OO},
\end{equation}
with $\delta \hat z_1^{\rm OO}$ the out-of-time-order term, $\mathfrak E_{p,\alpha} = \mathcal K_{\alpha,t_{k_p}}^{\dagger}\circ\mathsf L_{\hat X_p}\!\!+ \mathsf L_{\hat X_p}\!\!\circ\mathcal R_{\alpha,t_{k_p}}^{\dagger}$, $\mathcal R_{\alpha,t}^\dagger[\hat O_A] = i\sum_{j\omega} \Gamma_{j\alpha}^*(\omega)[\hat L_{j\omega}(t),\hat O_A]$, and  $\mathcal K_{\alpha,t}^\dagger[\hat O_A] = -i\sum_{j\omega}( \Gamma_{\alpha j}(\omega)\hat O_A \hat L_{j\omega}(t)- \Gamma^*_{\alpha j}(\omega)\hat L_{j\omega}^\dagger(t)\hat O_A)$.\\

\paragraph{$m=2$.}
If the two environment insertions occur at positions $a<b$, $\hat Y_a=\hat E_\alpha$, $\hat Y_b=\hat E_\beta$, and $\hat Y_i=\mathbbm 1_E$ for $i\neq a,b$,
\begin{equation}
\label{eq:generalized_QRT_m_two}
    \hat z_{\bm k}(t_0)
    \simeq 
    C_{\alpha\beta}
    \bigl(t_{k_a}-t_{k_b}\bigr)
    \Phi_{t_{k_1},t_0}^{\dagger}\!\!
    \circ\mathsf L_{\hat X_1}\!\!
    \circ\cdots\circ
    \Phi_{t_{k_a},t_{k_{a-1}}}^{\dagger}\circ
    \mathsf L_{a:b}
    \circ
    \Phi_{t_{k_{b+1}},t_{k_b}}^{\dagger}\!\!
    \circ\mathsf L_{\hat X_{b+1}}\!\!
    \circ\cdots\circ
    \Phi_{t_{k_n},t_{k_{n-1}}}^{\dagger}\!\!
    \circ\mathsf L_{\hat X_n}
    [\mathbbm 1_A],
\end{equation}
with $\mathsf L_{a:b} = \mathsf L_{\hat X_a}\circ\mathsf L_{\hat X_{a+1}}\circ\cdots\circ \mathsf L_{\hat X_b}$ and no out-of-time-order correction.\\

\paragraph{$m>2$.}
If the operator string contains more than two environment insertions,
\begin{equation}
\label{eq:generalized_QRT_m_greater_two}
    \hat z_{\bm k}(t_0)
    =
    \mathcal O(\tau_E^2),
\end{equation}
and is therefore negligible at the working order retained here.\\

\paragraph*{Out-of-time-order terms.} 
Let $I_i=[\min\{t_{k_{i-1}},t_{k_i}\},\max\{t_{k_{i-1}},t_{k_i}\}]$ (with $t_{k_0} = t_0$), and $\epsilon_i=\operatorname{sgn}(t_{k_i}-t_{k_{i-1}})$. We use $\hat X_{a:b}=\hat X_a\cdots\hat X_b$ and $\mathsf L_{a:b}=\mathsf L_{\hat X_a}\circ\cdots\circ\mathsf L_{\hat X_b}$, with both understood as the identity for $a>b$, and define
\begin{equation}
\label{eq:theorem_regression_string}
    \mathcal G_{a:b}^{\dagger}
    =
    \Phi_{t_{k_a},t_{k_{a-1}}}^{\dagger}
    \circ\mathsf L_{\hat X_a}
    \circ\cdots\circ
    \Phi_{t_{k_b},t_{k_{b-1}}}^{\dagger}
    \circ\mathsf L_{\hat X_b},
    \qquad
    \mathcal G_{a:b}^{\dagger}=\operatorname{id}
    \quad\text{for }a>b.
\end{equation}
Let $\mathcal P_{\rm ov} =\{(q,r):1\leq q<r\leq n,\ |I_q\cap I_r|>0\}$. The $m=0$ correction is
\begin{equation}
\label{eq:theorem_QRT_OO_correction}
    \delta\hat z_{\rm QRT}^{\rm OO}
    \simeq
    \sum_{(q,r)\in\mathcal P_{\rm ov}}
    \epsilon_q\epsilon_r
    \left.
    \mathcal G_{1:q-1}^{\dagger}
    \circ
    \Phi_{s,t_{k_{q-1}}}^{\dagger}
    \left[
        \left(
            \Phi_{t_{k_q},s}^{\dagger}
            \circ\mathsf L_{\hat X_q}
            \circ\mathcal G_{q+1:r-1}^{\dagger}
            [\mathbbm 1_A]
        \right)
        \left(
            \Phi_{t_{k_r},s}^{\dagger}
            \circ\mathsf L_{\hat X_r}
            \circ\mathcal G_{r+1:n}^{\dagger}
            [\mathbbm 1_A]
        \right)
    \right]
    \right|_{s=s_{q,r}^-}^{s=s_{q,r}^+},
\end{equation}
where $[s_{q,r}^-,s_{q,r}^+] = I_q\cap I_r$. In particular, $\delta\hat z_{\rm QRT}^{\rm OO}=0$ when $\mathcal P_{\rm ov}=\varnothing$.

Finally, let $\mathcal O^{(p)}_< =\{q=1,\ldots,p-1\,|\,t_{k_p}\in I_q\}$ and $\mathcal O^{(p)}_> =\{q=p+2,\ldots,n\,|\,t_{k_p}\in I_q\}$. The $m=1$ correction is
\begin{align}
\label{eq:theorem_m_one_OO}
    \delta\hat z_1^{\rm OO}
    &=
    \sum_{q\in\mathcal O^{(p)}_<}
    \epsilon_q\,
    \mathcal G_{1:q-1}^{\dagger}
    \circ\Phi_{t_{k_p},t_{k_{q-1}}}^{\dagger}
    \circ\mathcal K_{\alpha,t_{k_p}}^{\rm OO\dagger}
    \circ\mathsf L_{q:p}
    \circ\mathcal G_{p+1:n}^{\dagger}
    [\mathbbm 1_A]
    \nonumber\\
    &+
    \sum_{q\in\mathcal O^{(p)}_>}
    \epsilon_q\,
    \mathcal G_{1:p-1}^{\dagger}
    \circ\Phi_{t_{k_p},t_{k_{p-1}}}^{\dagger}
    \circ\mathsf L_{p:q-1}
    \circ\mathcal R_{\alpha,t_{k_p}}^{\rm OO\dagger}
    \circ\Phi_{t_{k_q},t_{k_p}}^{\dagger}
    \circ\mathsf L_{\hat X_q}
    \circ\mathcal G_{q+1:n}^{\dagger}
    [\mathbbm 1_A],
\end{align}
with $\mathcal K_{\alpha,t}^{\rm OO\dagger}[\hat O_A] = -i\sum_{j\omega}(\gamma_{\alpha j}(\omega)\hat O_A\hat L_{j\omega}(t) - \gamma_{\alpha j}^{*}(\omega)\hat L_{j\omega}^\dagger(t)\hat O_A)$, and $\mathcal R_{\alpha,t}^{\rm OO\dagger}[\hat O_A] = i\sum_{j\omega} \gamma_{j\alpha}^{*}(\omega) [\hat L_{j\omega}(t),\hat O_A]$.
Hence $\delta\hat z_1^{\rm OO}=0$ whenever both
$\mathcal O^{(p)}_<$ and $\mathcal O^{(p)}_>$ are empty.
\end{theorem}

We now prove the theorem. Let us start by reconstructing a Born-like approximation: for a product operator $\hat x_A\otimes \hat\ell_E\,\hat \sigma_E\,\hat r_E$ we can apply the $\mathbb P$--$\mathbb Q$ sectors decomposition onto its interaction picture evolution
\begin{equation}
    \hat z(t,t') = \hat U^{(I)}_{t,t'}(\hat x_A\otimes\hat\ell_E\,\hat \sigma_E\,\hat r_E)\hat U^{\dagger(I)}_{t,t'} = \tilde \Phi_{t,t'}^{(\ell,r)}[\hat x_A]\otimes\hat\sigma_E + \tilde\chi(t,t')~,
\end{equation}
where $\tilde\Phi_{t,t'}^{(\ell,r)}[\hat x_A] = \Tr_E[\hat z(t,t')]$. The first term is the projection onto the $\mathbb P$--sector and $\tilde\chi(t,t') = \mathbb Q[\hat z(t,t')]$.
Using that, by definition of the interaction picture, $\frac{d}{dt}\tilde \chi (t,t') = -i\mathbb Q[\tilde V(t),\hat z(t,t')]$, we have
\begin{equation}\label{eq:PQ_decomp}
    \tilde \chi(t,t')-\tilde \chi(t',t')=-i\int_{t'}^t\!ds\,\mathbb Q[\tilde V(s),\hat z(s,t')] = -i\int_{t'}^t\!ds\!\left([\tilde V(s),\tilde \Phi_{s,t'}^{(\ell,r)}[\hat x_A]\otimes\hat\sigma_E] + \mathbb Q[\tilde V(s),\tilde \chi(s,t')]\right)
\end{equation}
where $\tilde \chi(t',t') = \hat x_A\otimes\Delta_{\ell,r}[\hat\sigma_E]$, with $\Delta_{\ell,r}[\hat\sigma_E] = \hat\ell_E\hat\sigma_E\hat r_E - \rl\hat\sigma_E$. We also used that $[\tilde V(s),\tilde \Phi_{s,t'}[\hat x_A]\otimes\hat\sigma_E]$ lives entirely in the $\mathbb Q$ sector, since
\begin{equation}
    \mathbb P [\tilde V(s),\tilde \Phi_{s,t'}^{(\ell,r)}[\hat x_A]\otimes\hat\sigma_E] = \sum_j \Tr_E[\hat \sigma_E\tilde E_j(s)] [\lambda_j(s)\tilde A_k(s),\tilde \Phi_{s,t'}^{(\ell,r)}[\hat x_A]]\otimes\hat\sigma_E = 0,
\end{equation}
where we used $\Tr_E[\hat \sigma_E\tilde E_j(s)] = 0$. In a Born-like approximation we can drop the the last term in \eqref{eq:PQ_decomp}, so that we can approximate $\tilde \chi(t,t')$ by
\begin{equation}
    \tilde\chi_B(t,t') = \hat x_A\otimes\Delta_{\ell,r}[\hat \sigma_E] -i\int_{t'}^t\!ds [\tilde V(s),\tilde \Phi_{s,t'}^{(\ell,r)}[\hat x_A]\otimes\hat\sigma_E]~.
\end{equation}
Thus obtaining
\begin{equation}\label{eq:born}
    \hat z(t,t') = \tilde \Phi_{t,t'}^{(\ell,r)}[\hat x_A]\otimes\hat\sigma_E + \hat x_A\otimes\Delta_{\ell,r}[\hat \sigma_E] -i\int_{t'}^t\!ds [\tilde V(s),\tilde \Phi_{s,t'}^{(\ell,r)}[\hat x_A]\otimes\hat\sigma_E] + \mathcal E(t,t')~,
\end{equation}
where we defined the error by subtracting the truncated term from the exact one $\mathcal E(t,t') = \tilde\chi(t,t') - \tilde\chi_B(t,t')$. We can rewrite this error term as a Volterra equation
\begin{equation}
    \mathcal E(t,t') = -i\int_{t'}^t\!ds\,\mathbb Q\circ {\rm ad}_{\tilde V(s)}\circ\mathbb Q\,[\tilde\chi_B(s,t') + \mathcal E(s,t')]~,
\end{equation}
with ${\rm ad}_{\tilde V(s)}[\,\cdot\,] = [\tilde V(s),\,\cdot\,]$. We can solve it to obtain
\begin{equation}
    \mathcal E(t,t') = -i\int_{t'}^t\!ds\,\mathcal G_{\mathbb Q}(t,s)[\mathbb Q[\tilde V(s),\tilde \chi_B(s,t')]],\qquad \mathcal G_{\mathbb Q}(t,s) = \mathcal T e^{-i\int_{s}^t\!du\,\mathbb Q\circ{\rm ad}_{\tilde V(u)}\circ Q}~.
\end{equation}
We can now apply the approximation $\hat z(t,t')\simeq  \tilde \Phi_{t,t'}^{(\ell,r)}[\hat x_A]\otimes\hat\sigma_E + \tilde \chi_B(t,t')$ to the differential equation of $\tilde\Phi_{t,t'}^{(\ell,r)}$
\begin{align}
    \frac{d}{dt}\tilde\Phi_{t,t'}^{(\ell,r)}[\hat x_A] &= -i\Tr_E[[\tilde V(t),\hat z(t,t')]]~,\\
    &\simeq -i\Tr_E[[\tilde V(t),\tilde \chi_B(t,t')]]~,\\
    &= -i\Tr_E[[\tilde V(t),\hat x_A\otimes\Delta_{\ell,r}[\hat \sigma_E]]] - \int_{t'}^t\!\!ds\, \Tr_E[[\tilde V(t),[\tilde V(s),\tilde \Phi_{s,t'}^{(\ell,r)}[\hat x_A]\otimes\hat\sigma_E]]]~,\\
    &= -i\sum_j\lambda_j(t)\!\braket{\hat r_E\tilde E_j(t)\hat\ell_E}_{\hat\sigma_E}[\tilde A_j(t),\hat x_A] + \int_{t'}^t\!\!ds\, \mathcal K_{t,s}[\tilde \Phi_{s,t'}^{(\ell,r)}[\hat x_A]]~,\label{eq:diff_eq_phi_lr}
\end{align}
where we defined the Redfield kernel
\begin{equation}
\mathcal K_{t,s}[\hat x_A] = -\sum_{jj'}\lambda_j(t)\lambda_{j'}(s)\!\left(C_{jj'}(t-s)[\tilde A_j(t),\tilde A_{j'}(s)\hat x_A]-C_{jj'}^*(t-s)[\tilde A_j(t),\hat x_A\tilde A_{j'}(s)]\right)~.
\end{equation}
Since $\braket{\tilde E_j(t)}_{\hat\sigma_E} = 0$, we have that $\Phi^{(1,1)}_{t,t'}=\Phi_{t,t'}$ solves the homogeneous version of the integro-differential equation \eqref{eq:diff_eq_phi_lr}. Then it is solved by
\begin{equation}\label{eq:phi_lr}
    \tilde\Phi_{t,t'}^{(\ell,r)}[\hat x_A] = \rl \tilde\Phi_{t,t'}[\hat x_A] - i\sum_j\int_{t'}^t\!ds\, \lambda_j(s)\! \braket{\hat r_E\tilde E_j(s)\hat\ell_E}_{\hat\sigma_E}\!\tilde\Phi_{t,s}[[\tilde A_j(s),\hat x_A]].
\end{equation}
Let us consider the expectation values in this map for $\hat \ell_E = \tilde E_{a_1}(u_1)...\tilde E_{a_{n_\ell}}(u_{n_\ell})$ and $\hat r_E = \tilde E_{b_1}(v_1)...\tilde E_{b_{n_r}}(v_{n_r})$. If $n_\ell + n_r$ is even, then $\rel{j}{s} = 0$ and the map reduces to
\begin{equation}\label{eq:phi_lr_even}
    \tilde\Phi_{t,t'}^{(\ell,r)}[\hat x_A] = \rl \tilde\Phi_{t,t'}[\hat x_A] \qquad {\rm if}\qquad n_\ell + n_r \!\!\mod 2 = 0~,
\end{equation}
where $\rl$ is a Wick product, that produces over a sum over all possible $\tilde E$--pairs, each term a product of $\frac{n_\ell + n_r}{2}$ $C$--terms. Once integrated in time, these terms are of order $\mathcal{O}(\tau_E^{(n_\ell + n_r)/2})$ since $C$ has only support over a time window of width $\tau_E$ (bath correlation time). Instead, if $n_\ell + n_r$ is odd, then $\rl = 0$ and the map becomes 
\begin{equation}\label{eq:phi_lr_odd}
    \tilde\Phi_{t,t'}^{(\ell,r)}[\hat x_A] =  - i\sum_j\int_{t'}^t\!ds\, \lambda_j(s)\! \braket{\hat r_E\tilde E_j(s)\hat\ell_E}_{\hat\sigma_E}\!\tilde\Phi_{t,s}[[\tilde A_j(s),\hat x_A]] \qquad {\rm if}\qquad n_\ell + n_r \!\!\mod 2 = 1~,
\end{equation}
where the expected value produces terms of order $\mathcal{O}(\tau_E^{(n_\ell + n_r+1)/2})$. 
We now will apply this approximation to compute traces with the form 
\begin{equation}
    z_{k_1,...,k_n}\![\hat x_A,\hat \ell_E,\hat r_E](t_0) = \Tr[(\hat x_A\otimes \hat \ell_E\,\hat\sigma_E\, \hat r_E) \hat Z_{k_1,...,k_n}\!(t_0)]~,
\end{equation}
where $\hat Z_{k_1,...,k_n}\!(t_0)$ is a product of tensor-product Heisenberg-picture operators evaluated at different times 
\begin{equation}
    \hat Z_{k_1,...,k_n}(t_0) = \prod_{i=1}^n \hat U^\dagger_{t_{k_i},t_0}(\hat X_{k_i}\otimes \hat Y_{k_i})\hat U_{t_{k_i},t_0}~.
\end{equation}
We start by rewriting $\hat Z$ in with interaction picture operators
\begin{equation}
    \hat Z_{k_1,...,k_n}\!(t_0) = \prod_{i=1}^n \hat U^{\dagger(I)}_{t_{k_i},t_0}(\tilde X_{k_i}\!(t_{k_i})\otimes \tilde Y_{k_i}\!(t_{k_i}))\hat U_{t_{k_i},t_0}^{(I)}~,
\end{equation}
where we used $\hat U_{t,t'} = \hat U^{(0)}_{t,t'}\hat U^{(I)}_{t,t'}$.  
Using that $\hat U_{t',t}^{(I)} \hat Z_{k_1,...,k_n}\!(t) \hat U_{t',t}^{\dagger(I)} = \hat Z_{k_1,...,k_n}\!(t')$ and $\hat Z_{k_1,...,k_n}\!(t_0) = \hat U^{\dagger(I)}_{t_{k_1},t_0}(\tilde X_{k_1}\!(t_{k_1})\otimes \tilde Y_{k_1}\!(t_{k_1}))\hat U_{t_{k_1},t_0}^{(I)} \hat Z_{k_2,...,k_n}\!(t_0)$ we can apply the Born approximation
\begin{align}
    z_{k_1,...,k_n}\![\hat x_A,\hat \ell_E,\hat r_E](t_0) =&\, \Tr[\hat z(t_{k_1},t_0)(\tilde X_{k_1}\!(t_{k_1})\otimes \tilde Y_{k_1}\!(t_{k_1}))\hat Z_{k_2,...,k_n}\!(t_{k_1})]~,\\
    \simeq&~ \Tr[(\tilde \Phi_{t_{k_1},t_0}^{(\ell,r)}[\hat x_A]\tilde X_{k_1}\!(t_{k_1}) \otimes\hat\sigma_E \tilde Y_{k_1}\!(t_{k_1}))\hat Z_{k_2,...,k_n}\!(t_{k_1})]\\
    \nonumber
    &+ \Tr[(\hat x_A \tilde X_{k_1}\!(t_{k_1})\otimes\Delta_{\ell,r}[\hat\sigma_E]\tilde Y_{k_1}\!(t_{k_1}))\hat Z_{k_2,...,k_n}\!(t_{k_1})]\\
    \nonumber
    &- i\int_{t_0}^{t_{k_1}}\!\!ds\,\Tr[[\tilde V(s),\tilde \Phi_{s,t_0}^{(\ell,r)}[\hat x_A]\otimes\hat\sigma_E](\tilde X_{k_1}\!(t_{k_1})\otimes \tilde Y_{k_1}\!(t_{k_1}))\hat Z_{k_2,...,k_n}\!(t_{k_1})]~,\\
    \label{eq:rec_tmp}
    =&~ z_{k_2,...,k_n}\![\tilde \Phi_{t_{k_1},t_0}^{(\ell,r)}[\hat x_A]\tilde X_{k_1}\!(t_{k_1}),\mathbbm 1_E,\tilde Y_{k_1}\!(t_{k_1})](t_{k_1})\\
    \nonumber
    &+ z_{k_2,...,k_n}\![\hat x_A \tilde X_{k_1}\!(t_{k_1}),\hat\ell_E,\hat r_E\tilde Y_{k_1}\!(t_{k_1})](t_{k_1}) \\
    \nonumber
    &- \braket{\hat r_E\hat\ell_E}_{\hat\sigma_E} z_{k_2,...,k_n}\![\hat x_A \tilde X_{k_1}\!(t_{k_1}),\mathbbm 1_E,\tilde Y_{k_1}\!(t_{k_1})](t_{k_1})\\
    \nonumber
    &-i\sum_j\int_{t_0}^{t_{k_1}}\!\!ds\,z_{k_2,...,k_n}\![\lambda_j(s)\tilde A_j(s)\tilde \Phi_{s,t_0}^{(\ell,r)}[\hat x_A]\tilde X_{k_1}\!(t_{k_1}),\tilde E_j(s),\tilde Y_{k_1}\!(t_{k_1})](t_{k_1})\\
    \nonumber
    &+i\sum_j\int_{t_0}^{t_{k_1}}\!\!ds\,z_{k_2,...,k_n}\![\tilde \Phi_{s,t_0}^{(\ell,r)}[\hat x_A]\lambda_j(s)\tilde A_j(s)\tilde X_{k_1}\!(t_{k_1}),\mathbbm 1_E,\tilde E_j(s)\tilde Y_{k_1}\!(t_{k_1})](t_{k_1})~,
\end{align}
where we used that
\begin{align}
    [\tilde V(s),\tilde \Phi_{s,t_0}^{(\ell,r)}[\hat x_A]\otimes\hat\sigma_E] &= \sum_j[\lambda_j(s)\tilde A_j(s)\otimes\tilde E_j(s),\tilde \Phi_{s,t_0}^{(\ell,r)}[\hat x_A]\otimes\hat\sigma_E]~,\\
    &= \sum_j \lambda_j(s)\tilde A_j(s)\tilde \Phi_{s,t_0}^{(\ell,r)}[\hat x_A]\otimes\tilde E_j(s)\hat\sigma_E - \tilde \Phi_{s,t_0}^{(\ell,r)}[\hat x_A]\lambda_j(s)\tilde A_j(s)\otimes\hat\sigma_E\tilde E_j(s)~.
\end{align}
Equation~\eqref{eq:rec_tmp}, together with the terminal condition
\begin{equation}
z_{\varnothing}[\hat x_A,\hat\ell_E,\hat r_E](t)
=
\braket{\hat r_E\hat\ell_E}_{\hat\sigma_E}
\Tr_A[\hat x_A],
\end{equation}
already provides a closed recursive relation. Indeed, every application of Eq.~\eqref{eq:rec_tmp} removes one operator from the ordered product $\hat Z_{k_1,\ldots,k_n}$. After $n$ applications, every contribution is therefore reduced to an environment expectation value multiplied by a trace over $A$. Consequently, the recursion generates a finite sum of terms composed of nested reduced propagators, system-operator insertions, and environment correlation functions.

To determine which terms of Eq.~\eqref{eq:rec_tmp} are consistent with the accuracy of the Born--Markov approximation, we now introduce an explicit order counting. Let
\begin{equation}
    \mathcal I_E = \{ i\in\{1,\ldots,n\} \,|\, \hat Y_i\neq \mathbbm 1_E\},\qquad m=|\mathcal I_E|
\end{equation}
Thus, $m$ is the number of external environment operators appearing in the original operator string. In the application to the work moments, these operators originate from $\partial_t\hat V(t)= \sum_j\dot\lambda_j(t) \hat A_j\otimes\hat E_j$. The corresponding factors $\dot\lambda_j(t)$ are externally prescribed driving amplitudes and will therefore be kept explicitly: no power of the interaction strength $g$ is assigned to them (since in principle the driving speed $\dot\lambda$ can be of order $\mathcal O(1)$, it only needs to be small compared to the bath relaxation speed for the driven GKLS derivation to be consistent). Instead, $g$ will only be used as a bookkeeping parameter for environment operators generated by the interaction-picture evolution. Thus, every application of either of the last two terms of Eq.~\eqref{eq:rec_tmp}, or of the inhomogeneous term in Eq.~\eqref{eq:phi_lr}, increases the number $p$ of interaction-generated environment operators by one and contributes one power of $g$.

For a centered Gaussian environment, a terminal environment expectation containing $m$ external and $p$ interaction-generated operators can be nonzero only if
\begin{equation}\label{eq:gaussian_parity}
    m+p=2q, \qquad q\in\mathbb N.
\end{equation}
The terminal Wick expansion then contains $q$ two-point environment correlation functions. Since each correlation function is supported only when its two time arguments differ by at most $\tau_E$, every contraction involving an externally integrated work time localizes one relative-time integration to a window of width $\tau_E$. Consequently, a contribution containing $p$ interaction-generated operators has the schematic order
\begin{equation}\label{eq:external_sector_counting}
    z_{k_1,\ldots,k_n}^{[p]} = \mathcal O\!\left(g^p \tau_E^{(m+p)/2}\right),
\end{equation}
The minimum number of interaction-generated operators compatible with Eq.~\eqref{eq:gaussian_parity} is therefore
\begin{equation}
    p_{\min} =
    \begin{cases}
        0, & m\ {\rm even},\\
        1, & m\ {\rm odd}.
    \end{cases}
\end{equation}
The leading contribution to a word containing $m$ external environment operators consequently scales as
\begin{equation}\label{eq:leading_m_scaling}
    z_{k_1,\ldots,k_n}^{\rm lead} = \mathcal O\!\left(g^{m\bmod 2}\tau_E^{\lceil m/2\rceil}\right).
\end{equation}
In particular,
\begin{equation}\label{eq:m_sector_orders}
    \begin{aligned}
        m=0:&\qquad \mathcal O(1),\\
        m=1:&\qquad \mathcal O(g\tau_E),\\
        m=2:&\qquad \mathcal O(\tau_E),\\
        m=3:&\qquad \mathcal O(g\tau_E^2),\\
        m=4:&\qquad \mathcal O(\tau_E^2).
    \end{aligned}
\end{equation}
All sectors with $m\geq3$ are therefore at least of order $\mathcal O(\tau_E^2)$.

At each application of Eq.~\eqref{eq:rec_tmp}, a given recursion branch can generate at most one explicit environment operator. Indeed, the inhomogeneous contribution to $\tilde\Phi_{t,t'}^{(\ell,r)}$ and the last two terms of Eq.~\eqref{eq:rec_tmp} belong to alternative additive branches. Consequently, two explicit interaction-generated environment operators must originate from two distinct recursion steps and hence from two distinct propagation intervals.

To determine the correlation-time order of their contraction, we associate with the $i$-th recursion step the physical-time interval $I_i= [\min\{\tau_{i-1},\tau_i\},\max\{\tau_{i-1},\tau_i\}]$. The integration $\int_{\tau_{i-1}}^{\tau_i}ds$ remains oriented, so that the direction of propagation is already accounted for by its integration limits. Suppose that two interaction-generated environment operators are produced in the $q$-th and $r$-th recursion steps. Their contraction contains an integral of the form
\begin{equation}
    \int_{\tau_{q-1}}^{\tau_q}\!ds\int_{\tau_{r-1}}^{\tau_r}\!du\, C_{jj'}(u-s)\,\hat F_{q,r}(s,u),
\end{equation}
where $\hat F_{q,r}(s,u)$ varies slowly on the bath-correlation timescale. If $I_q\cap I_r$ has nonzero length, the average time can range over the whole overlap, whereas the relative time is restricted to $|u-s|\lesssim\tau_E$. The contraction is therefore of order $\mathcal O(\tau_E)$. If the two intervals intersect only at a common endpoint, both integration variables are restricted to a boundary layer of width $\tau_E$, and the contraction is instead of order $\mathcal O(\tau_E^2)$. Finally, if the intervals are separated by more than a bath-correlation time, the contraction is negligible in the Markov limit.

Consequently, contractions between two interaction-generated environment operators must be retained at order $\mathcal O(\tau_E)$ whenever the corresponding recursion intervals overlap over a finite physical-time range. This modifies the truncation only in the $m=0$ sector, where the zero-vertex ordinary quantum-regression contribution must be supplemented by a two-vertex term of order $\mathcal O(g^2\tau_E)$. In the $m=1$ sector, Gaussian parity requires an odd number of generated environment operators, so the leading contribution still contains one such operator and is of order $\mathcal O(g\tau_E)$; the effect of the interval overlaps is instead to enlarge the set of recursion steps from which this operator may be generated. In the $m=2$ sector, the leading term remains the direct contraction of the two external environment operators, of order $\mathcal O(\tau_E)$, whereas a contribution containing two additional generated operators involves two bath correlations and is therefore of order $\mathcal O(g^2\tau_E^2)$. All sectors with $m\geq3$ likewise require at least two bath correlations and begin at order $\mathcal O(\tau_E^2)$.

Let us now compare this counting with the error introduced by the Born approximation. Expanding the Volterra equation for $\mathcal E(t,t')$, the contribution generated from the initial correlation operator begins with
\begin{equation}\label{eq:born_error_first}
    \mathcal E_{\Delta}^{[1]}(t,t') = -i\int_{t'}^t\!ds\,\mathbb Q\circ\operatorname{ad}_{\tilde V(s)} \left[\hat x_A\otimes\Delta_{\ell,r}[\hat\sigma_E]\right],
\end{equation}
This term contains one interaction vertex and is therefore formally of order $\mathcal O(g)$. A second class of contributions is obtained by applying the error kernel to the one-interaction term already contained in $\tilde\chi_B(t,t')$. To lowest order, this gives
\begin{equation}\label{eq:born_error_second}
    \mathcal E_B^{[2]}(t,t') = -\int_{t'}^t\!ds \int_{t'}^s\!du\,\mathbb Q\circ{\rm ad}_{\tilde V(s)} \circ\mathbb Q\circ{\rm ad}_{\tilde V(u)}\left[ \tilde\Phi_{u,t'}^{(\ell,r)}[\hat x_A]\otimes\hat\sigma_E \right]=\mathcal O(g^2).
\end{equation}
The nominal coupling order in Eq.~\eqref{eq:born_error_second} does not necessarily coincide with the first nonvanishing contribution to a given external-operator sector, since the generalized propagator $\tilde\Phi^{(\ell,r)}$ or the final Gaussian environment expectation may vanish at that order. Owing to the outer $\mathbb Q$ projector, both contributions remain in the $\mathbb Q$--sector and therefore have vanishing direct environment trace,
\begin{equation}
    \Tr_E\!\left[ \mathcal E_{\Delta}^{[1]}(t,t')\right] = \Tr_E\!\left[\mathcal E_B^{[2]}(t,t')\right] = 0.
\end{equation}
They can contribute to a multi-time expectation value only when subsequent external environment insertions or further interaction vertices produce a nonvanishing environment trace.

Consider first the $m=0$ sector. Since the boundary words are initially trivial, $\Delta_{1,1}[\hat\sigma_E]=0$, and hence the contribution $\mathcal E_{\Delta}^{[1]}$ is absent. The leading error is therefore generated by $\mathcal E_B^{[2]}$, which contains two interaction-generated environment operators but remains in the $\mathbb Q$ sector and consequently has vanishing direct environment trace. A nonzero contribution to the multi-time expectation value requires further interaction vertices. One additional vertex gives a centered Gaussian three-point function and vanishes, while two additional vertices produce the first potentially nonzero four-point expectation. The leading Born error in the $m=0$ sector therefore scales as
\begin{equation}
    \delta z_{m=0}^{\rm Born} = \mathcal O(g^4\tau_E^2).
\end{equation}
By contrast, the contraction of two interaction-generated operators on overlapping recursion intervals contributes at order $z_{m=0}^{[2]} = \mathcal O(g^2\tau_E)$. The two-vertex correction in the $m=0$ sector is therefore controlled at the working order retained here.

Consider next a word containing a single external environment operator. When this operator enters the boundary correlation operator $\Delta_{\ell,r}[\hat\sigma_E]$, the latter contains one environment operator. The first-order error $\mathcal E_{\Delta}^{[1]}$ introduces one additional interaction-generated environment operator and therefore contains two environment operators in total: one external and one interaction-generated operator. Nevertheless, its direct environment trace vanishes because of the outer $\mathbb Q$ projector. Since all subsequent external insertions act trivially on the environment, a nonzero contribution can arise only if further interaction vertices are generated. Adding one further interaction vertex introduces a third environment operator, giving a centered Gaussian three-point function, which vanishes. Adding two further interaction vertices instead produces four environment operators in total and therefore gives the first potentially nonvanishing Wick contractions. The first nonzero contribution generated from $\mathcal E_{\Delta}^{[1]}$ consequently contains three interaction-generated environment operators altogether and is of order $\mathcal O(g^3\tau_E^2)$.  

Applying a similar reasoning to $\mathcal E_B^{[2]}$, its nominal second-order contribution vanishes in the single-insertion sector. Indeed, when the boundary word contains one environment operator, the zero-vertex part of the generalized propagator is proportional to $\braket{\hat r_E\hat\ell_E}_{\hat\sigma_E} = 0$. The first nonzero part of $\tilde\Phi^{(\ell,r)}$ therefore requires one additional interaction vertex. Consequently, the first nonvanishing contribution generated from $\mathcal E_B^{[2]}$ also contains three interaction vertices and two environment correlation functions, thus it is also of order $\mathcal O(g^3\tau_E^2)$.
The total Born error in the single-insertion sector thus satisfies
\begin{equation}
    \delta z_{m=1}^{\rm Born} = \mathcal O(g^3\tau_E^2).
\end{equation}
By contrast, the leading single-insertion contribution contains one external and one interaction-generated environment operator and is of order $z_{m=1}^{\rm lead} = \mathcal O(g\tau_E)$. The leading $m=1$ contribution is therefore unaffected by the Born error at the order retained here.

For a word containing two external environment operators, the leading contribution is their direct contraction, $z_{m=2}^{\rm lead} = \mathcal O(\tau_E)$. A contribution containing a single interaction-generated environment operator has odd Gaussian parity and vanishes. The first nonvanishing Born correction therefore contains two interaction-generated operators and scales as
\begin{equation}
    \delta z_{m=2}^{\rm Born}
    =
    \mathcal O(g^2\tau_E^2).
\end{equation}
The leading two-insertion contribution is thus also unaffected by the Born error at order $\mathcal O(\tau_E)$.

The situation changes for a word containing three external environment operators. After one external operator has entered $\Delta_{\ell,r}[\hat\sigma_E]$, the first-order error $\mathcal E_{\Delta}^{[1]}$ contains this external operator and one interaction-generated operator. The two remaining external operators can then be inserted subsequently, producing a terminal expectation value containing four environment operators. Its Wick expansion contains two two-point correlation functions, and therefore
\begin{equation}
    \delta z_{m=3}^{\rm Born} = \mathcal O(g\tau_E^2).
\end{equation}
This is of the same order as the leading three-insertion contribution, $z_{m=3}^{\rm lead} = \mathcal O(g\tau_E^2)$. The leading coefficient of the $m=3$ sector therefore cannot be determined consistently from the Born-truncated recursion \eqref{eq:rec_tmp}.

The resulting order counting is summarized by
\begin{equation}\label{eq:born_error_order_table}
    \begin{array}{c|c|c}
        m & \text{retained contributions} & \delta z_m^{\rm Born} \\
        \hline
        0 &\mathcal O(1) + \mathcal O(g^2\tau_E) & \mathcal O(g^4\tau_E^2) \\
        1 & \mathcal O(g\tau_E) & \mathcal O(g^3\tau_E^2) \\
        2 & \mathcal O(\tau_E) & \mathcal O(g^2\tau_E^2) \\
        3 & \mathcal O(g\tau_E^2) & \mathcal O(g\tau_E^2)
    \end{array}
\end{equation}
Consequently, the $m=0$, $m=1$, and $m=2$ sectors are controlled through order $\mathcal O(\tau_E)$, whereas the leading $m=3$ contribution is already of the same order as the Born error. More generally, every sector with $m\geq3$ begins at order $\mathcal O(\tau_E^2)$ or higher and therefore lies beyond the accuracy retained here.

It follows that, when computing the work moments to first order in the explicit bath-correlation-time contributions, the sum over operator strings may be truncated as
\begin{equation}\label{eq:work_word_truncation}
    z_{k_1,\ldots,k_n} =
    \begin{cases}
        z_{k_1,\ldots,k_n}^{[0]} + z_{k_1,\ldots,k_n}^{[2]} + \mathcal O(\tau_E^2), & m=0,\\
        z_{k_1,\ldots,k_n}^{[1]} + \mathcal O(\tau_E^2), & m=1,\\
        z_{k_1,\ldots,k_n}^{[0]} + \mathcal O(\tau_E^2), & m=2,\\
        \mathcal O(\tau_E^2), & m\geq3.
    \end{cases}
\end{equation}
Here, the superscript $[p]$ denotes the contribution containing exactly $p$ additional interaction vertices generated by the recursion. The $m=0$ sector contains the zero-vertex quantum regression contribution and, for overlapping recursion intervals, a two-vertex correction. No interaction-generated operator is required for the leading $m=2$ contribution, whereas precisely one is required for the leading $m=1$ contribution.

To implement this truncation, we use the parity-dependent form of the generalized propagator derived above. At a given recursion step, let $n_\ell+n_r$ denote the number of environment operators contained in the boundary words $\hat\ell_E$ and $\hat r_E$. Since the environment is centered and Gaussian, the two terms in Eq.~\eqref{eq:phi_lr} cannot contribute simultaneously. If $n_\ell+n_r$ is even, then
\begin{equation}\label{eq:phi_lr_even_truncation}
    \tilde\Phi_{t,t'}^{(\ell,r)}[\hat x_A] = \braket{\hat r_E\hat\ell_E}_{\hat\sigma_E}\tilde\Phi_{t,t'}[\hat x_A], \qquad n_\ell+n_r\!\!\mod 2=0.
\end{equation}
Conversely, if $n_\ell+n_r$ is odd, then
\begin{equation}\label{eq:phi_lr_odd_truncation}
    \tilde\Phi_{t,t'}^{(\ell,r)}[\hat x_A] = -i\sum_j\int_{t'}^t\!ds\, \lambda_j(s)\braket{\hat r_E\tilde E_j(s)\hat\ell_E}_{\hat\sigma_E} \tilde\Phi_{t,s}\left[[\tilde A_j(s),\hat x_A]\right], \qquad n_\ell+n_r\!\!\mod 2=1.
\end{equation}
For later convenience, we denote the right-hand side of Eq.~\eqref{eq:phi_lr_odd_truncation} by
\begin{equation}\label{eq:odd_phi_source}
    \mathcal J_{t,t'}^{(\ell,r)}[\hat x_A] =-i\sum_j \int_{t'}^t\!ds\,\lambda_j(s) \braket{\hat r_E\tilde E_j(s)\hat\ell_E}_{\hat\sigma_E}\tilde\Phi_{t,s} \left[[\tilde A_j(s),\hat x_A]\right],
\end{equation}
when $n_\ell+n_r$ is odd, and set $\mathcal J_{t,t'}^{(\ell,r)}[\hat x_A] = 0$ when $n_\ell+n_r$ is even.\\
To simplify the recursion notation, we write
\begin{equation}
    \tau_i=t_{k_i}, \qquad \hat X_i=\tilde X_{k_i}(\tau_i), \qquad \hat Y_i=\tilde Y_{k_i}(\tau_i), \qquad \tau_0=t_0,
\end{equation}
and introduce
\begin{equation}\label{eq:mathsf_z_definition}
    \mathsf z_i^{[p]} [\hat O_A,\hat L_E,\hat R_E] = z_{k_i,\ldots,k_n}^{[p]} [\hat O_A,\hat L_E,\hat R_E](\tau_{i-1}),
\end{equation}
where the superscript $[p]$ denotes the contribution containing exactly $p$ interaction-generated environment operators from the $i$-th recursion step onward. Whenever $\hat L_E=\mathbbm 1_E$, we use the shorthand
\begin{equation}
    \mathsf z_i^{[p]}[\hat O_A,\hat R_E] = \mathsf z_i^{[p]}[\hat O_A,\mathbbm 1_E,\hat R_E].
\end{equation}
By construction, from \eqref{eq:rec_tmp} it follows the following recursion rule (in $i$ and $p$)
\begin{align}
    \mathsf z_i^{[p]}[\hat O_A,\hat L_E,\hat R_E] \simeq&~ \mathsf z_{i+1}^{[p]}[\hat O_A\hat X_i,\hat L_E,\hat R_E\hat Y_i]
    \nonumber\\
    &+ \braket{\hat R_E\hat L_E}_{\hat\sigma_E}
    \left(\mathsf z_{i+1}^{[p]}[\tilde\Phi_{\tau_i,\tau_{i-1}}[\hat O_A]\hat X_i,\hat Y_i] - \mathsf z_{i+1}^{[p]}[\hat O_A\hat X_i,\hat Y_i]\right)
    \nonumber\\
    &+
    \mathsf z_{i+1}^{[p-1]}[\mathcal J_{\tau_i,\tau_{i-1}}^{(L,R)}[\hat O_A]\hat X_i,\hat Y_i]
    \nonumber\\
    &-i\braket{\hat R_E\hat L_E}_{\hat\sigma_E}\sum_j\int_{\tau_{i-1}}^{\tau_i}\!ds\,\mathsf z_{i+1}^{[p-1]}[\lambda_j(s)\tilde A_j(s)\tilde\Phi_{s,\tau_{i-1}}[\hat O_A]\hat X_i,\tilde E_j(s),\hat Y_i]
    \nonumber\\
    &+i\braket{\hat R_E\hat L_E}_{\hat\sigma_E}\sum_j\int_{\tau_{i-1}}^{\tau_i}\!ds\,\mathsf z_{i+1}^{[p-1]}[\tilde\Phi_{s,\tau_{i-1}}[\hat O_A]\lambda_j(s)\tilde A_j(s)\hat X_i,\tilde E_j(s)\hat Y_i],
    \label{eq:recursion}
\end{align}
where we loose one order of $p$ whenever we contract a bath vertex contribution (i.e. originating internally to the recursion formula, not from the external $\hat Y_i$ terms). The terminal conditions are
\begin{align}
    \mathsf z_{n+1}^{[0]}[\hat O_A,\hat L_E,\hat R_E] &= \Tr_A[\hat O_A]\braket{\hat R_E\hat L_E}_{\hat\sigma_E},
    \label{eq:terminal_zero_vertex}\\
    \mathsf z_{n+1}^{[p]}[\hat O_A,\hat L_E,\hat R_E] &= 0, \qquad p\geq1.
    \label{eq:terminal_one_vertex}
\end{align}
The first two contributions in Eq.~\eqref{eq:recursion} propagate the branch without generating the remaining interaction-generated environment operator. The third term generates it through the odd-parity contribution $\mathcal J_{\tau_i,\tau_{i-1}}^{(L,R)}$ to the generalized propagator and is therefore nonzero only when the boundary word contains an odd number of environment operators. The last two terms instead generate it explicitly through the Born correlation term and contribute only when $\braket{\hat R_E\hat L_E}_{\hat\sigma_E}\neq0$, which requires the boundary word to contain an even number of environment operators. In the $m=1$ sector, the last two terms generate the interaction operator before the external environment operator is encountered, when the boundary word is even. After the external operator has entered the boundary word, the latter is odd, and the remaining interaction operator can instead be generated through the $\mathcal J$ term. The same recursion will later be used in the $m=0$ two-vertex sector. There, the first interaction-generated environment operator is produced explicitly by one of the last two terms of Eq.~\eqref{eq:rec_tmp}. The resulting boundary word contains one environment operator and is therefore odd. From that point onward, the branch is described by Eq.~\eqref{eq:recursion}, and the second interaction-generated operator is supplied by the $\mathcal J$ term. Once the required interaction-generated operators have been produced, the remaining steps are governed by the zero-vertex recursion.

We first consider the contribution containing no interaction-generated environment operator. The last two terms of Eq.~\eqref{eq:recursion} are then absent, while the generalized propagator contributes only through Eq.~\eqref{eq:phi_lr_even_truncation}. The resulting recursion is
\begin{equation}\label{eq:zero_vertex_recursion}
    \mathsf z_i^{[0]}[\hat O_A,\hat L_E,\hat R_E] \simeq \mathsf z_{i+1}^{[0]}[\hat O_A\hat X_i,\hat L_E,\hat R_E\hat Y_i] + \braket{\hat R_E\hat L_E}_{\hat\sigma_E}\left(\mathsf z_{i+1}^{[0]}[\tilde\Phi_{\tau_i,\tau_{i-1}}[\hat O_A]\hat X_i,\hat Y_i] - \mathsf z_{i+1}^{[0]}[\hat O_A\hat X_i,\hat Y_i]\right).
\end{equation}
In particular, when $\hat L_E=\mathbbm 1_E$, this becomes
\begin{equation}\label{eq:zero_vertex_recursion_short}
    \mathsf z_i^{[0]}[\hat O_A\hat R_E] \simeq \mathsf z_{i+1}^{[0]}[\hat O_A\hat X_i,\hat R_E\hat Y_i] + \braket{\hat R_E}_{\hat\sigma_E}\left(\mathsf z_{i+1}^{[0]}[\tilde\Phi_{\tau_i,\tau_{i-1}}[\hat O_A]\hat X_i,\hat Y_i] - \mathsf z_{i+1}^{[0]}[\hat O_A\hat X_i,\hat Y_i]\right).
\end{equation}
Equation~\eqref{eq:zero_vertex_recursion_short} determines the zero-vertex contributions in the $m=0$ and $m=2$ sectors. For $m=0$, it gives the ordinary quantum-regression contribution $z^{[0]}$, while for $m=2$ it gives the leading direct contraction between the two external environment operators. For $m=0$, one has $\hat Y_i=\mathbbm 1_E$ for every $i$, and the environment word remains equal to $\mathbbm 1_E$. The first term in Eq.~\eqref{eq:zero_vertex_recursion_short} then cancels the negative term inside the parentheses, giving
\begin{equation}\label{eq:QRT}
    \mathsf z_i^{[0]}[\hat O_A,\mathbbm 1_E]\simeq \mathsf z_{i+1}^{[0]}[\tilde\Phi_{\tau_i,\tau_{i-1}}[\hat O_A]\hat X_i,\mathbbm 1_E],
\end{equation}
which is the zero-vertex contribution to the $m=0$ sector and coincides with the ordinary quantum regression theorem.

For $m=2$, the recursion carries the first external environment operator in the residual word $\hat R_E$ until the second external operator is reached through the first term. While $\hat R_E$ contains a single environment operator, $\braket{\hat R_E}_{\hat\sigma_E}=0$, and only the first term of Eq.~\eqref{eq:zero_vertex_recursion_short} contributes. The system operators between the two external insertions are therefore multiplied without an intervening reduced propagator. Once the second external operator is inserted, the terminal environment expectation produces their direct two-point correlation function. No interaction-generated environment operator is required at this order.

We have therefore reduced the computation of the work moments, up to corrections of order $\mathcal O(\tau_E^2)$, to the following three sectors:
\begin{equation}
    \begin{aligned}
        m=0: &\qquad \text{ordinary quantum regression }(p=0)  + \text{two-vertex contraction } (p=2), \\
        m=1: &\qquad \text{one interaction-generated environment operator } (p=1)\\
        m=2: &\qquad \text{one direct contraction between the two external operators } (p=0).
    \end{aligned}
\end{equation}
All strings containing three or more insertions of $\partial_t\hat V$ contribute only at order $\mathcal O(\tau_E^2)$ or higher and are therefore discarded at the accuracy retained here.

We now solve explicitly the contributions retained above. The zero-vertex contribution to the $m=0$ sector has already been obtained and coincides with the ordinary quantum regression theorem. We first solve the zero-vertex $m=2$ sector and then the one-vertex $m=1$ sector. The latter result will subsequently be used to obtain the remaining two-vertex contribution to the $m=0$ sector.
\begin{equation}\label{eq:regression_prefix}
    \hat S_0 = \hat x_A, \qquad \hat S_i = \tilde\Phi_{\tau_i,\tau_{i-1}}[\hat S_{i-1}]\hat X_i, \qquad i=1,\ldots,n,
\end{equation}
and the future regression functional
\begin{equation}\label{eq:future_regression}
    \mathcal F_q[\hat O_A]= \Tr_A\!\left[\left(\mathsf R_{\hat X_n}\!\circ \tilde\Phi_{\tau_n,\tau_{n-1}}\right) \circ\cdots\circ \left(\mathsf R_{\hat X_{q+1}}\!\circ\tilde\Phi_{\tau_{q+1},\tau_q}\right) \![\hat O_A]\right],
\end{equation}
where
\begin{equation}
    \mathsf R_{\hat X}[\hat O_A] = \hat O_A\hat X, \qquad \mathcal F_n[\hat O_A] = \Tr_A[\hat O_A].
\end{equation}
In this notation, the zero-vertex contribution to the $m=0$ sector is
\begin{equation}\label{eq:solution_m_zero}
    z_{k_1,\ldots,k_n}^{[0]}[\hat x_A,\mathbbm 1_E,\mathbbm 1_E](t_0) \simeq \Tr_A[\hat S_n].
\end{equation}
It is useful to solve the $m=2$ sector before the $m=1$ sector, since the zero-vertex tails generated in the latter are precisely of the form encountered in the former. Suppose that the two external environment operators occur at positions $a<b$,
\begin{equation}
    \hat Y_a=\tilde E_{\alpha}(\tau_a), \qquad \hat Y_b=\tilde E_{\beta}(\tau_b), \qquad \hat Y_i=\mathbbm 1_E \quad \text{for} \quad i\neq a,b.
\end{equation}
Before the first external insertion, repeated application of Eq.~\eqref{eq:zero_vertex_recursion_short} gives the ordinary regression prefix,
\begin{equation}
    \mathsf z_1^{[0]}[\hat x_A,\mathbbm 1_E] \simeq \mathsf z_a^{[0]} [\hat S_{a-1},\mathbbm 1_E].
\end{equation}
At the $a$-th step, the first and third terms of Eq.~\eqref{eq:zero_vertex_recursion_short} cancel, yielding
\begin{equation}
    \mathsf z_a^{[0]}[\hat S_{a-1},\mathbbm 1_E] \simeq \mathsf z_{a+1}^{[0]}[\hat S_a,\tilde E_{\alpha}(\tau_a)].
\end{equation}
Since the residual environment word now contains a single centered environment operator, its expectation value vanishes. Hence, until the second external insertion is reached, only the first term of Eq.~\eqref{eq:zero_vertex_recursion_short} contributes:
\begin{equation}
    \mathsf z_{a+1}^{[0]}[\hat S_a,\tilde E_{\alpha}(\tau_a)] \simeq \mathsf z_b^{[0]}[\hat S_a\hat X_{a+1}\cdots\hat X_{b-1},\tilde E_{\alpha}(\tau_a)].
\end{equation}
The insertion at position $b$ closes the environment contraction,
\begin{equation}
    \mathsf z_b^{[0]}[\hat S_a\hat X_{a+1}\cdots\hat X_{b-1},\tilde E_{\alpha}(\tau_a)] \simeq \mathsf z_{b+1}^{[0]}[\hat S_a\hat X_{a+1}\cdots\hat X_b,\tilde E_{\alpha}(\tau_a)\tilde E_{\beta}(\tau_b)].
\end{equation}
After the second external insertion, the residual environment word is $\hat R_{ab} = \tilde E_{\alpha}(\tau_a) \tilde E_{\beta}(\tau_b)$. Since $\hat Y_i=\mathbbm 1_E$ for every $i>b$, the remaining recursion takes the form
\begin{equation}\label{eq:closed_pair_recursion}
    \mathsf z_i^{[0]}[\hat O_A,\hat R_{ab}] \simeq \mathsf z_{i+1}^{[0]}[\hat O_A\hat X_i,\hat R_{ab}] + \braket{\hat R_{ab}}_{\hat\sigma_E}\mathsf z_{i+1}^{[0]}[\tilde\Phi_{\tau_i,\tau_{i-1}}[\hat O_A]\hat X_i,\mathbbm 1_E] - \braket{\hat R_{ab}}_{\hat\sigma_E}\mathsf z_{i+1}^{[0]}[\hat O_A\hat X_i,\mathbbm 1_E].
\end{equation}
The first term in Eq.~\eqref{eq:closed_pair_recursion} still carries the two-operator residual word and must therefore be iterated further. Upon applying the recursion to this term at the next step, it generates the positive contribution
\begin{equation}
    \braket{\hat R_{ab}}_{\hat\sigma_E}\mathsf z_{i+2}^{[0]}[\tilde\Phi_{\tau_{i+1},\tau_i}[\hat O_A\hat X_i]\hat X_{i+1},\mathbbm 1_E].
\end{equation}
By the $m=0$ solution, this is precisely
\begin{equation}
    \braket{\hat R_{ab}}_{\hat\sigma_E}\mathsf z_{i+1}^{[0]}[\hat O_A\hat X_i,\mathbbm 1_E],
\end{equation}
and hence cancels the last term of Eq.~\eqref{eq:closed_pair_recursion}. The same cancellation occurs successively at every remaining recursion step. The resulting telescopic expansion leaves only
\begin{equation}
    \mathsf z_{b+1}^{[0]}[\hat O_A,\hat R_{ab}] \simeq \mathsf z_{n+1}^{[0]}[\hat O_A\hat X_{b+1}\cdots\hat X_n,\hat R_{ab}] + \braket{\hat R_{ab}}_{\hat\sigma_E}\mathsf z_{b+2}^{[0]}[\tilde\Phi_{\tau_{b+1},\tau_b}[\hat O_A]\hat X_{b+1},\mathbbm 1_E] - \braket{\hat R_{ab}}_{\hat\sigma_E}\mathsf z_{n+1}^{[0]}[\hat O_A\hat X_{b+1}\cdots\hat X_n,\mathbbm 1_E].
\end{equation}
The first and last terms cancel by the terminal condition. Therefore, using the explicit $m=0$ solution for the remaining tail, we obtain
\begin{equation}
    \mathsf z_{b+1}^{[0]}[\hat O_A,\hat R_{ab}] \simeq \braket{\hat R_{ab}}_{\hat\sigma_E}\mathcal F_b[\hat O_A].
    \label{eq:closed_pair_solution}
\end{equation}
Applying this result to $\hat O_A = \hat S_a\hat X_{a+1}\cdots\hat X_b$, the solution in the $m=2$ sector is therefore
\begin{equation}\label{eq:solution_m_two}
    z_{k_1,\ldots,k_n}[\hat x_A,\mathbbm 1_E,\mathbbm 1_E](t_0) \simeq C_{\alpha\beta}(\tau_a-\tau_b)\mathcal F_b[\hat S_a \hat X_{a+1}\cdots\hat X_b].
\end{equation}
Equivalently, writing both the future regression and the prefix
$\hat S_a$ explicitly,
\begin{equation}
    z_{k_1,\ldots,k_n}[\hat x_A,\mathbbm 1_E,\mathbbm 1_E](t_0) \simeq C_{\alpha\beta}(\tau_a-\tau_b)\Tr_A\!\big[\mathsf R_{\hat X_n}\tilde\Phi_{\tau_n,\tau_{n-1}}\cdots \mathsf R_{\hat X_{b+1}}\tilde\Phi_{\tau_{b+1},\tau_b}[\mathsf R_{\hat X_a}\tilde\Phi_{\tau_a,\tau_{a-1}}\cdots\mathsf R_{\hat X_1}\tilde\Phi_{\tau_1,t_0}[\hat x_A]\hat X_{a+1}\cdots\hat X_b] \big].
\end{equation}

Finally, we treat the $m=1$ case. Let the unique external environment operator occur at position $p$,
\begin{equation}
    \hat Y_p=\tilde E_{\alpha}(\tau_p), \qquad \hat Y_i=\mathbbm 1_E \quad\text{for}\quad i\neq p,
\end{equation}
Before the external insertion is reached, the boundary word is trivial. Thus, for $i<p$, one has $\hat L_E=\hat R_E=\hat Y_i=\mathbbm 1_E$ and $\mathcal J_{\tau_i,\tau_{i-1}}^{(L,R)}=0$. The first term of Eq.~\eqref{eq:recursion} then cancels the negative $\mathsf z^{[1]}$ term inside the parentheses. Consequently, the recursion reduces to
\begin{align}\label{eq:m_one_recursion_before_p}
    \mathsf z_i^{[1]}[\hat O_A,\mathbbm 1_E] \simeq&~ \mathsf z_{i+1}^{[1]}[\tilde\Phi_{\tau_i,\tau_{i-1}}[\hat O_A]\hat X_i,\mathbbm 1_E]
    \nonumber\\
    & -i\sum_j\int_{\tau_{i-1}}^{\tau_i}\!ds\,\mathsf z_{i+1}^{[0]}[\lambda_j(s)\tilde A_j(s)\tilde\Phi_{s,\tau_{i-1}}[\hat O_A]\hat X_i,\tilde E_j(s),\mathbbm 1_E]
    \nonumber\\
    & +i\sum_j\int_{\tau_{i-1}}^{\tau_i}\!ds\,\mathsf z_{i+1}^{[0]}[\tilde\Phi_{s,\tau_{i-1}}[\hat O_A]\lambda_j(s)\tilde A_j(s)\hat X_i,\tilde E_j(s)].
\end{align}
The first term carries the unresolved $\mathsf z^{[1]}$ branch to the next step and generates the ordinary regression prefix. The last two terms instead select the current interval as the one in which the interaction-generated environment operator is introduced. Once either of these source terms is selected, the remaining tail is governed by the zero-vertex recursion. Iterating the first term of Eq.~\eqref{eq:m_one_recursion_before_p} up to the $q$-th step gives the prefix $\hat S_{q-1}$. Hence, we obtain that at step $p+1$ we have
\begin{align}
    \mathsf z_{1}^{[1]}[\hat x_A,\mathbbm 1_E]\simeq&~ \mathsf z_{p+1}^{[1]}[\hat S_p,\tilde E_{\alpha}(\tau_p)] \\
    & -i\sum_{q=1}^p\sum_j\int_{\tau_{q-1}}^{\tau_q}\!ds\,\mathsf z_{q+1}^{[0]}[\lambda_j(s)\tilde A_j(s)\tilde\Phi_{s,\tau_{q-1}}[\hat S_{q-1}]\hat X_q,\tilde E_j(s),\hat Y_q]
    \nonumber\\
    & +i\sum_{q=1}^p\sum_j\int_{\tau_{q-1}}^{\tau_q}\!ds\,\mathsf z_{q+1}^{[0]}[\tilde\Phi_{s,\tau_{q-1}}[\hat S_{q-1}]\lambda_j(s)\tilde A_j(s)\hat X_q,\tilde E_j(s)\hat Y_q].
    \label{eq:m_one_before_mathsf}
\end{align}
With the $m=2$ solution we can treat the last two lines, for $q\leq p$ we have
\begin{align}\label{eq:m_two_left_tail}
    \mathsf z_{q+1}^{[0]}[\hat O_A,\tilde E_j(s),\hat Y_q] &\simeq \braket{\tilde E_{\alpha}(\tau_p)\tilde E_j(s)}_{\hat\sigma_E} \mathcal F_p[\hat O_A\hat X_{q+1}\cdots\hat X_p],\\
    \label{eq:m_two_right_tail}
    \mathsf z_{q+1}^{[0]}[\hat O_A,\tilde E_j(s)\hat Y_q] &\simeq \braket{\tilde E_j(s)\tilde E_{\alpha}(\tau_p)}_{\hat\sigma_E} \mathcal F_p[\hat O_A\hat X_{q+1}\cdots\hat X_p],
\end{align}
where the products $\hat X_{q+1}\cdots\hat X_p$ are understood as the identity for $q=p$. While the first term is obtained from the recursion Eq.~\eqref{eq:recursion}, now we have only the first and third line that survive because $\braket{\tilde E_{\alpha}(\tau_p)}_{\hat \sigma_E} = 0$, for $p < q \leq n$ we have
\begin{equation}
    \mathsf z_{q}^{[1]}[\hat O_A,\tilde E_{\alpha}(\tau_p)] \simeq \mathsf z_{q+1}^{[1]}[\hat O_A\hat X_{q},\tilde E_{\alpha}(\tau_p)] + \mathsf z_{q+1}^{[0]}[\mathcal J_{\tau_{q},\tau_{q-1}}^{(1,E)}[\hat O_A]\hat X_{q},\mathbbm 1_E].
    \label{eq:m_one_after_p}
\end{equation}
Where the first term carries through the unmatched external environment operator. And the second term corresponds to it being matched by the interaction vertex within $\tilde \Phi^{(L,R)}$, this term can be simply treated by $m=0$ solution, yielding
\begin{equation}
    \mathsf z_{q+1}^{[0]}[\mathcal J_{\tau_{q},\tau_{q-1}}^{(1,E)}[\hat O_A]\hat X_{q},\mathbbm 1_E] = \mathcal F_{q}[\mathcal J_{\tau_{q},\tau_{q-1}}^{(1,E)}[\hat O_A]\hat X_{q}].
\end{equation}
Instead, the first term eventually reaches the terminal condition $\mathsf z_{n+1}^{[1]}[\hat O_A,\hat R_E] = 0$. Recombining all the terms, we thus obtain
\begin{align}\nonumber
    \mathsf z_{1}^{[1]}[\hat x_A,\mathbbm 1_E]\simeq&~ \sum_{q=p+1}^n \mathcal F_q[\mathcal J_{\tau_{q},\tau_{q-1}}^{(1,E)}[\hat S_p \hat X_{p+1}\cdots\hat X_{q-1}]\hat X_{q}]\\
    & -i\sum_{q=1}^p\sum_j\int_{\tau_{q-1}}^{\tau_q}\!ds \braket{\tilde E_{\alpha}(\tau_p)\tilde E_j(s)}_{\hat\sigma_E} \mathcal F_p[\lambda_j(s)\tilde A_j(s)\tilde\Phi_{s,\tau_{q-1}}[\hat S_{q-1}]\hat X_q\cdots\hat X_p]
    \nonumber\\
    & +i\sum_{q=1}^p\sum_j\int_{\tau_{q-1}}^{\tau_q}\!ds \braket{\tilde E_j(s)\tilde E_{\alpha}(\tau_p)}_{\hat\sigma_E}\mathcal F_p[\tilde\Phi_{s,\tau_{q-1}}[\hat S_{q-1}]\lambda_j(s)\tilde A_j(s)\hat X_q\cdots\hat X_p],
\end{align}
where the product $\hat X_{p+1}\cdots\hat X_{q-1}$ is the identity for $q=p+1$. By expanding $\mathcal J$ we get
\begin{align}\nonumber
    \mathsf z_{1}^{[1]}[\hat x_A,\mathbbm 1_E]\simeq&~ -i\sum_{q=p+1}^n\sum_j\int_{\tau_{q-1}}^{\tau_q}\!ds\braket{\tilde E_{\alpha}(\tau_p)\tilde E_j(s)}_{\hat\sigma_E} \mathcal F_q[\tilde \Phi_{\tau_{q},s}[[\lambda_j(s)\tilde A_j(s),\hat S_p \hat X_{p+1}\cdots\hat X_{q-1}]]\hat X_{q}]\\
    & -i\sum_{q=1}^p\sum_j\int_{\tau_{q-1}}^{\tau_q}\!ds\braket{\tilde E_{\alpha}(\tau_p)\tilde E_j(s)}_{\hat\sigma_E} \mathcal F_p[\lambda_j(s)\tilde A_j(s)\tilde\Phi_{s,\tau_{q-1}}[\hat S_{q-1}]\hat X_q\cdots\hat X_p]
    \nonumber\\
    & +i\sum_{q=1}^p\sum_j\int_{\tau_{q-1}}^{\tau_q}\!ds \braket{\tilde E_j(s)\tilde E_{\alpha}(\tau_p)}_{\hat\sigma_E}\mathcal F_p[\tilde\Phi_{s,\tau_{q-1}}[\hat S_{q-1}]\lambda_j(s)\tilde A_j(s)\hat X_q\cdots\hat X_p].
\end{align}

The bath correlation functions localize the interaction time $s$ to a neighborhood of $\tau_p$ of width $\tau_E$. This localization must be applied over the complete operator-time
hypercube, since no ordering is imposed on
$\tau_1,\ldots,\tau_n$. Using the physical recursion intervals
$I_q = [\min\{\tau_{q-1},\tau_q\},\max\{\tau_{q-1},\tau_q\}]$, we introduce the sets
\begin{align}
\label{eq:out_of_order_sets}
    \mathcal O^{(p)}_{<}(\tau_p)
    &=
    \left\{
        q\in\{1,\ldots,p-1\}
        \,\middle|\,
        \tau_p\in I_q
    \right\},
    \\
    \mathcal O^{(p)}_{>}(\tau_p)
    &=
    \left\{
        q\in\{p+2,\ldots,n\}
        \,\middle|\,
        \tau_p\in I_q
    \right\}.
\end{align}
The set $\mathcal O^{(p)}_{<}(\tau_p)$ contains the nonadjacent recursion intervals generated before the external environment operator enters the boundary word whose physical-time interval contains $\tau_p$. Similarly, $\mathcal O^{(p)}_{>}(\tau_p)$ contains the nonadjacent recursion intervals generated after the external insertion whose physical-time interval contains $\tau_p$. We also define $\epsilon_q = {\rm sgn}(\tau_q-\tau_{q-1})$, which keeps track of the orientation already present in the integral $\int_{\tau_{q-1}}^{\tau_q}ds$.

In a chronologically ordered region of the hypercube, both
$\mathcal O^{(p)}_{<}(\tau_p)$ and $\mathcal O^{(p)}_{>}(\tau_p)$ are empty.
The usual ordered-time argument then shows that only the adjacent
intervals $q=p$ and $q=p+1$ contribute at order
$\mathcal O(\tau_E)$. Indeed, a nonadjacent contribution can remain
non-negligible only if the intermediate ordered times are confined
to a window of width $\mathcal O(\tau_E)$, which supplies at least
one additional power of $\tau_E$.

Over the complete hypercube, however, the sets
$\mathcal O^{(p)}_{<}(\tau_p)$ and $\mathcal O^{(p)}_{>}(\tau_p)$ need not be
empty. Whenever
$q\in\mathcal O^{(p)}_{<}(\tau_p)\cup\mathcal O^{(p)}_{>}(\tau_p)$, the
correlation peak $s=\tau_p$ lies in the interior of the corresponding
recursion interval. The associated contribution is then of order
$\mathcal O(\tau_E)$ even though the interval is nonadjacent.
A nonadjacent interval that does not contain $\tau_p$ contributes
only when one of its endpoints lies within a distance of order
$\tau_E$ from $\tau_p$. Since this condition restricts one of the
outer hypercube integrations to a region of width
$\mathcal O(\tau_E)$, such boundary contributions are of order
$\mathcal O(\tau_E^2)$ and may be discarded.

Consequently, the leading single-insertion contribution is
\begin{align}
\label{eq:m_one_adjacent_and_OO}
    \mathsf z_{1}^{[1]}
    [\hat x_A,\mathbbm 1_E]
    \simeq{}&
    -i\sum_j
    \int_{\tau_p}^{\tau_{p+1}}\!ds\,
    C_{j\alpha}^{*}(s-\tau_p)
    \mathcal F_{p+1}\!\left[
        \tilde\Phi_{\tau_{p+1},s}
        \left[
            \left[
                \lambda_j(s)\tilde A_j(s),
                \hat S_p
            \right]
        \right]
        \hat X_{p+1}
    \right]
    \nonumber\\
    &-
    i\sum_j
    \int_{\tau_{p-1}}^{\tau_p}\!ds\,
    C_{\alpha j}(\tau_p-s)
    \mathcal F_p\!\left[
        \lambda_j(s)\tilde A_j(s)
        \tilde\Phi_{s,\tau_{p-1}}[\hat S_{p-1}]
        \hat X_p
    \right]
    \nonumber\\
    &+
    i\sum_j
    \int_{\tau_{p-1}}^{\tau_p}\!ds\,
    C_{\alpha j}^{*}(\tau_p-s)
    \mathcal F_p\!\left[
        \tilde\Phi_{s,\tau_{p-1}}[\hat S_{p-1}]
        \lambda_j(s)\tilde A_j(s)
        \hat X_p
    \right]
    \nonumber\\
    &-
    i\sum_{q\in\mathcal O^{(p)}_{>}(\tau_p)}
    \sum_j
    \int_{\tau_{q-1}}^{\tau_q}\!ds\,
    C_{j\alpha}^{*}(s-\tau_p)
    \mathcal F_q\!\left[
        \tilde\Phi_{\tau_q,s}
        \left[
            \left[
                \lambda_j(s)\tilde A_j(s),
                \hat S_p
                \hat X_{p+1}\cdots\hat X_{q-1}
            \right]
        \right]
        \hat X_q
    \right]
    \nonumber\\
    &-
    i\sum_{q\in\mathcal O^{(p)}_{<}(\tau_p)}
    \sum_j
    \int_{\tau_{q-1}}^{\tau_q}\!ds\,
    C_{\alpha j}(\tau_p-s)
    \mathcal F_p\!\left[
        \lambda_j(s)\tilde A_j(s)
        \tilde\Phi_{s,\tau_{q-1}}[\hat S_{q-1}]
        \hat X_q\cdots\hat X_p
    \right]
    \nonumber\\
    &+
    i\sum_{q\in\mathcal O^{(p)}_{<}(\tau_p)}
    \sum_j
    \int_{\tau_{q-1}}^{\tau_q}\!ds\,
    C_{\alpha j}^{*}(\tau_p-s)
    \mathcal F_p\!\left[
        \tilde\Phi_{s,\tau_{q-1}}[\hat S_{q-1}]
        \lambda_j(s)\tilde A_j(s)
        \hat X_q\cdots\hat X_p
    \right]
    +
    \mathcal O(\tau_E^2).
\end{align}
The first three terms are generated by the two intervals adjacent to
the external insertion. The remaining three terms are the additional
out-of-order contributions associated with
$\mathcal O^{(p)}_{>}(\tau_p)$ and $\mathcal O^{(p)}_{<}(\tau_p)$.

For the adjacent intervals, the Markov approximation gives the
one-sided maps
\begin{align}
\label{eq:K_R_maps}
    \tilde{\mathcal K}_{\alpha,t}[\hat O_A]
    &=
    -i\sum_{j\omega}\lambda_j(t)
    \left(
        \Gamma_{\alpha j}(\omega)
        \tilde A_j(\omega)\hat O_A
        -
        \Gamma_{\alpha j}^{*}(\omega)
        \hat O_A\tilde A_j^\dagger(\omega)
    \right),
    \\
    \tilde{\mathcal R}_{\alpha,t}[\hat O_A]
    &=
    -i\sum_{j\omega}\lambda_j(t)
    \Gamma_{j\alpha}^{*}(\omega)
    [\tilde A_j(\omega),\hat O_A].
\end{align}
For an out-of-order interval, the correlation peak lies in the
interior of the integration domain, and the Markov approximation
samples the full two-sided bath correlation function. We define the corresponding out-of-order maps
\begin{align}
\label{eq:K_R_OO_maps}
    \tilde{\mathcal K}_{\alpha,t}^{\mathrm{OO}}[\hat O_A]
    &=
    -i\sum_{j\omega}\lambda_j(t)
    \left(
        \gamma_{\alpha j}(\omega)
        \tilde A_j(\omega)\hat O_A
        -
        \gamma_{\alpha j}^{*}(\omega)
        \hat O_A\tilde A_j^\dagger(\omega)
    \right),
    \\
    \tilde{\mathcal R}_{\alpha,t}^{\mathrm{OO}}[\hat O_A]
    &=
    -i\sum_{j\omega}\lambda_j(t)
    \gamma_{j\alpha}^{*}(\omega)
    [\tilde A_j(\omega),\hat O_A].
\end{align}
Applying the Markov approximation to
Eq.~\eqref{eq:m_one_adjacent_and_OO}, we obtain the complete solution
for the $m=1$ sector,
\begin{align}
\label{eq:solution_m_one}
    \mathsf z_{1}^{[1]}
    [\hat x_A,\mathbbm 1_E]
    \simeq{}&
    \mathcal F_p\!\left[
        \tilde{\mathcal K}_{\alpha,\tau_p}
        \!\left[
            \tilde\Phi_{\tau_p,\tau_{p-1}}
            [\hat S_{p-1}]
        \right]
        \hat X_p
        +
        \tilde{\mathcal R}_{\alpha,\tau_p}[\hat S_p]
    \right]
    \nonumber\\
    &+
    \sum_{q\in\mathcal O^{(p)}_{<}(\tau_p)}
    \epsilon_q\,
    \mathcal F_p\!\left[
        \tilde{\mathcal K}_{\alpha,\tau_p}^{\mathrm{OO}}
        \!\left[
            \tilde\Phi_{\tau_p,\tau_{q-1}}
            [\hat S_{q-1}]
        \right]
        \hat X_q\cdots\hat X_p
    \right]
    \nonumber\\
    &+
    \sum_{q\in\mathcal O^{(p)}_{>}(\tau_p)}
    \epsilon_q\,
    \mathcal F_q\!\left[
        \tilde\Phi_{\tau_q,\tau_p}
        \!\left[
            \tilde{\mathcal R}_{\alpha,\tau_p}^{\mathrm{OO}}
            \!\left[
                \hat S_p
                \hat X_{p+1}\cdots\hat X_{q-1}
            \right]
        \right]
        \hat X_q
    \right]
    +
    \mathcal O(\tau_E^2).
\end{align}
The first line is the local contribution that survives in a
chronologically ordered region. The second and third lines are the
additional out-of-order terms. They vanish when
$\mathcal O^{(p)}_{<}(\tau_p)=\mathcal O^{(p)}_{>}(\tau_p)=\varnothing$, but are
generally present when the operator times are integrated over the
complete hypercube.

Equivalently, writing the future regression and the ordinary
regression prefixes explicitly, the complete $m=1$ solution reads
\begin{multline}
\label{eq:solution_m_one_explicit}
    \mathsf z_{1}^{[1]}[\hat x_A,\mathbbm 1_E]\simeq
    \Tr_A\!\Big[
        \mathsf R_{\hat X_n}
        \tilde\Phi_{\tau_n,\tau_{n-1}}
        \cdots
        \mathsf R_{\hat X_{p+1}}
        \tilde\Phi_{\tau_{p+1},\tau_p}
        \left(
            \mathsf R_{\hat X_p}\tilde{\mathcal K}_{\alpha,\tau_p}
            +
            \tilde{\mathcal R}_{\alpha,\tau_p}\mathsf R_{\hat X_p}
        \right)
        \tilde\Phi_{\tau_p,\tau_{p-1}}
        \mathsf R_{\hat X_{p-1}}
        \tilde\Phi_{\tau_{p-1},\tau_{p-2}}
        \cdots
        \mathsf R_{\hat X_1}
        \tilde\Phi_{\tau_1,t_0}
        [\hat x_A]
    \Big]\\
    +
    \sum_{q\in\mathcal O^{(p)}_{<}(\tau_p)}\!\!\!
    \epsilon_q\,
    \Tr_A\!\Big[
        \mathsf R_{\hat X_n}
        \tilde\Phi_{\tau_n,\tau_{n-1}}
        \cdots
        \mathsf R_{\hat X_{p+1}}
        \tilde\Phi_{\tau_{p+1},\tau_p}
        \mathsf R_{\hat X_p}\cdots\mathsf R_{\hat X_q}
        \tilde{\mathcal K}_{\alpha,\tau_p}^{\mathrm{OO}}
        \tilde\Phi_{\tau_p,\tau_{q-1}}
        \mathsf R_{\hat X_{q-1}}
        \tilde\Phi_{\tau_{q-1},\tau_{q-2}}
        \cdots
        \mathsf R_{\hat X_1}
        \tilde\Phi_{\tau_1,t_0}
        [\hat x_A]
    \Big]\\
    +
    \sum_{q\in\mathcal O^{(p)}_{>}(\tau_p)}\!\!\!
    \epsilon_q\,
    \Tr_A\!\Big[
        \mathsf R_{\hat X_n}
        \tilde\Phi_{\tau_n,\tau_{n-1}}
        \cdots
        \mathsf R_{\hat X_{q+1}}
        \tilde\Phi_{\tau_{q+1},\tau_q}
        \mathsf R_{\hat X_q}
        \tilde\Phi_{\tau_q,\tau_p}
        \tilde{\mathcal R}_{\alpha,\tau_p}^{\mathrm{OO}}
        \mathsf R_{\hat X_{q-1}}\cdots\mathsf R_{\hat X_p}
        \tilde\Phi_{\tau_p,\tau_{p-1}}
        \mathsf R_{\hat X_{p-1}}
        \tilde\Phi_{\tau_{p-1},\tau_{p-2}}
        \cdots
        \mathsf R_{\hat X_1}
        \tilde\Phi_{\tau_1,t_0}
        [\hat x_A]
    \Big].
\end{multline}

Finally, we return to the $m=0$ sector and compute the contribution
containing exactly two interaction-generated environment operators.
This is the leading correction to the ordinary quantum regression
theorem. Since there are no external environment insertions, one has
$\hat Y_i=\mathbbm 1_E$ for every $i$. Before either interaction
operator is generated, the boundary word is trivial and the
inhomogeneous contribution
$\mathcal J^{(L,R)}_{\tau_i,\tau_{i-1}}$ vanishes. Selecting from
Eq.~\eqref{eq:recursion} the terms containing exactly two generated
environment operators therefore gives
\begin{align}
\label{eq:two_vertex_recursion_m_zero}
    \mathsf z_i^{[2]}[\hat O_A,\mathbbm 1_E]
    \simeq{}&
    \mathsf z_{i+1}^{[2]}
    \!\left[
        \tilde\Phi_{\tau_i,\tau_{i-1}}[\hat O_A]\hat X_i,
        \mathbbm 1_E
    \right]
    \nonumber\\
    &-
    i\sum_j\int_{\tau_{i-1}}^{\tau_i}\!ds\,
    \mathsf z_{i+1}^{[1]}
    \!\left[
        \lambda_j(s)\tilde A_j(s)
        \tilde\Phi_{s,\tau_{i-1}}[\hat O_A]\hat X_i,
        \tilde E_j(s),
        \mathbbm 1_E
    \right]
    \nonumber\\
    &+
    i\sum_j\int_{\tau_{i-1}}^{\tau_i}\!ds\,
    \mathsf z_{i+1}^{[1]}
    \!\left[
        \tilde\Phi_{s,\tau_{i-1}}[\hat O_A]
        \lambda_j(s)\tilde A_j(s)\hat X_i,
        \mathbbm 1_E,
        \tilde E_j(s)
    \right].
\end{align}
The first term carries the unresolved two-vertex branch to the next recursion step and generates the ordinary regression prefix. The last two terms select the current recursion interval as the interval in which the first interaction-generated environment operator is introduced.

For later convenience, define
\begin{align}
\label{eq:first_generated_left_right}
    \hat B_{j,q}^{\rm L}(s)
    &=
    \lambda_j(s)\tilde A_j(s)
    \tilde\Phi_{s,\tau_{q-1}}[\hat S_{q-1}]
    \hat X_q,
    \\
    \hat B_{j,q}^{\rm R}(s)
    &=
    \tilde\Phi_{s,\tau_{q-1}}[\hat S_{q-1}]
    \lambda_j(s)\tilde A_j(s)
    \hat X_q.
\end{align}
Iterating the first term of Eq.~\eqref{eq:two_vertex_recursion_m_zero} until the the stopping condition yields
\begin{align}
\label{eq:first_vertex_m_zero}
    \mathsf z_1^{[2]}[\hat x_A,\mathbbm 1_E]
    \simeq{}&
    -i\sum_{q=1}^{n-1}\sum_j
    \int_{\tau_{q-1}}^{\tau_q}\!ds\,
    \mathsf z_{q+1}^{[1]}
    \!\left[
        \hat B_{j,q}^{\rm L}(s),
        \tilde E_j(s),
        \mathbbm 1_E
    \right]
    \nonumber\\
    &+
    i\sum_{q=1}^{n-1}\sum_j
    \int_{\tau_{q-1}}^{\tau_q}\!ds\,
    \mathsf z_{q+1}^{[1]}
    \!\left[
        \hat B_{j,q}^{\rm R}(s),
        \mathbbm 1_E,
        \tilde E_j(s)
    \right].
\end{align}
The upper limit $n-1$ follows because of applying the stopping condition to $\mathsf z^{[1]}_{q+1}$ at $q=n$.

After the first interaction operator has been generated, the boundary word contains one centered environment operator. Hence, $\braket{\hat R_E\hat L_E}_{\hat\sigma_E}=0$. The last two source terms of Eq.~\eqref{eq:recursion} therefore vanish, while the odd-parity term $\mathcal J^{(L,R)}$ is nonzero. The recursion for the remaining interaction-generated operator reduces to
\begin{equation}
\label{eq:one_remaining_generated_operator}
    \mathsf z_i^{[1]}[\hat O_A,\hat L_E,\hat R_E]
    \simeq
    \mathsf z_{i+1}^{[1]}
    [\hat O_A\hat X_i,\hat L_E,\hat R_E]
    +
    \mathsf z_{i+1}^{[0]}
    \!\left[
        \mathcal J_{\tau_i,\tau_{i-1}}^{(L,R)}
        [\hat O_A]\hat X_i,
        \mathbbm 1_E
    \right].
\end{equation}
The first term carries the unmatched environment operator to the next recursion step. The second term generates the remaining interaction operator through the generalized propagator and closes the environment contraction. Iterating Eq.~\eqref{eq:one_remaining_generated_operator} and using the terminal condition  gives, for any $q<n$,
\begin{align}
\label{eq:one_remaining_generated_solution}
    \mathsf z_{q+1}^{[1]}
    [\hat O_A,\hat L_E,\hat R_E]
    \simeq
    \sum_{r=q+1}^{n}
    \mathcal F_r\!\left[
        \mathcal J_{\tau_r,\tau_{r-1}}^{(L,R)}
        \!\left[
            \hat O_A
            \hat X_{q+1}\cdots\hat X_{r-1}
        \right]
        \hat X_r
    \right].
\end{align}
The product $\hat X_{q+1}\cdots\hat X_{r-1}$ is understood as the identity when $r=q+1$.

Substituting Eq.~\eqref{eq:one_remaining_generated_solution} into Eq.~\eqref{eq:first_vertex_m_zero}, we obtain
\begin{align}
\label{eq:two_vertex_solution_J}
    \mathsf z_1^{[2]}[\hat x_A,\mathbbm 1_E]
    \simeq{}&
    -i\sum_{1\leq q<r\leq n}\sum_j
    \int_{\tau_{q-1}}^{\tau_q}\!ds\,
    \mathcal F_r\!\left[
        \mathcal J_{\tau_r,\tau_{r-1}}^{(E_j(s),1)}
        \!\left[
            \hat B_{j,q}^{\rm L}(s)
            \hat X_{q+1}\cdots\hat X_{r-1}
        \right]
        \hat X_r
    \right]
    \nonumber\\
    &+
    i\sum_{1\leq q<r\leq n}\sum_j
    \int_{\tau_{q-1}}^{\tau_q}\!ds\,
    \mathcal F_r\!\left[
        \mathcal J_{\tau_r,\tau_{r-1}}^{(1,E_j(s))}
        \!\left[
            \hat B_{j,q}^{\rm R}(s)
            \hat X_{q+1}\cdots\hat X_{r-1}
        \right]
        \hat X_r
    \right].
\end{align}
This expression displays explicitly the two distinct possibilities for generating the first interaction operator: to the left or to the right of the reduced system operator. In either case, the second operator is generated through the odd-parity contribution to the generalized propagator at a later recursion step.

Expanding $\mathcal J$ according to Eq.~\eqref{eq:odd_phi_source}, Eq.~\eqref{eq:two_vertex_solution_J} becomes
\begin{align}
\label{eq:two_vertex_solution_expanded}
    \mathsf z_1^{[2]}[\hat x_A,\mathbbm 1_E]\simeq&-\sum_{1\leq q<r\leq n}
    \sum_{jj'}
    \int_{\tau_{q-1}}^{\tau_q}\!ds
    \int_{\tau_{r-1}}^{\tau_r}\!du\,
    C_{j'j}(u-s)
    \mathcal F_r\!\left[
        \tilde\Phi_{\tau_r,u}
        \!\left[
            \left[
                \lambda_{j'}(u)\tilde A_{j'}(u),
                \hat B_{j,q}^{\rm L}(s)
                \hat X_{q+1}\cdots\hat X_{r-1}
            \right]
        \right]
        \hat X_r
    \right]\\
    &+
    \sum_{1\leq q<r\leq n}
    \sum_{jj'}
    \int_{\tau_{q-1}}^{\tau_q}\!ds
    \int_{\tau_{r-1}}^{\tau_r}\!du\,
    C_{jj'}(s-u)
    \nonumber
    \mathcal F_r\!\left[
        \tilde\Phi_{\tau_r,u}
        \!\left[
            \left[
                \lambda_{j'}(u)\tilde A_{j'}(u),
                \hat B_{j,q}^{\rm R}(s)
                \hat X_{q+1}\cdots\hat X_{r-1}
            \right]
        \right]
        \hat X_r
    \right].
\end{align}

The two bath operators generated at steps $q$ and $r$ are integrated
over the physical recursion intervals $I_q$ and $I_r$, respectively.
Their contraction is supported only when $|u-s|\lesssim\tau_E$.
Accordingly, define the set of recursion-interval pairs with a
finite-length overlap,
\begin{equation}
\label{eq:overlapping_interval_pairs}
    \mathcal P_{\rm ov}
    =
    \left\{
        (q,r)
        \,\middle|\,
        1\leq q<r\leq n,\ 
        \left|I_q\cap I_r\right|>0
    \right\},
\end{equation}
where $\left|I_q\cap I_r\right|$ denotes the length of the
intersection.

If $(q,r)\in\mathcal P_{\rm ov}$, one integration variable may range
over the finite overlap $I_q\cap I_r$, while the correlation function
localizes the relative variable $u-s$ to a region of width
$\tau_E$. The corresponding contribution is therefore of order
$\mathcal O(g^2\tau_E)$. If the two intervals meet only at a common
endpoint, both $s$ and $u$ must lie within a distance of order
$\tau_E$ from that endpoint, and the contribution is instead of
order $\mathcal O(g^2\tau_E^2)$. Separated intervals give no
contribution in the Markov limit, except in boundary regions of the
outer hypercube whose additional measure again makes them
$\mathcal O(g^2\tau_E^2)$. Therefore we can write
\begin{align}
\label{eq:QRT_two_vertex_correction}
    \mathsf z_1^{[2]}[\hat x_A,\mathbbm 1_E]
    \simeq&
    -\sum_{(q,r)\in\mathcal P_{\rm ov}}
    \sum_{jj'}
    \int_{\tau_{q-1}}^{\tau_q}\!ds
    \int_{\tau_{r-1}}^{\tau_r}\!du\,
    C_{jj'}^*(s-u)
    \mathcal F_r\!\left[
        \tilde\Phi_{\tau_r,u}
        \!\left[
            \left[
                \lambda_{j'}(u)\tilde A_{j'}(u),
                \hat B_{j,q}^{\rm L}(s)
                \hat X_{q+1}\cdots\hat X_{r-1}
            \right]
        \right]
        \hat X_r
    \right]
    \nonumber\\
    &+
    \sum_{(q,r)\in\mathcal P_{\rm ov}}
    \sum_{jj'}
    \int_{\tau_{q-1}}^{\tau_q}\!ds
    \int_{\tau_{r-1}}^{\tau_r}\!du\,
    C_{jj'}(s-u)
    \mathcal F_r\!\left[
        \tilde\Phi_{\tau_r,u}
        \!\left[
            \left[
                \lambda_{j'}(u)\tilde A_{j'}(u),
                \hat B_{j,q}^{\rm R}(s)
                \hat X_{q+1}\cdots\hat X_{r-1}
            \right]
        \right]
        \hat X_r
    \right].
\end{align}
After applying the Markov and secular approximations (as in the GKLS procedure), the two-vertex
correction becomes
\begin{align}
\label{eq:QRT_two_vertex_correction_markov}
    \mathsf z_1^{[2]}[\hat x_A,\mathbbm 1_E]
    \simeq{}&
    -\sum_{(q,r)\in\mathcal P_{\rm ov}}
    \epsilon_q\epsilon_r
    \sum_{jj'\omega}
    \int_{I_q\cap I_r}\!\!\!ds\,
    \lambda_j(s)\lambda_{j'}(s)\gamma_{jj'}^*(\omega)
    \mathcal F_r\!\left[
        \tilde\Phi_{\tau_r,s}\!\left[
            \left[
                \tilde A_{j'}^\dagger(\omega),
                \tilde A_j(\omega)
                \tilde\Phi_{s,\tau_{q-1}}[\hat S_{q-1}]
                \hat X_{q:r-1}
            \right]
        \right]\hat X_r
    \right]
    \nonumber\\
    &+
    \sum_{(q,r)\in\mathcal P_{\rm ov}}
    \epsilon_q\epsilon_r
    \sum_{jj'\omega}
    \int_{I_q\cap I_r}\!\!\!ds\,
    \lambda_j(s)\lambda_{j'}(s)\gamma_{jj'}(\omega)
    \mathcal F_r\!\left[
        \tilde\Phi_{\tau_r,s}\!\left[
            \left[
                \tilde A_{j'}(\omega),
                \tilde\Phi_{s,\tau_{q-1}}[\hat S_{q-1}]
                \tilde A_j^\dagger(\omega)
                \hat X_{q:r-1}
            \right]
        \right]\hat X_r
    \right],
\end{align}
where $\hat X_{a:b} = \hat X_a\hat X_{a+1}\cdots\hat X_b$, with $\hat X_{a:b}=\mathbbm 1_A$ whenever $a>b$.  The interaction-picture GKLS dissipator and its adjoint are given by
\begin{align}
\label{eq:interaction_picture_dissipator}
    \tilde{\mathcal D}_s[\hat O_A]
    &=
    \sum_{jj'\omega}
    \lambda_j(s)\lambda_{j'}(s)\gamma_{jj'}(\omega)
    \left(
        \tilde A_{j'}(\omega)
        \hat O_A
        \tilde A_j^\dagger(\omega)
        -
        \frac{1}{2}
        \left\{
            \tilde A_j^\dagger(\omega)
            \tilde A_{j'}(\omega),
            \hat O_A
        \right\}
    \right),
    \\
\label{eq:interaction_picture_dissipator_adjoint}
    \tilde{\mathcal D}_s^\dagger[\hat O_A]
    &=
    \sum_{jj'\omega}
    \lambda_j(s)\lambda_{j'}(s)\gamma_{jj'}(\omega)
    \left(
        \tilde A_j^\dagger(\omega)
        \hat O_A
        \tilde A_{j'}(\omega)
        -
        \frac{1}{2}
        \left\{
            \tilde A_j^\dagger(\omega)
            \tilde A_{j'}(\omega),
            \hat O_A
        \right\}
    \right).
\end{align}
A direct expansion of the two commutators in
Eq.~\eqref{eq:QRT_two_vertex_correction_markov}, together with the
Hermiticity relation
$\gamma_{jj'}^*(\omega)=\gamma_{j'j}(\omega)$, shows that their sum
can be written as
\begin{equation}
\label{eq:QRT_two_vertex_dissipator_identity}
    \tilde{\mathcal D}_s
    \!\left[
        \tilde\Phi_{s,\tau_{q-1}}[\hat S_{q-1}]
        \hat X_{q:r-1}
    \right]
    +
    \tilde{\mathcal D}_s
    \!\left[
        \tilde\Phi_{s,\tau_{q-1}}[\hat S_{q-1}]
    \right]
    \hat X_{q:r-1}
    -
    \tilde\Phi_{s,\tau_{q-1}}[\hat S_{q-1}]
    \tilde{\mathcal D}_s^\dagger
    [\hat X_{q:r-1}].
\end{equation}
Consequently, the two-vertex contribution obtained directly from
the microscopic recursion is
\begin{multline}
\label{eq:QRT_two_vertex_correction_physical}
    \delta z_{\rm QRT}^{[2]}[\hat x_A]
    \simeq
    \sum_{(q,r)\in\mathcal P_{\rm ov}}
    \epsilon_q\epsilon_r
    \int_{I_q\cap I_r}\!ds\,
    \mathcal F_r\!\Bigg[
        \tilde\Phi_{\tau_r,s}\!\Bigg[
            \tilde{\mathcal D}_s
            \!\left[
                \tilde\Phi_{s,\tau_{q-1}}[\hat S_{q-1}]
                \hat X_{q:r-1}
            \right]\\
            +
            \tilde{\mathcal D}_s
            \!\left[
                \tilde\Phi_{s,\tau_{q-1}}[\hat S_{q-1}]
            \right]
            \hat X_{q:r-1}
            -
            \tilde\Phi_{s,\tau_{q-1}}[\hat S_{q-1}]
            \tilde{\mathcal D}_s^\dagger
            [\hat X_{q:r-1}]
        \Bigg]
        \hat X_r
    \Bigg].
\end{multline}
The reduced propagators appearing up to this point arise directly
from the microscopic recursion,
\begin{equation}
    \tilde\Phi_{t,t'}^{\rm phys}[\hat O_A]
    =
    \Tr_E\!\left[
        \hat U_{t,t'}^{(I)}
        (\hat O_A\otimes\hat\sigma_E)
        \hat U_{t,t'}^{\dagger(I)}
    \right].
\end{equation}
For $t\geq t'$, this map agrees, at the working order of the
Born--Markov approximation, with the two-time GKLS propagator
\begin{equation}
    \tilde\Phi_{t,t'}^{\rm GKLS}
    =
    \tilde\Phi_{t,t_0}
    \circ
    \tilde\Phi_{t',t_0}^{-1}.
\end{equation}
For $t<t'$, however, the microscopic map describes physical reverse
evolution, whereas the GKLS expression is the inverse of the forward
reduced evolution. If $u$ denotes the positive duration of the
backward interval, their generators are respectively
\begin{equation}
    \frac{d}{du}\tilde\Phi^{\rm phys}
    =
    \left(
        -\tilde{\mathcal H}
        +
        \tilde{\mathcal D}
    \right)
    \circ\tilde\Phi^{\rm phys},
    \qquad
    \frac{d}{du}\tilde\Phi^{\rm GKLS}
    =
    -
    \left(
        \tilde{\mathcal H}
        +
        \tilde{\mathcal D}
    \right)
    \circ\tilde\Phi^{\rm GKLS}.
\end{equation}
Equivalently, for $t<t'$ they satisfy
\begin{equation}
\label{eq:physical_to_GKLS_backward_propagator}
    \tilde\Phi_{t,t'}^{\rm phys}
    =
    \tilde\Phi_{t,t'}^{\rm GKLS}
    +
    2\int_t^{t'}\!du\,
    \tilde\Phi_{t,u}^{\rm phys}
    \circ
    \tilde{\mathcal D}_u
    \circ
    \tilde\Phi_{u,t'}^{\rm GKLS}.
\end{equation}
This distinction does not affect the $m=1$ and $m=2$ results at the
order retained here. Indeed, their leading contributions are already
of order $\mathcal O(g\tau_E)$ and $\mathcal O(\tau_E)$,
respectively, so that replacing a physical backward propagator by
its GKLS counterpart changes them only at order
$\mathcal O(g^3\tau_E^2)$ and
$\mathcal O(g^2\tau_E^2)$. In the $m=0$ sector, by contrast, the
zero-vertex regression term is of order one. The difference between
the two backward propagators therefore produces an additional
contribution of order $\mathcal O(g^2\tau_E)$, which must be combined
with Eq.~\eqref{eq:QRT_two_vertex_correction_physical}.

In the following, we return to the notation $\tilde\Phi_{t,t'}$,
which will now denote the two-time GKLS propagator for either
orientation of its time arguments. Accordingly, the zero-vertex
regression prefix is understood as
\begin{equation}
\label{eq:GKLS_regression_prefix}
    \hat S_0=\hat x_A,
    \qquad
    \hat S_i
    =
    \tilde\Phi_{\tau_i,\tau_{i-1}}[\hat S_{i-1}]
    \hat X_i,
    \qquad
    i=1,\ldots,n.
\end{equation}

Let
\begin{equation}
    \mathcal B_-
    :=
    \left\{
        r\in\{1,\ldots,n\}
        \,\middle|\,
        \epsilon_r=-1
    \right\}
\end{equation}
denote the set of backward recursion intervals. Only the propagators
associated with these intervals differ from the physical propagators
generated by the microscopic recursion. To the working order,
Eq.~\eqref{eq:physical_to_GKLS_backward_propagator} gives
\begin{equation}
\label{eq:physical_backward_interval_expansion}
    \tilde\Phi_{\tau_r,\tau_{r-1}}^{\rm phys}
    =
    \tilde\Phi_{\tau_r,\tau_{r-1}}
    +
    2\int_{\tau_r}^{\tau_{r-1}}\!ds\,
    \tilde\Phi_{\tau_r,s}
    \circ\tilde{\mathcal D}_s
    \circ\tilde\Phi_{s,\tau_{r-1}}
    +
    \mathcal O(\tau_E^2),
    \qquad r\in\mathcal B_-.
\end{equation}
Inserting this expansion once in the zero-vertex regression chain
therefore produces
\begin{align}
\label{eq:QRT_zero_vertex_conversion_backward_intervals}
    \delta z_{\rm prop}^{[0]}[\hat x_A]
    \simeq{}&
    2
    \sum_{r\in\mathcal B_-}
    \int_{\tau_r}^{\tau_{r-1}}\!ds\,
    \mathcal F_r\!\Bigg[
        \tilde\Phi_{\tau_r,s}\!\left[
            \tilde{\mathcal D}_s\!\left[
                \tilde\Phi_{s,\tau_{r-1}}[\hat S_{r-1}]
            \right]
        \right]
        \hat X_r
    \Bigg].
\end{align}
This is the direct correction obtained by replacing the physical
backward propagators in the zero-vertex term by the inverse GKLS
propagators.

To combine this contribution with the two-vertex correction, it is
useful to rewrite Eq.~\eqref{eq:QRT_zero_vertex_conversion_backward_intervals}
as a sum over overlapping recursion intervals. The oriented intervals
traversed before the $r$-th recursion step satisfy, away from their
measure-zero endpoints, $\sum_{q=1}^{r-1} \epsilon_q\,\mathbbm 1_{I_q}(s) = \mathbbm 1_{[t_0,\tau_{r-1}]}(s)$. Indeed, $\epsilon_q\,\mathbbm 1_{I_q}(s) = \Theta(\tau_q-s)-\Theta(\tau_{q-1}-s)$,
so that the sum telescopes from $t_0$ to $\tau_{r-1}$. Therefore
\begin{equation}
\label{eq:backward_interval_overlap_identity}
    \mathbbm 1_{\{\epsilon_r=-1\}}
    \mathbbm 1_{I_r}(s)
    =
    -
    \epsilon_r
    \sum_{q=1}^{r-1}
    \epsilon_q\,
    \mathbbm 1_{I_q\cap I_r}(s).
\end{equation}
For a backward interval, the right-hand side equals one throughout
$I_r$, whereas for a forward interval it vanishes.

Moreover, because the integrand in
Eq.~\eqref{eq:QRT_zero_vertex_conversion_backward_intervals} already
contains one dissipator, its incoming state may be evaluated at
zeroth dissipative order. On $I_q\cap I_r$, this gives
\begin{equation}
\label{eq:zero_vertex_incoming_state_rearrangement}
    \tilde\Phi_{s,\tau_{r-1}}[\hat S_{r-1}]
    =
    \tilde\Phi_{s,\tau_{q-1}}[\hat S_{q-1}]
    \hat X_{q:r-1}
    +
    \mathcal O(\tau_E).
\end{equation}
The error in Eq.~\eqref{eq:zero_vertex_incoming_state_rearrangement}
is acted upon by $\tilde{\mathcal D}_s$ and hence contributes only at
order $\mathcal O(\tau_E^2)$.

Using Eqs.~\eqref{eq:backward_interval_overlap_identity} and
\eqref{eq:zero_vertex_incoming_state_rearrangement},
Eq.~\eqref{eq:QRT_zero_vertex_conversion_backward_intervals} becomes
\begin{align}
\label{eq:QRT_zero_vertex_propagator_conversion}
    \delta z_{\rm prop}^{[0]}[\hat x_A]
    \simeq{}&
    -2
    \sum_{(q,r)\in\mathcal P_{\rm ov}}
    \epsilon_q\epsilon_r
    \int_{I_q\cap I_r}\!ds\,
    \mathcal F_r\!\Bigg[
        \tilde\Phi_{\tau_r,s}\!\Bigg[
            \tilde{\mathcal D}_s\!\left[
                \tilde\Phi_{s,\tau_{q-1}}[\hat S_{q-1}]
                \hat X_{q:r-1}
            \right]
        \Bigg]
        \hat X_r
    \Bigg].
\end{align}

All propagators in the already dissipative two-vertex contribution
may simultaneously be replaced by their GKLS counterparts, since
their difference would contribute only at order
$\mathcal O(\tau_E^2)$. Combining
Eqs.~\eqref{eq:QRT_two_vertex_correction_physical} and
\eqref{eq:QRT_zero_vertex_propagator_conversion}, we obtain a simple sign flip for one of the terms
\begin{multline}
\label{eq:QRT_OO_GKLS_bare}
    \delta z_{\rm QRT}^{\rm OO}[\hat x_A]
    \simeq
    \sum_{(q,r)\in\mathcal P_{\rm ov}}
    \epsilon_q\epsilon_r
    \int_{I_q\cap I_r}\!ds\,
    \mathcal F_r\!\Bigg[
        \tilde\Phi_{\tau_r,s}\!\Bigg[
            -
            \tilde{\mathcal D}_s\!\left[
                \tilde\Phi_{s,\tau_{q-1}}[\hat S_{q-1}]
                \hat X_{q:r-1}
            \right]\\
            +
            \tilde{\mathcal D}_s\!\left[
                \tilde\Phi_{s,\tau_{q-1}}[\hat S_{q-1}]
            \right]
            \hat X_{q:r-1}
            -
            \tilde\Phi_{s,\tau_{q-1}}[\hat S_{q-1}]
            \tilde{\mathcal D}_s^\dagger
            [\hat X_{q:r-1}]
        \Bigg]
        \hat X_r
    \Bigg].
\end{multline}
Combining Eq.~\eqref{eq:QRT_OO_GKLS_bare} with the zero-vertex
solution, the complete result in the $m=0$ sector, expressed in terms
of the two-time GKLS propagator, is
\begin{equation}
\label{eq:solution_m_zero_corrected_state}
    z_{k_1,\ldots,k_n}
    [\hat x_A,\mathbbm 1_E,\mathbbm 1_E](t_0)
    \simeq
    \Tr_A[\hat S_n]
    +
    \delta z_{\rm QRT}^{\rm OO}[\hat x_A]
    +
    \mathcal O(\tau_E^2),
\end{equation}
where $\hat S_n$ is constructed using the two-time GKLS propagators,
and $\delta z_{\rm QRT}^{\rm OO}[\hat x_A]$ is given by
Eq.~\eqref{eq:QRT_OO_GKLS_bare}.

To obtain the form given in the theorem of the OO term, we now determine the corresponding Heisenberg-picture operator
$\delta\hat z_{\rm QRT}^{\rm OO}$, defined by
\begin{equation}
\label{eq:QRT_OO_operator_definition}
    \delta z_{\rm QRT}^{\rm OO}[\hat x_A]
    =
    \Tr_A\!\left[
        \hat x_A\,
        \delta\hat z_{\rm QRT}^{\rm OO}
    \right].
\end{equation}
For this purpose, introduce the interaction-picture regression
strings
\begin{equation}
\label{eq:interaction_picture_regression_string}
    \tilde{\mathcal G}_{a:b}^{\dagger}
    =
    \tilde\Phi_{\tau_a,\tau_{a-1}}^\dagger
    \circ\mathsf L_{\hat X_a}
    \circ\cdots\circ
    \tilde\Phi_{\tau_b,\tau_{b-1}}^\dagger
    \circ\mathsf L_{\hat X_b},
    \qquad
    \tilde{\mathcal G}_{a:b}^{\dagger}
    =
    \operatorname{id}
    \quad\text{for }a>b.
\end{equation}
Using
$\Tr_A[\tilde{\mathcal D}_s[\hat\rho_A]\hat O_A]
=
\Tr_A[\hat\rho_A\tilde{\mathcal D}_s^\dagger[\hat O_A]]$,
the Heisenberg-picture out-of-time-order contribution is
\begin{align}
\label{eq:QRT_OO_GKLS_Heisenberg}
    \delta\hat z_{\rm QRT}^{\rm OO}
    \simeq{}&
    \sum_{(q,r)\in\mathcal P_{\rm ov}}
    \epsilon_q\epsilon_r
    \int_{I_q\cap I_r}\!ds\,
    \tilde{\mathcal G}_{1:q-1}^{\dagger}
    \circ
    \tilde\Phi_{s,\tau_{q-1}}^\dagger
    \circ
    \Big(
        \tilde{\mathcal D}_s^\dagger
        \circ\mathsf L_{\hat X_{q:r-1}}
        -
        \mathsf L_{
            \tilde{\mathcal D}_s^\dagger
            [\hat X_{q:r-1}]
        }
        -
        \mathsf L_{\hat X_{q:r-1}}
        \circ\tilde{\mathcal D}_s^\dagger
    \Big)
    \circ
    \tilde\Phi_{\tau_r,s}^\dagger
    \circ\mathsf L_{\hat X_r}
    \circ
    \tilde{\mathcal G}_{r+1:n}^{\dagger}
    [\mathbbm 1_A].
\end{align}
To evaluate the integral, define
\begin{align}
\label{eq:QRT_overlap_operators}
    \hat O_{q,r}(s)
    :=
    \tilde\Phi_{\tau_q,s}^\dagger
    \circ
    \mathsf L_{\hat X_q}
    \circ
    \tilde{\mathcal G}_{q+1:r-1}^{\dagger}
    [\mathbbm 1_A],\qquad
    \hat O_r(s)
    :=
    \tilde\Phi_{\tau_r,s}^\dagger
    \circ
    \mathsf L_{\hat X_r}
    \circ
    \tilde{\mathcal G}_{r+1:n}^{\dagger}
    [\mathbbm 1_A].
\end{align}
We then introduce the interpolating operator
\begin{equation}
\label{eq:QRT_overlap_interpolating_function}
    \hat F_{q,r}(s)
    :=
    \tilde{\mathcal G}_{1:q-1}^{\dagger}
    \circ
    \tilde\Phi_{s,\tau_{q-1}}^\dagger
    \left[
        \hat O_{q,r}(s)\hat O_r(s)
    \right].
\end{equation}
Since the two-time maps are now the GKLS propagators, we have
\begin{equation}
\label{eq:QRT_overlap_operator_derivatives}
    \frac{d}{ds}\hat O_{q,r}(s)
    =
    -
    \tilde{\mathcal L}_s^\dagger
    [\hat O_{q,r}(s)],
    \qquad
    \frac{d}{ds}\hat O_r(s)
    =
    -
    \tilde{\mathcal L}_s^\dagger
    [\hat O_r(s)].
\end{equation}
Differentiating Eq.~\eqref{eq:QRT_overlap_interpolating_function}
and using Eq.~\eqref{eq:QRT_overlap_operator_derivatives}, we obtain
\begin{align}
\label{eq:QRT_overlap_interpolating_derivative}
    \frac{d}{ds}\hat F_{q,r}(s)
    =
    \tilde{\mathcal G}_{1:q-1}^{\dagger}
    \circ
    \tilde\Phi_{s,\tau_{q-1}}^\dagger
    \Big[
        \tilde{\mathcal L}_s^\dagger
        [\hat O_{q,r}(s)\hat O_r(s)]
        -
        \tilde{\mathcal L}_s^\dagger
        [\hat O_{q,r}(s)]
        \hat O_r(s)
        -
        \hat O_{q,r}(s)
        \tilde{\mathcal L}_s^\dagger
        [\hat O_r(s)]
    \Big].
\end{align}
Writing $\tilde{\mathcal L}_s^\dagger = \tilde{\mathcal H}_s^\dagger + \tilde{\mathcal D}_s^\dagger$, the Hamiltonian terms cancel because
$\tilde{\mathcal H}_s^\dagger$ satisfies the Leibniz rule,
\begin{equation}
    \tilde{\mathcal H}_s^\dagger
    [\hat O_1\hat O_2]
    =
    \tilde{\mathcal H}_s^\dagger[\hat O_1]\hat O_2
    +
    \hat O_1\tilde{\mathcal H}_s^\dagger[\hat O_2].
\end{equation}
Therefore,
\begin{align}
\label{eq:QRT_overlap_interpolating_derivative_dissipative}
    \frac{d}{ds}\hat F_{q,r}(s)
    ={}&
    \tilde{\mathcal G}_{1:q-1}^{\dagger}
    \circ
    \tilde\Phi_{s,\tau_{q-1}}^\dagger
    \Big[
        \tilde{\mathcal D}_s^\dagger
        [\hat O_{q,r}(s)\hat O_r(s)]
        -
        \tilde{\mathcal D}_s^\dagger
        [\hat O_{q,r}(s)]
        \hat O_r(s)
        -
        \hat O_{q,r}(s)
        \tilde{\mathcal D}_s^\dagger
        [\hat O_r(s)]
    \Big].
\end{align}
Comparing Eq.~\eqref{eq:QRT_overlap_interpolating_derivative_dissipative}
with the integrand of Eq.~\eqref{eq:QRT_OO_GKLS_Heisenberg}, we
immediately obtain
\begin{equation}
\label{eq:QRT_OO_as_total_derivative}
    \delta\hat z_{\rm QRT}^{\rm OO}
    \simeq
    \sum_{(q,r)\in\mathcal P_{\rm ov}}
    \epsilon_q\epsilon_r
    \int_{I_q\cap I_r}\!ds\,
    \frac{d}{ds}\hat F_{q,r}(s).
\end{equation}
Writing $I_q\cap I_r=[s_{q,r}^-,s_{q,r}^+]$, the integral reduces directly to
\begin{equation}
\label{eq:QRT_OO_GKLS_final}
    \delta\hat z_{\rm QRT}^{\rm OO}
    \simeq
    \sum_{(q,r)\in\mathcal P_{\rm ov}}
    \epsilon_q\epsilon_r
    \left[
        \hat F_{q,r}(s_{q,r}^+)
        -
        \hat F_{q,r}(s_{q,r}^-)
    \right].
\end{equation}
Substituting the definition of $\hat F_{q,r}(s)$ gives
\begin{equation}
\label{eq:QRT_OO_GKLS_final_expanded}
    \delta\hat z_{\rm QRT}^{\rm OO}
    \simeq
    \sum_{(q,r)\in\mathcal P_{\rm ov}}
    \epsilon_q\epsilon_r
    \left.
    \tilde{\mathcal G}_{1:q-1}^{\dagger}
    \circ
    \tilde\Phi_{s,\tau_{q-1}}^\dagger
    \left[
        \left(
            \tilde\Phi_{\tau_q,s}^\dagger
            \circ\mathsf L_{\hat X_q}
            \circ\tilde{\mathcal G}_{q+1:r-1}^{\dagger}
            [\mathbbm 1_A]
        \right)
        \left(
            \tilde\Phi_{\tau_r,s}^\dagger
            \circ\mathsf L_{\hat X_r}
            \circ\tilde{\mathcal G}_{r+1:n}^{\dagger}
            [\mathbbm 1_A]
        \right)
    \right]
    \right|_{s=s_{q,r}^-}^{s=s_{q,r}^+}.
\end{equation}

\section{Gauge dependency} \label{app:gauge}
We now illustrate and resolve an apparent paradox that arises in our framework when adding an arbitrary time-dependent term proportional to the identity to the system Hamiltonian. Take the time-dependent Hamiltonian acting on the system, $\hat{H}_S(t)$, and define
\begin{equation}\label{eq:mod_ham}
    \hat{H}_S'(t) = \hat{H}_S(t) + f(t)\,\mathbbm{1}_S.
\end{equation}
One can easily show that the modified Hamiltonian generates exactly the same dynamics on the system as the original one. In fact, the identity-proportional term ultimately reduces to a global phase, i.e.
\begin{equation}
    \hat U'(t,t_0)
    =\mathcal{T}\exp\!\left(-i\int_{t_0}^t \!dt'\, \hat H'_S(t')\right)
    =\mathcal{T}\exp\!\left(-i\int_{t_0}^t \!dt'\, \hat H_S(t')\right)\exp\!\left(-i\int_{t_0}^t \!dt'\, f(t') \right)
    = e^{i\phi}\,\hat U(t,t_0),
\end{equation}
where we have used that the identity term factors out, and identified $\phi=-\int_{t_0}^t \!dt'\, f(t')$. Consequently, the additional term cannot be detected by analyzing the system dynamics alone. This creates an apparent paradox, since the work operator corresponding to the new Hamiltonian yields a measurable difference:
\begin{equation}\label{eq:mod_work}
    \hat W'(t,t') = \hat U^\dagger(t,t')\,\hat H'_S(t)\,\hat U(t,t') - \hat H'_S(t') = \hat W(t,t') + \big(f(t)-f(t')\big)\,\mathbbm{1}_S.
\end{equation}
Such a discrepancy seemingly puts tension on this definition of work, since an additive term in the Hamiltonian that leaves the dynamics invariant should not have any physical meaning. One can identify the addition of the function $f(t)$ as a different gauge fixing of the model, leaving the dynamics invariant. Observing that the work operator depends on this function would then suggest that it is a gauge-dependent quantity and thus not completely operationally meaningful.

However, this inconsistency is also fully present in a classical context. This is clear because \eqref{eq:mod_work} retains the discrepancy for classical protocols. E.g. suppose a process where someone is raising and lowering some weight in a gravitational potential while walking up and down a hill. Here the position of the agent on the hill plays the role of $f(t)$, and indeed if the agent is climbing the hill while lifting the weight then they have to spend extra energy (beyond that to move themselves) in order to further raise the system too.

This inconsistency was already found in the classical context in Ref.~\cite{PhysRevLett.100.020601}, starting from the result of Ref.~\cite{Horowitz_2007}, challenging the connection between work and free energy. The apparent issue has been widely discussed in the community~\cite{imparato2007commentonfailureworkhamiltonian, PhysRevLett.101.098901, vilar2007inconsistencynonstandarddefinitionwork}, and ultimately settled~\cite{peliti2007commentoninconsistencynonstandard, Peliti_2008}. There, the author shows that this apparent paradox disappears when considering the physical setup needed to perform the manipulation described by the time-dependent Hamiltonian. We show here how our autonomization framework naturally provides an analogous solution to the paradox in the quantum case as well.

Let us construct the two Hamiltonians corresponding to the autonomous descriptions of $\hat H_S(t)$ and $\hat H'_S(t)$. From Eqs.~(\ref{eq:int_ham}--\ref{eq:clock_ham}) we have
\begin{align}
    \hat{H}_{\text{tot}} &= \hat H_C + \hat H_{\text{int}}
    = \frac{1}{2\pi}\int_{\mathbb R}\! dE\,dt\,dt'\; E \, e^{iE(t-t')}\,\ket{t\,}\!\bra{t'} 
      + \int_{\mathbb{R}}\! dt'\, \hat H_S(t')\otimes \ket{t'}\!\bra{t'},\\
    \hat{H}'_{\text{tot}} &= \hat H_C + \hat H'_{\text{int}}
    = \frac{1}{2\pi}\int_{\mathbb R}\! dE\,dt\,dt'\; E \, e^{iE(t-t')}\,\ket{t\,}\!\bra{t'} 
      + \int_{\mathbb{R}}\! dt'\, \hat H'_S(t')\otimes \ket{t'}\!\bra{t'}.
\end{align}
Plugging the explicit expression for the modified Hamiltonian in Eq.~\eqref{eq:mod_ham} we get
\begin{equation}
    \hat{H}'_{\text{tot}} = \frac{1}{2\pi}\int_{\mathbb R}\! dE\,dt\,dt'\; E \, e^{iE(t-t')}\,\ket{t\,}\!\bra{t'}
    + \int_{\mathbb{R}}\! dt'\, \big(\hat H_S(t') + f(t')\,\mathbbm{1}_S\big)\otimes \ket{t'}\!\bra{t'}.
\end{equation}
Since the added $f(t)$ in the interaction Hamiltonian acts trivially on the system, we can equivalently ascribe it to a modification of the clock Hamiltonian:
\begin{equation}
    \hat{H}'_{\text{tot}} = \hat H'_C + \hat H_{\text{int}}
    = \frac{1}{2\pi}\int_{\mathbb R}\! dE\,dt\,dt'\; \big(E \, e^{iE(t-t')} + f(t')\,\delta(t-t')\big)\ket{t\,}\!\bra{t'}
    + \int_{\mathbb{R}}\! dt'\, \hat H_S(t')\otimes \ket{t'}\!\bra{t'}.
\end{equation}
This rewrite shows explicitly how the dynamics of the system itself is unaffected by the added term. However, the modification in $\hat H'_C$ induces a different dynamics for the clock system, which explains the detectable difference in expended work. In other words, $\hat H_S(t)$ and $\hat H'_S(t)$ indeed represent different physical situations: they implement the same process on the system via different evolutions of the time-control system, thereby incurring in different work costs. This implies that the two settings compared above do not merely correspond to the same physics in different gauges, but rather represent distinct physical implementations, resulting in processes with different energetic costs depending on the evolution of the control system (here, the clock).

Finally, one can consider one last resolution of the apparent paradox that is agnostic to the physical scenario: cycles. By imposing that the initial and final Hamiltonian coincide, then we have full thermodynamic consistency a priori of identifying the physical scenario.

\end{document}